\documentclass[letterpaper]{article} 
\usepackage{aaai2026}  
\usepackage{times}  
\usepackage{helvet}  
\usepackage{courier}  
\usepackage[hyphens]{url}  
\usepackage{graphicx} 
\usepackage{natbib}  
\usepackage{caption} 
\usepackage[utf8]{inputenc} 
\usepackage[T1]{fontenc}    

\usepackage{url}            
\usepackage{booktabs}       
\usepackage{amsfonts}       

\usepackage{amsmath,amsthm, amssymb}
\usepackage{algorithm}
\usepackage{algpseudocode}
\usepackage[thinc]{esdiff}

\usepackage{caption}
\usepackage{subcaption} 
\nocopyright
\newcommand{\tool}{\textsc{GreedyMI}}
\newcommand{\erdosrenyi}{Erd\H{o}s-R\'{e}nyi}
\usepackage{newfloat}
\usepackage{listings}
\DeclareCaptionStyle{ruled}{labelfont=normalfont,labelsep=colon,strut=off} 
\floatstyle{ruled}
\newfloat{listing}{tb}{lst}{}
\floatname{listing}{Listing}
\newtheorem{theorem}{Theorem} 
\newtheorem{corollary}[theorem]{Corollary}
\newtheorem{lemma}{Lemma}

\newtheorem{observation}[theorem]{Observation}

\renewcommand{\tilde}{\widetilde}

\def\max{\qopname\relax n{max}}

\def\argmin{\qopname\relax n{argmin}}
\def\argmax{\qopname\relax n{argmax}}

\newcommand{\prob}{\textsc{TestPrev}}

\newcommand{\obj}{\mathcal{M}}
\newcommand{\p}{\lambda}

\title{Information Theoretic Optimal Surveillance for Epidemic Prevalence in Networks}

\begin{document}

\author {   
    Ritwick Mishra\textsuperscript{\rm 1,2},
    Abhijin Adiga\textsuperscript{\rm 1},
    Madhav Marathe\textsuperscript{\rm 1, 2},
    S. S. Ravi\textsuperscript{\rm 1},
    Ravi Tandon\textsuperscript{\rm 3},
    Anil Vullikanti\textsuperscript{\rm 1,2}
}
\affiliations {   
    \textsuperscript{\rm 1} Biocomplexity Institute, University of Virginia \\
    \textsuperscript{\rm 2} Department of Computer Science, University of Virginia \\
    \textsuperscript{\rm 3} Department of Electrical and Computer Engineering, University of Arizona\\
    \{mbc7bu, abhijin, marathe, ssravi,vsakumar\}@virginia.edu,
    tandonr@arizona.edu
}

\maketitle

\begin{abstract}
Estimating the true prevalence of an epidemic outbreak is a key public health problem. 
This is challenging because surveillance is usually resource intensive and biased.
In the network setting, prior work on cost sensitive disease surveillance has focused on choosing a subset of individuals (or nodes) to minimize objectives such as probability of outbreak detection.
Such methods do not give insights into the outbreak size distribution which, despite being complex and  multi-modal, is very useful in public health planning.

We introduce \prob{}, a problem of choosing a subset of nodes which maximizes the mutual information with disease prevalence, which directly provides information about the outbreak size distribution.
We show that, under the independent cascade~(IC) model, solutions computed by all prior disease surveillance approaches are highly sub-optimal for \prob{} in general.
We also show that \prob{} is hard to even approximate.
While this mutual information objective is computationally challenging
for general networks, we show that it can be computed efficiently for various network classes. 
We present a greedy strategy, called \tool{}, that uses estimates of mutual information from cascade simulations and thus can be applied on any network and disease model. 
We find that \tool{} does better than natural baselines in terms of maximizing the mutual information as well as reducing the expected variance in outbreak size, under the IC model.

\end{abstract}

\section{Introduction}
\label{sec:intro}

Effective surveillance  in the context of disease outbreaks and contagion processes (e.g., modeled as SIR type processes on networks~\cite{kempeMaximizing2003}) involves monitoring or testing a subset of individuals (referred to henceforth as \underline{nodes}) to determine characteristics of an outbreak; see e.g.,~\cite{leskovec2007cost, christakis2010social, shao2018forecasting, tsui2024toward, bai2017optimizing, heavey2022provable}. This is also referred to as the problem of choosing ``sensors''~\cite{leskovec2007cost}.
This is particularly important from a public health perspective because testing is very expensive, and resources are limited.
In networked models of disease spread, a lot of prior work on surveillance has been on choosing optimal sensor sets, so that monitoring them allows a good estimation of metrics such as probability of infections, delay in detecting the outbreak  (e.g.,~\cite{leskovec2007cost, heavey2022provable}) and time of peak (e.g.,~\cite{christakis2010social, shao2018forecasting}).
However, none of the prior methods provides information about the disease outbreak at a distribution level.

Information theoretic strategies are a natural way to capture distribution level information, and such approaches have been developed in other settings, such as  placing sensors for environmental monitoring.
Using a Gaussian process based modeling approach, Caselton 
and Zidek~(\citeyear{caselton1984Optimal})
proposed the use of 
\emph{mutual information}\footnote{Definitions of information theoretic concepts used in this paper can be
found in standard texts such as 
\citet{Cover-Thomas-1991,MacKay-2003}.}
as the optimization criterion.
Specifically, for a sensor set $A\subset V$,
their goal is to maximize the mutual information $I(X_A; X_{V\setminus A})$, where $X_S$ denotes the state vector for a subset $S\subset V$ of nodes.
Krause, Singh, and 
Guestrin~(\citeyear{krause2008Near}) showed that 
this objective  (referred to henceforth as the CZK objective) leads to
sensor placements that are most informative regarding locations without sensors, compared to other objectives.

Krause et al.~(\citeyear{krause2008Near}) showed that finding a placement of sensors that maximizes this objective is NP-hard; 
they also showed that 
$I(X_A; X_{V\setminus A})$ is a submodular function of $A$, and thus can be approximated well using a greedy strategy. 
Such an information theoretic perspective has not been studied for surveillance of disease outbreaks, and is our focus here.
Specifically, we study the following: how should we select a
subset $A$ of nodes on a network which gives the most information about the outbreak size distribution, denoted by $P(Z)$, 
from the results obtained by testing only the nodes in $A$?
This is particularly useful in the case of disease outbreaks because they often exhibit threshold effects, e.g.,~\cite{marathe:cacm13, pastor2015epidemic}, and a distribution level understanding gives better insights into the risk of having large outbreaks.
Prior surveillance methods only optimize specific epidemic metrics, such as the probability of detection or peak. Their solutions cannot provide such insights about the outbreak, as we formally show.

\smallskip 

\noindent
\textbf{Our contributions.}
1. We introduce a novel information theoretic criterion for optimizing surveillance for prevalence estimation.
Motivated by~\cite{caselton1984Optimal,krause2008Near}, our goal is to choose a subset $A$ which maximizes the mutual information $I(X_A; Z)$ with the prevalence (more generally, a weighted prevalence); we refer to this as the \prob{} problem.
We show that, in general, subsets that optimize other epidemic objectives 
(e.g., detection probability), do not maximize the mutual information; they can be, in fact, arbitrarily away from the optimal mutual information.\\
2. We show that \prob{} is NP-hard, even to approximate within a $\Theta(\log{n})$ factor, where $n$ is the number of nodes.
We also show that there are significant differences between the \prob{} and CZK objectives.
First, optimal solutions for the CZK objective can have arbitrarily low mutual information with the prevalence.
Second, the CZK objective is computationally much simpler--- it is submodular, and can be approximated within a $(1-1/e)$ factor by a greedy algorithm, while the \prob{} objective cannot be approximated to within a $\Theta(\log{n})$ factor, unless
\textbf{P} = \textbf{NP}.\\
\noindent
3. In general, computing the mutual information is computationally challenging since it involves an exponential sum.
We show that for trees and one-hop disease models, the value of mutual information can be computed efficiently.
For a simple path network, we derive a closed
form expression for the optimal solution to \prob{}.
For the general setting, we present a greedy strategy, called \tool{}, based on estimating the mutual information from cascade samples.\\
4. We evaluate our method through simulations over both synthetic networks and a real contact network between patients and providers in a hospital ICU, under various disease probability regimes. 
We compare the surveillance sets found by our method with natural baselines to show its efficacy in finding robust solutions.
We find that \tool{} outperforms the considered baselines and for budget as low as~2\%, gives solutions with more than 60\% reduction in variance about prevalence in many networks.\\
5. We analyze the solution sets in terms of their structural and dynamical properties. Our analysis reveals that \textsc{GreedyMI} achieves a favorable balance between relevance (i.e., high information content) and redundancy (i.e., low marginal information gain), outperforming baseline methods in this trade-off.\\

\noindent
For space reasons, proofs of some results
are omitted. They can be found in the technical supplement.

\section{Problem Definition}
\label{sec:prob_def}

\paragraph{Disease model. } We consider the simplest SIR-type disease model on a network $G=(V, E)$, the Independent Cascade (IC) model. All nodes in the network except the Infected seed~(s) start as Susceptible (S). 
At each time step, an Infected (I) node $u$ can activate an adjacent Susceptible node $v$ with probability $\p_{(u,v)}$, after which, it is Removed (R) from the process. The process continues until there are no more new infections. 
In a $d$-hop IC model, the spread is limited to a maximum of $d$-hops from the seed(s). 
Thus, the IC model is denoted by $IC(\boldsymbol{\p}, d)$, where the vector $\boldsymbol{\p}$ is composed of all the edge-wise disease probabilities $\p_{e}$ and $d$ is the number of hops; when there are no hop restrictions, we denote the model by $IC(\boldsymbol{\p}, \infty)$.
In the homogeneous setting, where~$\lambda_e=\lambda$, we use the notation~$IC(\lambda,d)$ to denote such a model.

We will use upper-case letters for random variables and lower-case letters for deterministic variables; for example, $X$ takes value $x$ in case of scalars or $\mathbf{x}$ for vectors.
We use a random variable $X_v=1$ to indicate that node $v\in V$ is infected.
The \emph{weighted disease prevalence} (sometimes referred to as prevalence 
for simplicity), $Z=\sum_{i\in V} w_i X_i$, is a random variable representing the weighted sum of infections, where  $w_i$ denotes a non-negative weight associated with node $i\in V$.
Let $Z_A$ = $\sum_{i\in A} w_iX_i$ denote the weighted sum of the state variables of nodes in $A$, while $Z_A^-$ = $\sum_{i\in V\setminus A} w_i X_i$
be the weighted sum of the state variables of the \emph{remaining} nodes in the network.

Let $h(p)$ denote the binary entropy function defined by $h(p) = -p\log p - (1-p)\log (1-p)$, 
$p\in [0,1]$.
All the logarithms in this paper use base 2. 
Note that $h(p)$ represents the entropy $H(Q)$ of a random
variable $Q$ which takes on one of two values with
probabilities $p$ and $1-p$ 
respectively~\cite{Cover-Thomas-1991}.
For two random variables $P$ and $Q$,
$H(P|Q)$ represents the conditional 
entropy of $P$ given $Q$.
We use the following lemma regarding
conditional entropy repeatedly (proof in the supplement).
\begin{lemma}
    \label{lem:remain-sum}
   $H(Z|X_A) = H(Z_A^-| X_A)$. 
When the variables are all independent, $H(Z|X_A) = H(Z_A^-)$.
\end{lemma}

\noindent
\textbf{The prevalence mutual information criterion.}
Given a set of nodes $A$, we define a function $\obj{}:2^V \rightarrow \mathbb{R}_{\geq 0}$ as the mutual information  between $X_A$ and prevalence $Z$. 
\begin{align}
    \obj{}(A) = I(X_A; Z) = \sum_{z =0}^n\sum_{x\in \mathcal{X}_A} p(x, z) \log \frac{p(x,z)}{p(x)p(z)}
    \label{eqn:mi}
\end{align}
Here, $p(x,z)$ is the joint probability mass function of $X_A$ and $Z$, while $p(x)$ and $p(z)$ are their respective marginal probabilities. $\mathcal{X}_A$ represents the support of $X_A$.
$\obj(\cdot)$ quantifies the effect of knowing the states of a set of nodes on the prevalence distribution, which serves as our optimization criterion. In a limited budget setting, we want to query a limited subset of nodes whose effect on the cascade size distribution is the greatest among all such sets. 

\noindent
\textbf{The \prob{} problem.} Given a network $G=(V,E)$,  disease parameters  $\boldsymbol{\p}, d$, weights $w_i$ and costs $c_i, i\in V$, and budget $k$, our goal is to find a set of nodes with the maximum mutual information, i.e., $A^* \in \argmax_{A\subset V, \sum_{i\in A}c_i \leq k} \obj{}(A)$. 
The weights can be used to model the relative importance of subgroups within the population.

Since $I(X_A; Z)$ = 
$H(Z) - H(Z|X_A)$~\cite{Cover-Thomas-1991},
maximizing $I(X_A; Z)$ is equivalent to minimizing $H(Z|X_A)$. 
In some of our results, we will consider uniform weights and costs, i.e., $w_i=1, c_i=1$ for all $i\in V$.

\noindent
\textbf{Caselton-Zidek-Krause (CZK) objective.} 
This objective  
denoted by $\mathcal{K}(A)$, is defined by 
$\mathcal{K}(A)$ = $I(X_A; X_{V\setminus A}) = H(X_{V\setminus A}) -H(X_{V\setminus A}| X_A)$.
It chooses a subset $A$ which maximizes the reduction in conditional entropy over the rest of the nodes.
This was first proposed by \cite{caselton1984Optimal} and studied extensively by \cite{krause2008Near} for spatial processes; it has not been used for disease spread.

\section{Related Work}
\label{sec:related}
There is considerable prior work on epidemic surveillance with diverse objectives such as to detect outbreaks~\cite{leskovec2007cost,bai2017optimizing}, determine outbreak characteristics~\cite{christakis2010social,shao2018forecasting}. Our surveillance problem is new and is generally open in the context of epidemic processes over networks.
As mentioned earlier, similar mutual information criteria have been used for problems such as observation selection~\cite{krause2007near,krause2012near}, sensor placement~\cite{krause2008Near,caselton1984Optimal} and feature selection~\cite{peng2005feature,brown2012conditional}. 
\cite{tsui2024toward} consider an active learning framework for node subset selection, where the test feedback is used to inform the next choice. 
The problem of placing sensors for detecting outbreaks has been considered in \cite{leskovec2007cost,adhikari2019fast,heavey2022provable}. Our goal instead is to find nodes which give the most information about the prevalence. 
We adapt the algorithm on trees in  \citet{burkholz2021cascade} to compute the conditional prevalence given a set of observations. 
The problem of estimating entropy has a long history \cite{paninski2003estimation} with more recent works~\cite{valiant2011estimating, wu2016minimax} focusing on advanced estimators to achieve a near-optimal asymptotic sample complexity. Here, we use a simple plug-in estimator for implementability and low computational overhead. \citet{wangFastApproximationEmpirical2019} speed-up the estimation of empirical entropy by taking a random subsample of the dataset. Here, our goal is to estimate the true entropy of the prevalence from samples.

\section{Analytical Results}

\label{sec:analytical}

\smallskip 

\noindent
\textbf{Complexity results for \prob{}.}~
In establishing our complexity results, we use the fact (from Section~\ref{sec:prob_def}) that
maximizing $I(X_A;Z)$ is equivalent to
minimizing the conditional entropy $H(Z|X_A)$.
Thus, we will consider the following 
(equivalent) version
of \prob{}:
given a directed contact network $G=(V,E)$, 
with a  non-negative weight $w_i$ and cost $c_i$
for each node $v_i \in V$, a probability $\lambda_e$ for each
directed edge $e \in E$, a non-negative weight bound $k \leq \sum_{i=1}^{n} c_i$ and a non-negative rational value $R$, is there a subset $A \subseteq V$ with  weight $\leq k$ such that $H(Z|X_A) \leq R$?
The following result establishes the
computational intractability of \prob{}.

\begin{theorem}\label{thm:gen_prob_hardness}
(a) \prob{} is NP-hard even when the
propagation is restricted to one hop.
(b) The problem remains
NP-hard even if the constraint on the weight of $A$ can be violated by a factor $(1-\epsilon)\log{n}$, for any $\epsilon < 1$,
where $n = |V|$.
\end{theorem}
\noindent
The proof is by a reduction from the Minimum Set Cover (MSC) problem (details in the supplement). 

\noindent
\textbf{\prob{} vs. the CZK objective.}
We show here that there is a significant difference \prob{} and the CZK objective, both in terms of the objective value (they can differ arbitrarily), and structure of optimal solutions (the \prob{} objective is not always submodular, unlike CZK).

\begin{observation}
\label{obs:CZK}
    There exist instances in which optimizing the CZK criterion $\mathcal{K}(A) = I(X_A; X_{V\setminus A})$ does not optimize for mutual information with prevalence $\obj(A) = I(X_A;Z)$.
\end{observation}

\noindent
We also show that $\mathcal{M}(\cdot)$ has a different structure from the CZK objective.

\begin{observation}\label{obs:ind-supmod}
  Given a set of nodes $V$ whose states are mutually independent random variables, the function $\obj{}(A) = I(Z_V;X_A)$ is supermodular in $A \subseteq V$.
\end{observation}

\begin{observation}
\label{obs:submodular}
    There exist instances in which the function $\obj$ is submodular.
\end{observation}

\noindent
\textbf{\prob{} vs. optimizing surveillance objectives.}
We show that prior work on optimizing epidemic objectives, e.g., detection probability, does not give good solutions to \prob{}.

\begin{observation}
\label{obs:detprob}
There exist instances in which node selection for maximizing detection likelihood can lead to solutions with objective value for \prob{} less than the optimal by $\Theta(n)$.
\end{observation}
\noindent
All these observations are discussed in detail in the supplement.

\section{Our Approach}
\label{sec:approach}

\noindent 
\textbf{Overview.}~
Since the \prob{} objective is computationally very challenging, we first study it for special classes of networks and disease models. We show that it can be computed efficiently for certain special networks and disease models.
Finally, we develop a sampling-based greedy strategy for finding solutions to \prob{} for general networks and disease models.

\subsection{1-Hop Disease Spread}

Consider a scenario where we are given a set of infectious individuals and we would like to know which among their immediate circle of contacts should be prioritized for testing. We model this as a directed bipartite network $G=(U\cup W, E)$ where $U$ is the set of infected nodes and $W$ is the set of susceptible nodes which are neighbors of~$U$. Figure \ref{fig:1hop-ex} has an example. The disease spreads for one time-step from $U$ to $W$. We are given disease probabilities $\boldsymbol{\p} = \{\lambda_{ji}~:~ (j,i) \in E\}$. The probability $p_i$ of a node $i \in W$ becoming infected is given by $p_i = 1 - \prod_{(j,i)\in E} ( 1 - \lambda_{ji})$. The goal of the \prob{} problem is to find a subset of nodes $A$ which maximizes $I(Z;X_A) = H(Z) - H(Z|X_A)$, which is equivalent to minimizing $H(Z|A)$.  Due to independence of $\{X_i\}_{i\in W}$, $H(Z|X_A) = H(Z_A^-)$ by Lemma \ref{lem:remain-sum}. Thus, choosing the best set of nodes minimizes the entropy of the sum of the remaining node states in $W$.

\begin{figure}
    \centering
    \includegraphics[width=0.33\linewidth]{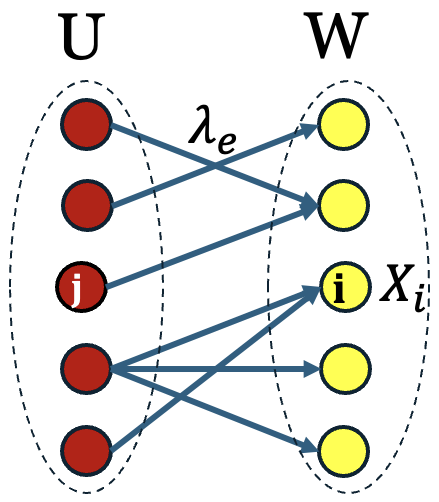}
    \caption{1-hop example on a bipartite network}
    \label{fig:1hop-ex}
\end{figure}

\begin{lemma}
\label{lem:onehop}
For the 1-hop disease spread scenario, the conditional entropy~$H(Z| X_A)$ can be computed exactly in polynomial time.    
\end{lemma}
\begin{proof} 
For any subset $A\subseteq W$, $Z_{A}^-$ follows a Poisson binomial distribution with parameters $\{p_i\}_{i\in W\setminus A} $.
We can compute the probability distribution of $Z_{A}^-$, using Direct Convolution method \cite{biscarri2018simple} in $O(|W|^2)$ time. This allows efficient computation of the entropy $H(Z_A^-)=H(Z|X_A)$.
\end{proof}

    \paragraph{Greedy heuristic.} 
    Due to all the node states being independent, $\obj(A) = I(Z; X_A)$ is supermodular from Observation~\ref{obs:ind-supmod}. As $I(Z;X_A)=H(Z) - H(Z|X_A)$ and $H(Z)$ is constant, $H(Z|X_A) = H(Z_A^-)$ is submodular in $A$. Here, \prob{} is minimizing a submodular function with cardinality constraints, which is generally NP-hard. Thus, in Algorithm 
    \ref{alg:1hop_greedy} 
    (in the supplement) we have a greedy heuristic which uses the exact computations of $H(Z|X_{A\cup \{i\}}) = H(Z_{A\cup \{i\}}^-)$ (based on Lemma~\ref{lem:onehop}) for each node $i\in W$ in the greedy step. 
\paragraph{Complexity. } The greedy heuristic takes $O(k|W|^3)$ under the assumption that the greedy step of computing $H(Z|X_{A\cup \{i\}})$ for a node $i$ takes $O(|W|^2)$ time.

\subsection{Rooted Trees}
Consider a tree network $T_r$ of size $n$ with a single source at the root $r$, under $IC(\boldsymbol{\p},\infty)$. We want to find a subset $A\subseteq V$ such that $\obj(A)$ is maximized, which is equivalent to minimizing $H(Z|X_A)$. 

\paragraph{Computing $H(Z|X_A)$.} Given a node
set $A$,
\textsc{EntropyOnTree} (Algorithm~\ref{alg:entr_tree} in the supplement) provides an exact method to compute the conditional entropy $H(Z|X_A)$ of the prevalence on a rooted tree with the root being the source of infection. It uses the following subroutines (whose pseudocodes appear in the supplement):\\
1. \textsc{Feasible}($T_r,\mathbf{x})$: filters out zero-probability infection status vectors~$\mathbf{x}$.\\
2. \textsc{Contract}($T_r, \Gamma_{rv}$): contracts the path from $r$ to $v$.\\
3. \textsc{Remove}($T_r, v$): removes the node $v$ from the tree.\\
4. \textsc{MessagePassing}($\tilde{T}_r, \boldsymbol{\p}, v$): computes the unconditional prevalence distribution, $P(Z_{\tilde{T}_r})$.
\begin{lemma}
\label{lem:tree}
    For any subset $A\subseteq T_r$,~\textsc{EntropyOnTree} exactly computes the conditional entropy of the prevalence $H(Z|X_A)$.
\end{lemma}
\noindent
\textbf{Proof sketch.}
We compute $P(X_A = \mathbf{x})$ for each $\mathbf{x}\in \mathcal{X}_A$ using the probability of the live-edge paths from the (un)-infected nodes. To compute $H(Z|X_A = \mathbf{x})$, we contract the live-edge paths and remove uninfected nodes to reduce the problem to computing unconditional prevalence distributions $P(Z_{\tilde{T}})$. This is done using a message-passing algorithm from \cite{burkholz2021cascade}. Finally, we compute $H(Z|X_A) = \sum_{\mathbf{x}\in \mathcal{X}_A} P(X_A=\mathbf{x})H(Z|X_A=\mathbf{x})$. More details are in the supplement.

\paragraph{Greedy heuristic.}
We use a greedy heuristic in Algorithm~\ref{alg:tree_greedy}  (in the supplement) which repeatedly adds nodes to the query set based on the exact computations of the conditional entropy of the prevalence using~\textsc{EntropyOnTree}.

\paragraph{Complexity.} \textsc{MessagePassing} has a complexity of $O(n^2)$ where $n$ is the size of the tree. It is called a maximum of $2^{k}$ times by \textsc{EntropyOnTree}. The greedy heuristic calls \textsc{EntropyOnTree} $O(kn)$ times, resulting in a complexity of $O(k2^{k}n^3)$. The running time is polynomial when $k$ is fixed.

\subsection{Path Networks}

We are given a path network of size $n+1$ with the node set being $V = \{0, 1, \ldots, n\}$. Assume a single source at node $0$. The \prob{} problem is to find a subset of nodes $\{i_j; i_j \in V\}$ which maximize the prevalence mutual information $\obj$ with budget $k$. Equivalently, we can find the sequence of optimal separations, $\{g_1 = i_1, g_2=i_2-i_1,\dots, g_k=i_k - i_{k-1}\}$, i.e., $\{g_j=i_j-i_{j-1}; j=1,\dots, k, g_j \geq 1\}$ where $i_0$ is defined to be $0$. Consider the version of this problem where we remove the integrality constraint on the variables, i.e., we allow $g_j \in \mathbb{R}_{\geq 0}$. With this relaxation, we can obtain a closed form solution.

\begin{theorem}
\label{lem:path}
For budget $k$, a sufficiently long path, i.e.,  $n > -\log (k+1)/\log \p$, and $IC(\p,\infty)$ with homogeneous disease probability $\p\in (0,1)$, the \prob{}-optimal separation  without integrality constraints is $\{g_j = \log (\frac{k+1-j}{k+2-j}) / \log \lambda; j=1,\dots, k\}$. 
\end{theorem}
\paragraph{Proof sketch.}
$\obj{\{X_{i_j}\}} = I(Z;\{X_{i_j}\}) = H(\{X_{i_j}\}) - H(\{X_{i_j}\}|Z) = H(\{X_{i_j}\})$ because once prevalence is given, all node states become deterministic in this path setting.~$H(\{X_{i_j}\})=h(\p^{g_1})+\p^{g_1}h(\p^{g_2})+ \p^{g_1+g_2}h(\p^{g_3})+\dots+\p^{\sum_{m=1}^{k-1} g_m}h(\p^{g_k})$ by chain rule, where $h$ is the binary entropy function. Using induction, we find the optimal separations in reverse-order via the first derivative test. The details are in the supplement.

\noindent 
\noindent
\textbf{Remark:}
   Since the optimal separations $\{g_j\}$ in Lemma~\ref{lem:path} may be fractional, our proposed simple heuristic in
   Algorithm~\ref{alg:path}
   (in the supplement) checks the nearby integral solutions. It has a time complexity of $O(k)$.

\subsection{Sampling-Based Methods}

On general networks, we rely on sampling-based techniques for estimating the mutual information function $\obj$ for any given query-set $A$.
The sampling method first constructs a dataset $D$ from large number of i.i.d. samples from the cascade distribution $IC(\boldsymbol{\p},d)$. Then we estimate the conditional entropy of the prevalence from the observations in $D$, ignoring the probability of any infection vectors and sizes  not seen in the dataset. The method is described in Algorithm \ref{alg:emp_entr}.

\paragraph{Greedy heuristic.} We provide a greedy heuristic which sequentially adds nodes to the query-set, maximizing the information gain in each step. Equivalently, the greedy step adds a node $v'$ to current query-set $A$ such that, $ v' \in \argmin_{v \in V\setminus A} H(Z|X_{A \cup\{v\}})$.
\textsc{GreedyMI} is described in Algorithm~\ref{alg:greedyquery} and has the empirical entropy subroutine in Algorithm~\ref{alg:emp_entr}.

\begin{algorithm}[tb]
\caption{\tool{}}
    \label{alg:greedyquery}
    \textbf{Input:} A contact network $G=(V,E)$, disease parameters $IC(\boldsymbol{\p}, d)$, budget $k$, number of cascade samples $T$.\\
    \textbf{Output:} Query-set $A\subseteq V$.
    \begin{algorithmic}[1]
        \State $A \leftarrow \phi$
        \State Sample $T$ i.i.d. disease cascades from $IC(\boldsymbol{\p}, d$) 
        \State Construct a matrix $D$ such that,
        $D_{i,j} = 1 $ if $j$-th node is infected in the $i$-th cascade sample, $0$ otherwise.
        \For{ $j$ = 1 to $k$}
            \For{each $v \in V\setminus A$}
                \State $\delta_v \leftarrow $ \Call{EmpiricalEntropy}{$D, A\cup \{v\}$}\
            \EndFor
            \State $v^* \leftarrow \argmin_{v\in V\setminus A} \delta_v$ 
            \State $A \leftarrow A \cup \{v^*\}$
        \EndFor
        \State\textbf{return} {$A$}
    \end{algorithmic}
\end{algorithm}

\begin{algorithm}[tb]
\textbf{Input:} Dataset of cascade samples $D$, subset of nodes $A$\\
\textbf{Output:} Empirical conditional entropy of prevalence $H_D(Z|A)$
\caption{\textsc{EmpiricalEntropy}}
\label{alg:emp_entr}
    \begin{algorithmic}[1]
        \State For each $\mathbf{x} \in \mathcal{X}_A$, compute the empirical probability $P_D(X_A = \mathbf{x})$,
        empirical conditional prevalence, $P_D(Z| X_A = \mathbf{x})$, thereby its entropy $H_D(Z|X_A = \mathbf{x})$.
        \State Compute the empirical conditional entropy, $H_D(Z|A) = \sum_{\mathbf{x} \in \mathcal{X}_A} P_D(X_A = \mathbf{x}) H_D(Z|X_A = \mathbf{x})$
        \State \textbf{return}  $H_D(Z|A)$
    \end{algorithmic}
\end{algorithm}

\paragraph{Complexity analysis.} The subroutine \ref{alg:emp_entr} can be implemented with a hashing-based grouping of the dataset $D$.  On average, the time complexity of \textsc{EmpiricalEntropy} is $O(kn)$. The time complexity of sampling from $IC(\boldsymbol{\p}, d)$ is $O(n+ m)$ where $m$ is the number of edges. This leads to a complexity of $O(T(n+m) + k^2 n^2)$ for \textsc{GreedyMI}, where $T$ denotes the number of cascade samples in Algorithm \ref{alg:greedyquery}.

\paragraph{Sample complexity.} The number of samples required for consistent estimation depends on the joint alphabet size of $P(Z,X_A)$, which can be as large as $n2^k$, but is usually much smaller in practice due to infeasible state configurations. With our plug-in estimator, we have a sample complexity, $O(n2^k/\epsilon^2)$ where $\epsilon$ is the desired additive error.

\section{Experiments}

We conduct extensive experiments using disease simulations on various network topologies, seeding scenarios, and disease transmission probability regimes to investigate the following research questions:\\
1. How does \textsc{GreedyMI} perform compared to baseline methods in identifying good node sets for disease surveillance?\\
2. What are the distinguishing topological and epidemiological properties of surveillance nodes selected by \textsc{GreedyMI} versus baseline methods?\\
3. What is the minimum number of cascade samples required for \textsc{TestPrev} to converge to a stable solution?

\newcommand{\pl}{\texttt{PowLaw}}
\newcommand{\er}{\texttt{ER}}
\newcommand{\icu}{\texttt{HospICU}}
\newcommand{\ksource}{\texttt{Known-source}}
\newcommand{\rsource}{\texttt{Random-source}}

\paragraph{Datasets and Methods. }
We use both synthetic and real-world networks. See Table~\ref{tab:datasets} for a summary.
\begin{enumerate}
    \item \pl: We construct several power-law networks using the Chung-Lu random graph 
    model~\cite{chung2002average}. The power-law exponent is set to $\gamma=2.5$. 
\item \er: We generate several \erdosrenyi{} networks using the $G(n,q)$ model with 
$n=1000, q=0.05$.
\item \icu: This is a contact network based on the co-location of patients and healthcare providers in the ICU of a large hospital (name withheld for anonymity), built using Electronic Health Records (EHR) collected between Jan 1, 2018 and Jan 8, 2018.
This network is quite relevant to the considered surveillance problem in the context of hospital acquired infections~\cite{heavey2022provable,jangRisk2021}.
\end{enumerate}
In the case of synthetic networks, we use~10 replicates for each graph family. Also, 
in each case, we use the largest connected component.

\begin{table}[tb]

\begin{center}
\begin{footnotesize}
    \begin{tabular}{|p{.55in}|p{0.34in}|p{0.42in}|p{0.33in}|p{0.72in}|}
    \hline
    \textbf{Network}
    & \textbf{Nodes} 
    & \textbf{Edges}
    & \textbf{Clust. coeff.} & 
    \textbf{Avg. Shortest Path Length}\\
    \hline \hline
    \pl & 675.3 & 1118.8 & 0.052 & 4.1\\
    \hline
    \er & 1000 & 24912.0 & 0.049 & 2.03\\
    \hline
    \icu & 879 & 3575 & 0.599 & 4.31\\
    \hline
    \end{tabular}
\end{footnotesize}
\end{center}
\caption{Networks and their properties. For the synthetic networks, average values are
reported across~10 replicates.}
\label{tab:datasets}
\end{table}

\paragraph{Disease scenarios.} We evaluate our methods on a range of model parameters
for $IC(\p,d)$. In each case, we assume a homogeneous disease probability setting with 
$\p \in \{0.1, 0.2\}$ for \pl{} and \icu, and $\p\in\{0.05,0.07\}$ for \er. We set the 
maximum number of hops~$d \in \{2,4\}$. We evaluate our method for a budget up to $k=10$.
These regimes are chosen so that the $d$-hop prevalence variance is high enough for surveillance to have an effect. Each simulation 
instance is initiated with a single seed node. We consider two seeding scenarios,\\
1. \ksource: The seed node is fixed for all cascades and chosen randomly from among the nodes. We consider 
10 replicates. Accordingly, we have 10 sets of cascades. \\
2. \rsource: For each cascade, the seed node is picked uniformly at random. 
In our experiments, we sample~30,000 cascades for each disease scenario.

\paragraph{Baselines. } We compare our method against these baselines, which have been used in prior work on surveillance based on epidemic metrics~\cite{leskovec2007cost, shao2018forecasting, marathe:cacm13}:
\begin{enumerate}
\item \textsc{Degree}: The top-$k$ nodes by degree form the surveillance set.
\item \textsc{Vulnerable}: In the set of sampled cascades used to compute the conditional entropy, we find the top-$k$ most frequently infected nodes.
This quantity is indirectly tied to $P(X_A)$ in the prevalence mutual information
expression in Equation~\ref{eqn:mi}.
\end{enumerate}

\paragraph{Evaluation metrics. } For evaluating performance, we use two metrics:\\
1. Prevalence mutual information: Given a set of nodes $A$, $\obj({A}) = I(Z;X_A)$. If this is high, the chosen node subset has high mutual information with the prevalence.\\
2. Expected standard deviation of the conditional prevalence: If $A$ is the subset of nodes selected by a method, this is computed as
     $E[\sigma(Z|X_A)] = \sum_{\mathbf{x} \in \mathcal{X}_A} P(X_A = \mathbf{x})\sigma(Z|X_A = \mathbf{x})$,
where $\sigma^2(Z|X_A = \mathbf{x})$ is the variance of~$Z$ conditioned on~$X_A=
\mathbf{x}$.
This measures the expected spread around the mean of the prevalence upon querying a node
set.

\subsection{Results}
\paragraph{Performance of \textsc{GreedyMI} versus baselines.}
Figure~\ref{fig:perf_greedy_fixed} shows the performance of \textsc{GreedyMI} against  
baselines under the \ksource{} seeding. We report the average scores over~10 replicates. 
We observe that \textsc{GreedyMI} consistently 
outperforms baselines, with the performance gap widening with increasing budget. We can observe a reduction in the expected standard deviation in prevalence ranging from  5\% in \er{} to 80\% in \icu.
We also 
observe that this expected standard deviation rapidly decreases with the 
first few node selections, followed by a more gradual decrease. This ``diminishing 
returns'' effect is especially pronounced in \pl{} and \icu{} networks while it is nearly
absent in \er{} networks. This is likely due to the shape of the degree distribution in 
each network. On \pl, \textsc{Degree} is nearly as good as \textsc{GreedyMI}, which is
linked to the structure of the network where there is a core of high-degree nodes. On
\icu, we observe that \textsc{Vulnerable} is superior to \textsc{Degree}, indicating that
dynamics-based selection can outperform structure-based selection.
In Figure \ref{fig:perf_greedy_random}, we show the performance of the methods under 
\rsource{} seeding. For \pl{} and \er, we average the scores over 10 network replicates.
We observe that \textsc{GreedyMI} still performs better than the baselines, although the 
gap is smaller than before.  In the \ksource{} scenario, the source node may have low-degree neighbors--nodes that are highly vulnerable but contribute little to the overall cascade size. This structural limitation can reduce the informativeness of such nodes. Such situations rarely arise in the \rsource{} scenario, where sources are selected randomly.

\begin{figure}[htb]
    \centering
\begin{subfigure}{\columnwidth}
   \includegraphics[width=0.46\linewidth]{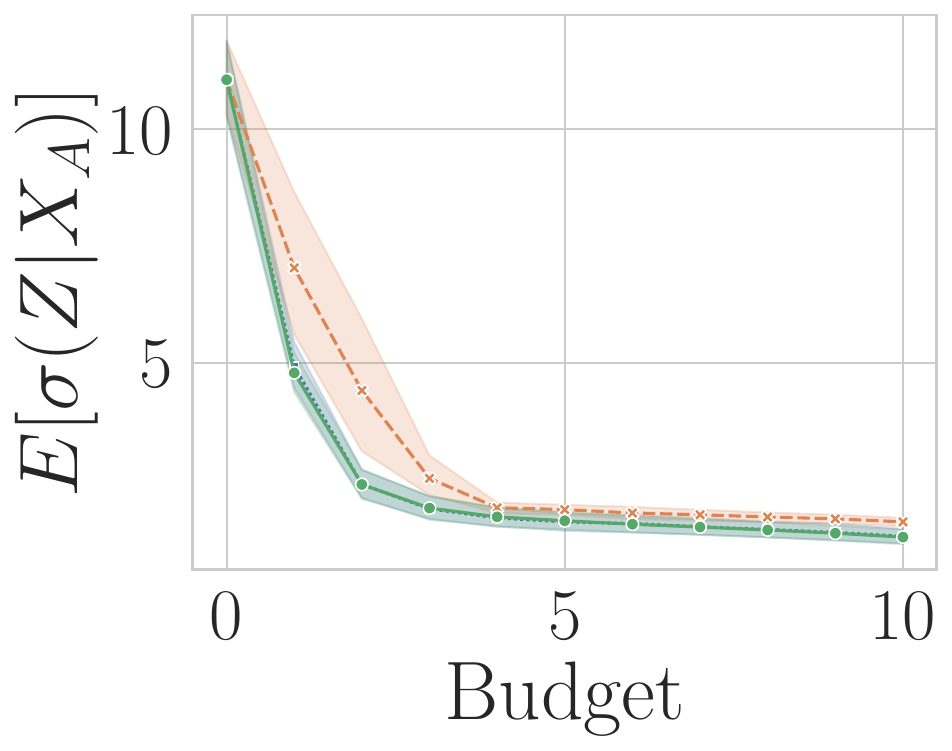}
    \includegraphics[width=0.46\linewidth]{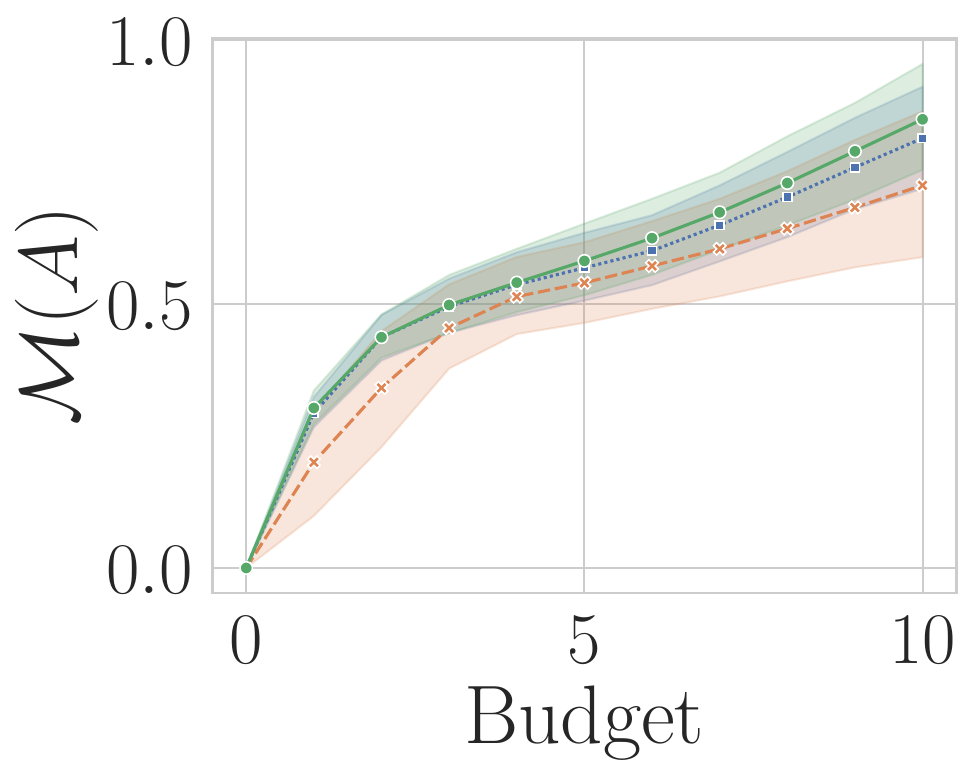}
    \caption{\pl~$IC(\p=0.1,d=4)$}
\end{subfigure}
\begin{subfigure}{\columnwidth}
    \includegraphics[width=0.46\linewidth]{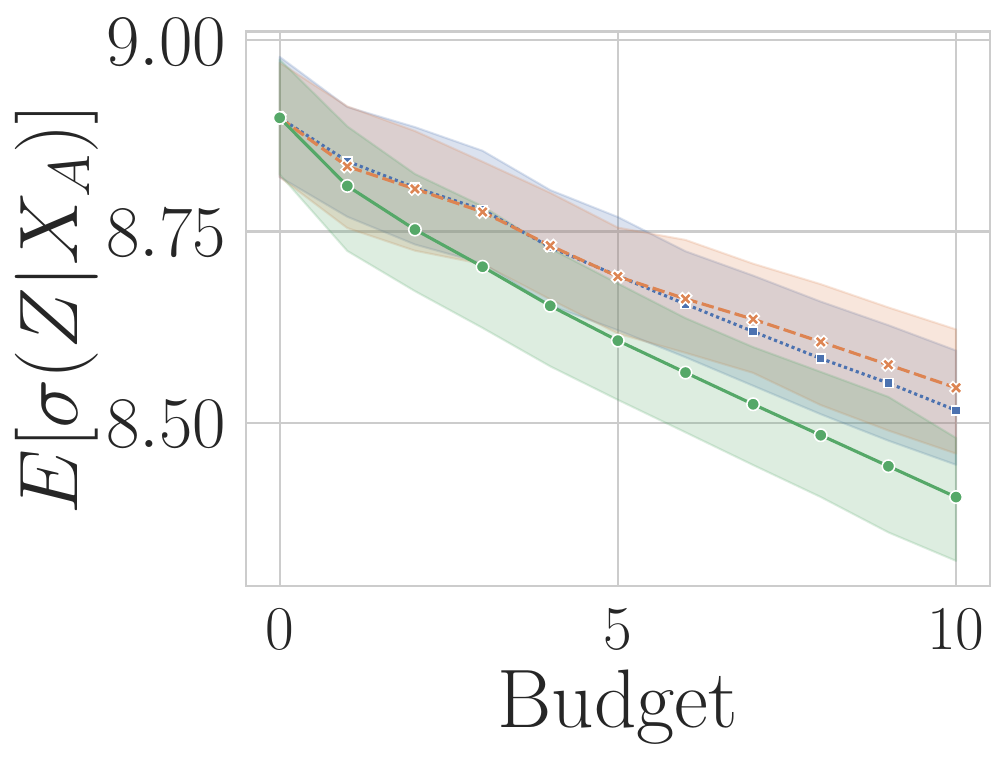}
    \includegraphics[width=0.46\linewidth]{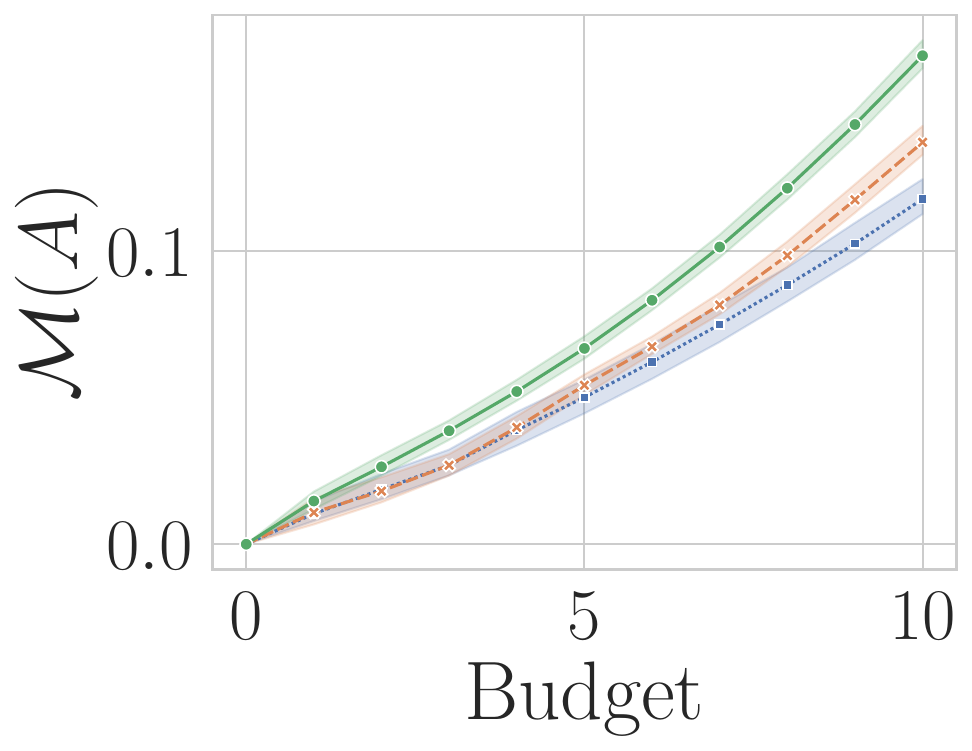}
    \caption{\er~$IC(\p=0.07,d=2)$}
\end{subfigure}
\begin{subfigure}{\columnwidth}
     \includegraphics[width=0.46\linewidth]{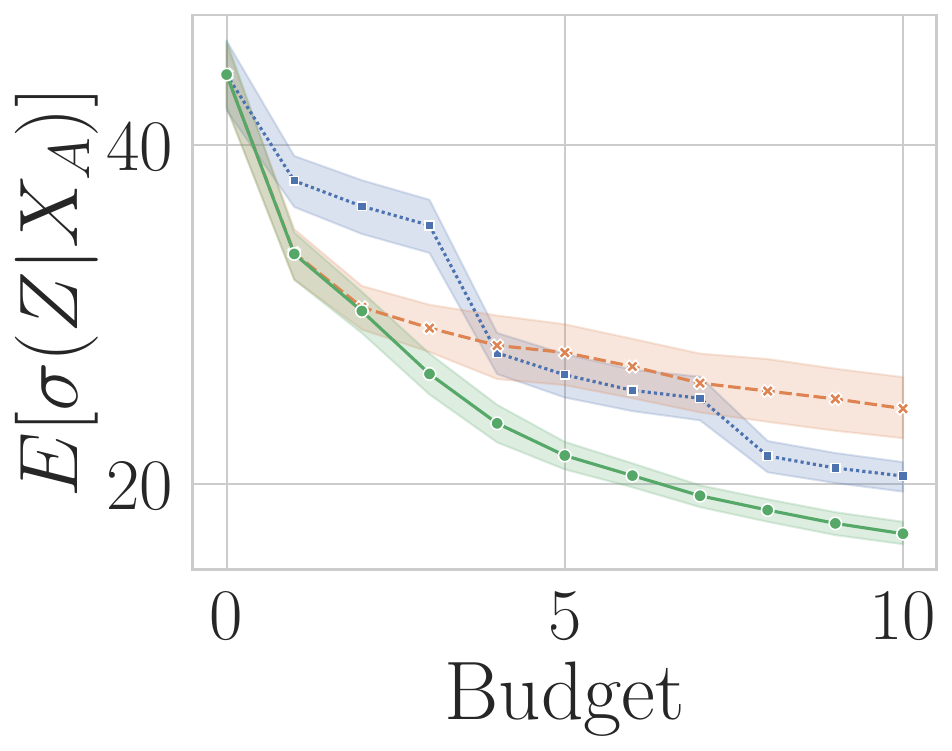}
     \includegraphics[width=0.46\linewidth]{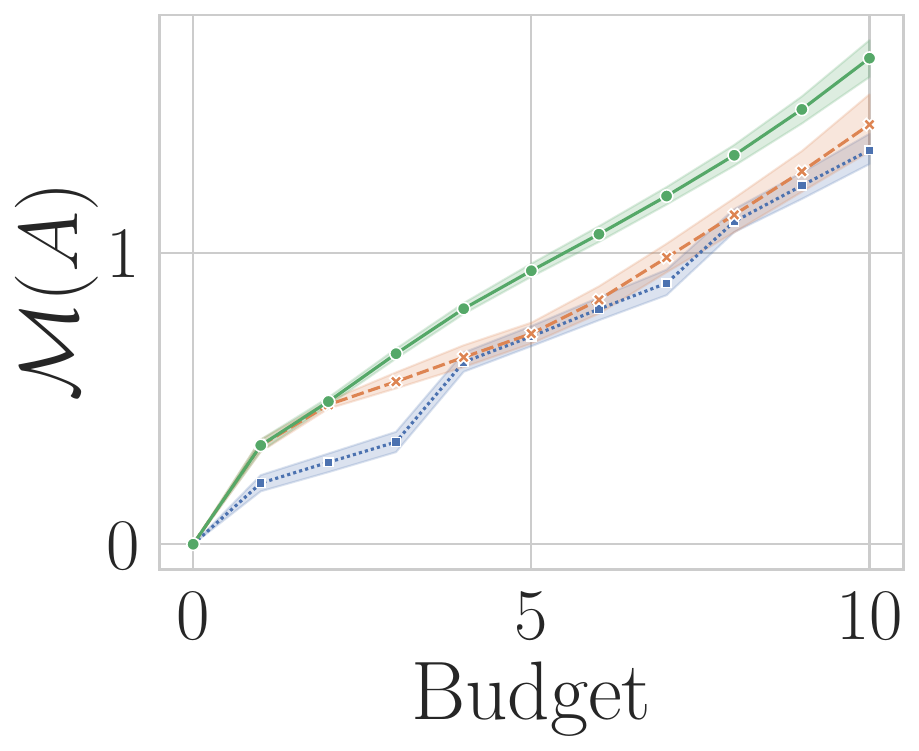}
    \caption{\icu~$IC(\p=0.2,d=4)$}
\end{subfigure}
    \includegraphics[width=0.95\linewidth]{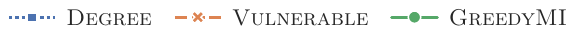}
\caption{Performance of \textsc{GreedyMI} vs baselines under \ksource~seeding (averaged over 10 runs).}
    \label{fig:perf_greedy_fixed}
\end{figure}

\begin{figure}[htb]
    \centering
    \begin{subfigure}{\columnwidth}
    \includegraphics[width=0.46\linewidth]{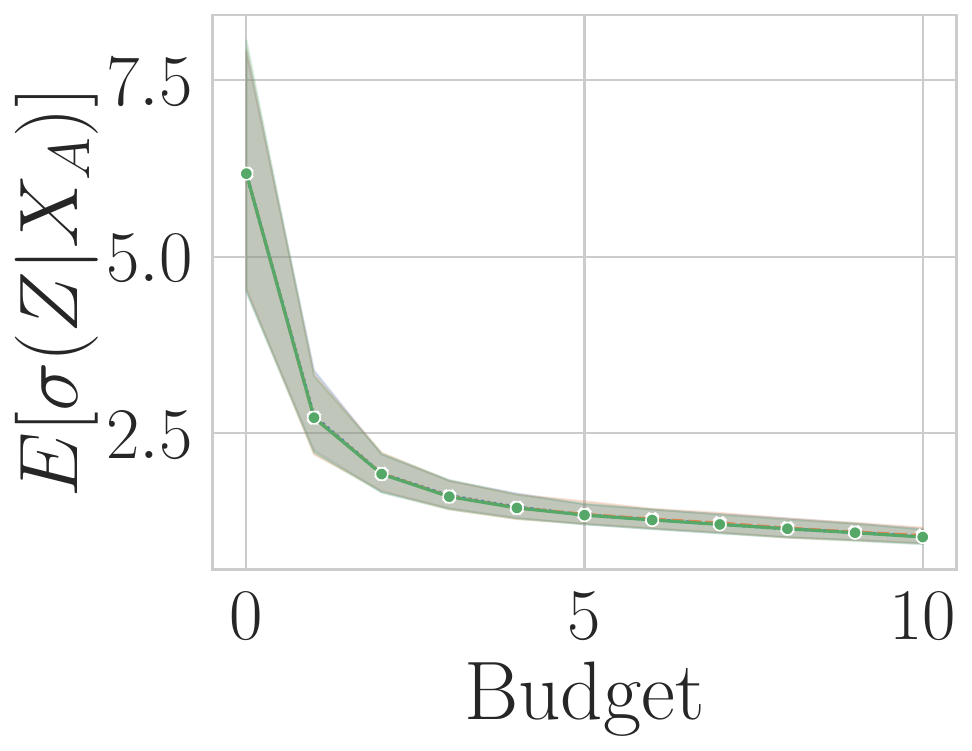}
    \includegraphics[width=0.46\linewidth]{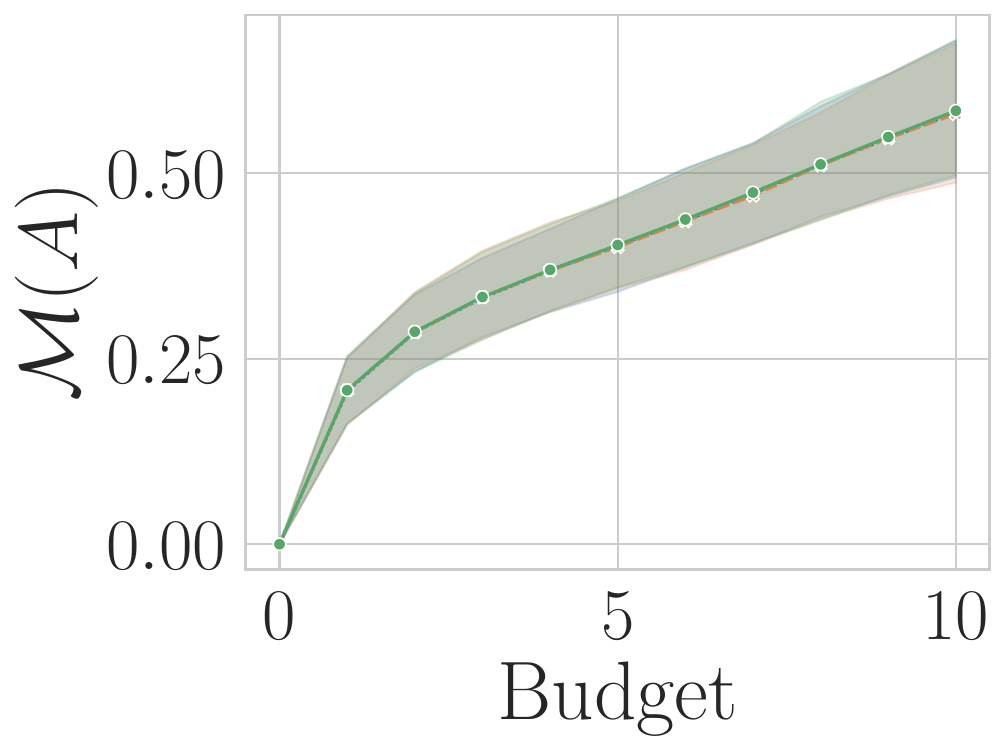}
    \caption{\pl~$IC(\p=0.1,d=4)$}
    \end{subfigure}
    \begin{subfigure}{\columnwidth}
    \includegraphics[width=0.46\linewidth]{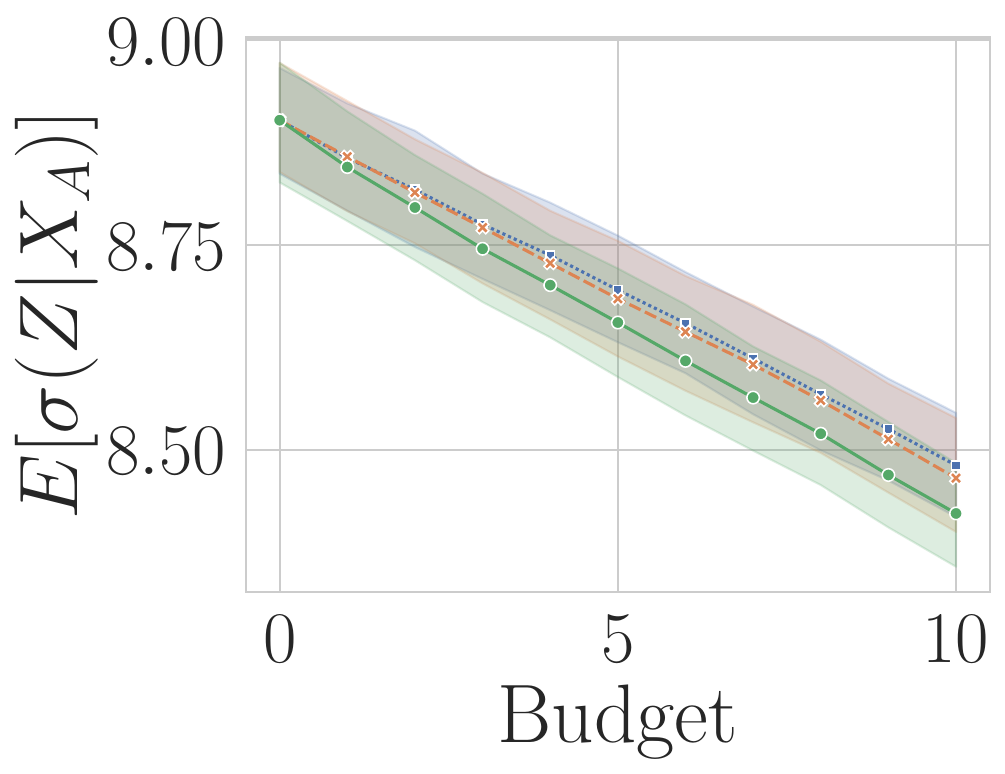}
    \includegraphics[width=0.46\linewidth]{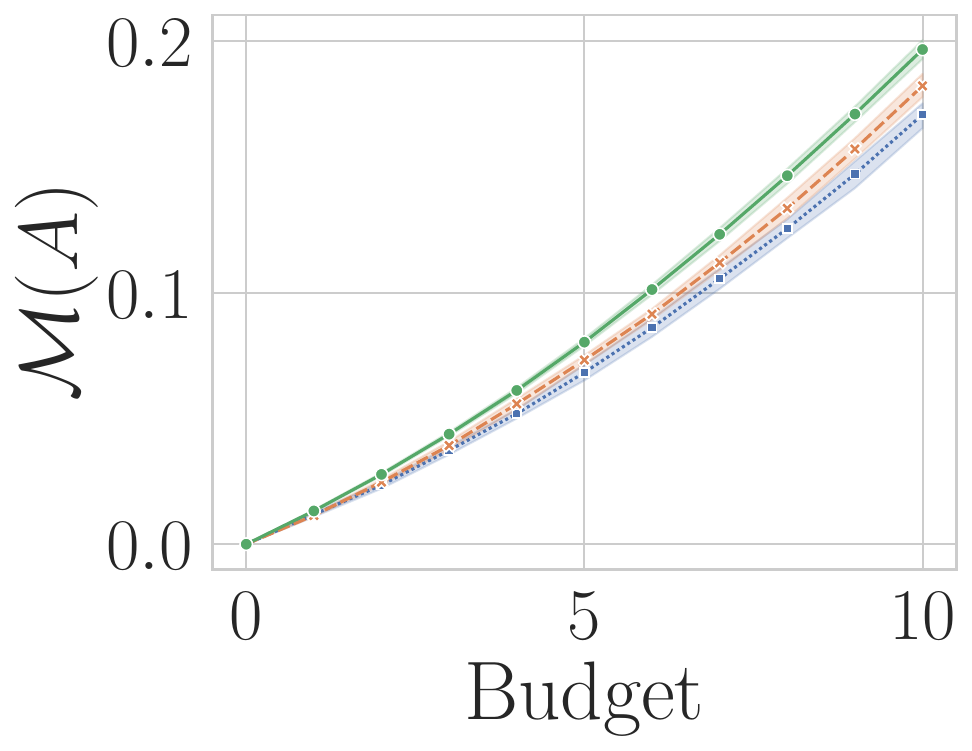}
    \caption{\er~$IC(\p=0.07,d=2)$}
    \end{subfigure}
    \begin{subfigure}{\columnwidth}
    \includegraphics[width=0.46\linewidth]{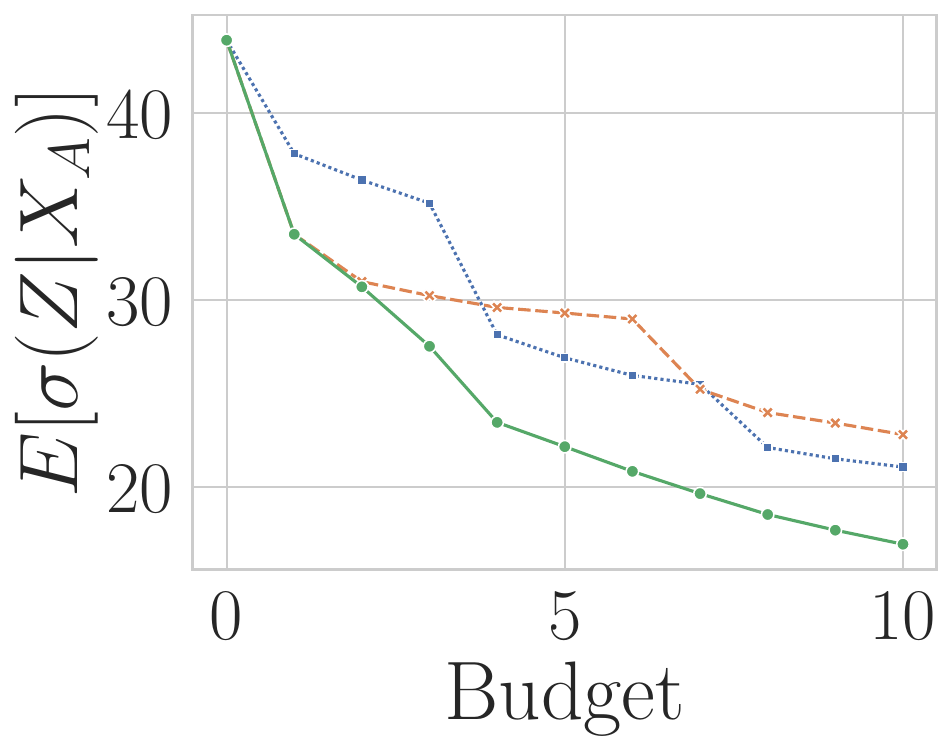}
    \includegraphics[width=0.46\linewidth]{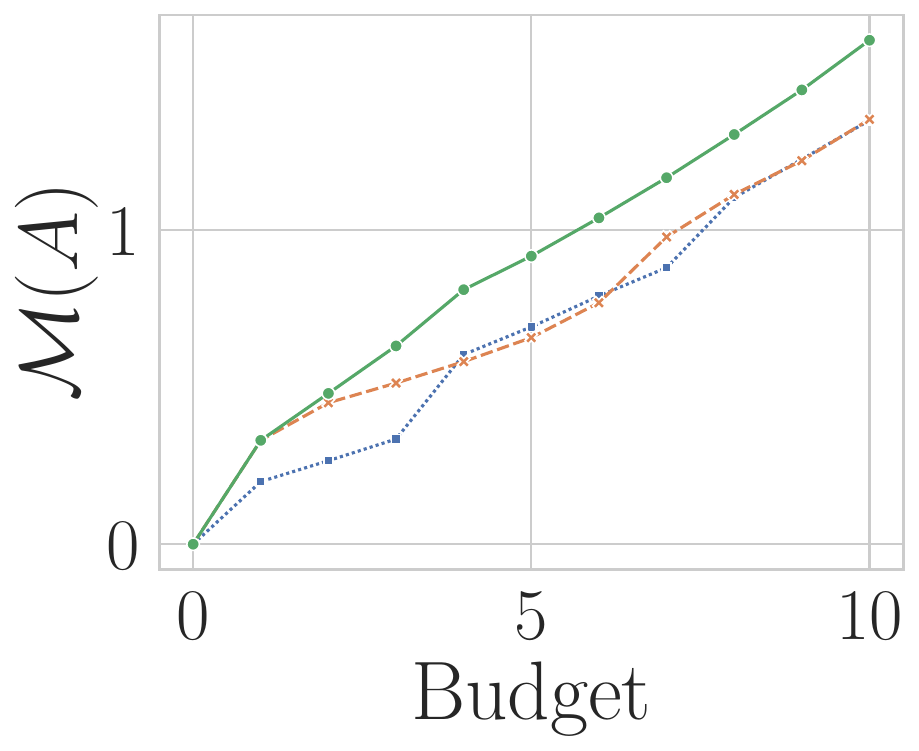}
    \caption{\icu~$IC(\p=0.2,d=4)$} 
    \includegraphics[width=0.95\linewidth]{methods_legend-2.pdf}
    \end{subfigure}
\caption{Performance of \textsc{GreedyMI} vs baselines under \rsource~seeding.}
    \label{fig:perf_greedy_random}
\end{figure}

\paragraph{Analysis of solutions.}
We analyze the solutions obtained by \textsc{GreedyMI} and the baselines to 
derive further insights into what characteristics make for a good surveillance node set. 
Specifically, the properties we consider for each node are the degree, vulnerability, 
and node-wise influence~(expected size of a cascade initiated from the node).
Figure~\ref{fig:degree_dist_random} shows the degrees of the selected nodes for different approaches along 
with those of the remaining nodes. 
Figure \ref{fig:inf_vul_scatter} shows the node-wise influence and vulnerability of the subset of nodes selected by each of the methods, with a random sample of network nodes forming the backdrop. 
We observe that although \textsc{Degree} nodes tend to have high influence, they are not
usually picked by \textsc{GreedyMI}, which has more overlap with \textsc{Vulnerable} node
sets. Any node selection algorithm must balance \emph{relevance} (i.e., information about prevalence) with \emph{redundancy} (i.e., low marginal information gain). 
Relevance depends on how vulnerable a node is and if infected, to what extent it can influence the cascade size. Very high or very low vulnerability, or low influence, all correspond to less information. 
Similarly, when states of two nodes are highly correlated, it makes sense to monitor one of the nodes, thus reducing redundancy.
Given the corresponding MI performance gaps, this suggests that \textsc{GreedyMI} balances this tradeoff better than the top-$k$ methods.

\begin{figure}[htb]
    \centering
   \begin{subfigure}{.32\linewidth} 
    \includegraphics[width=\linewidth]{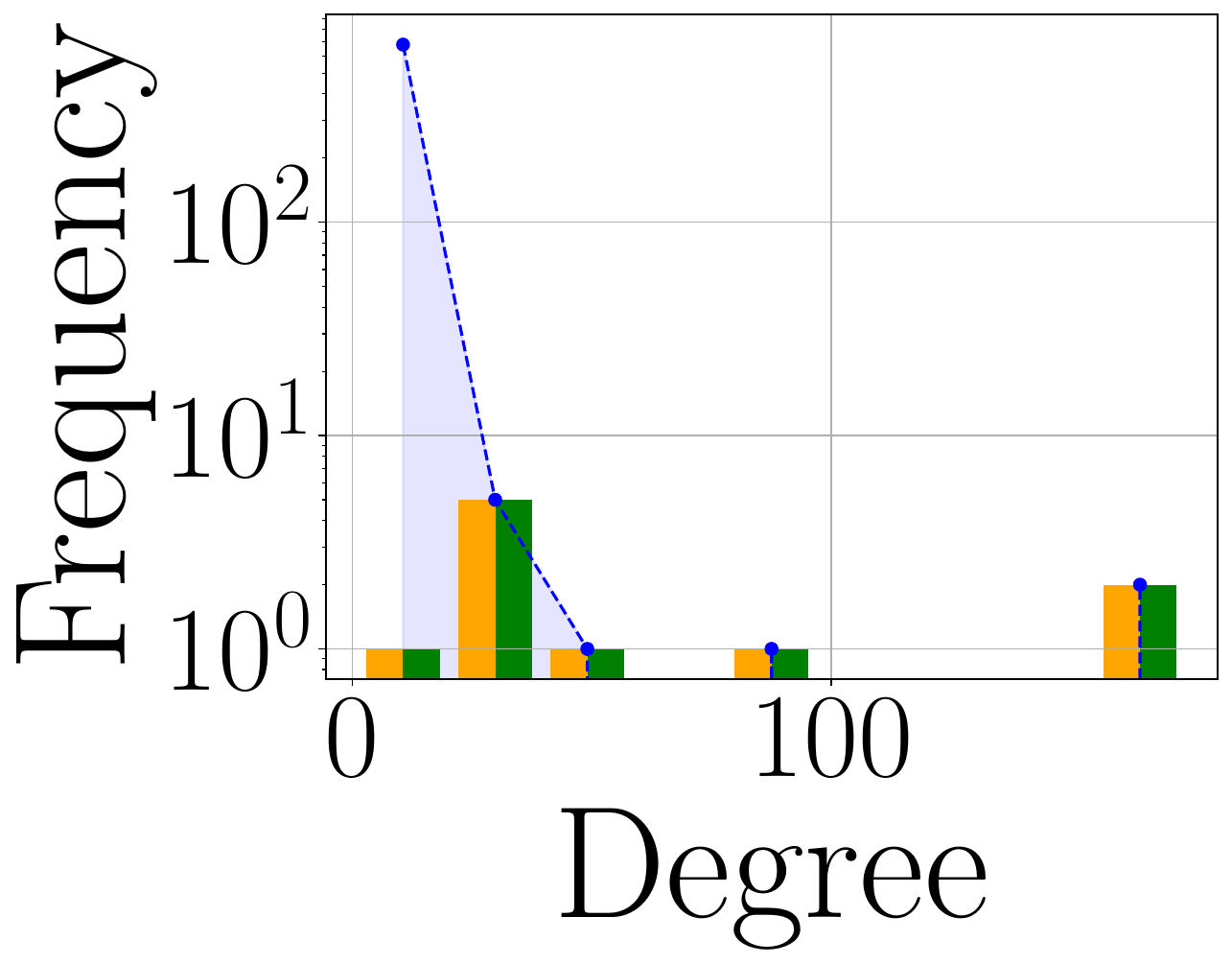}
    \caption{\pl{}\\$IC(0.1, 4)$}
    \end{subfigure}
   \begin{subfigure}{.32\linewidth} 
    \includegraphics[width=\linewidth]{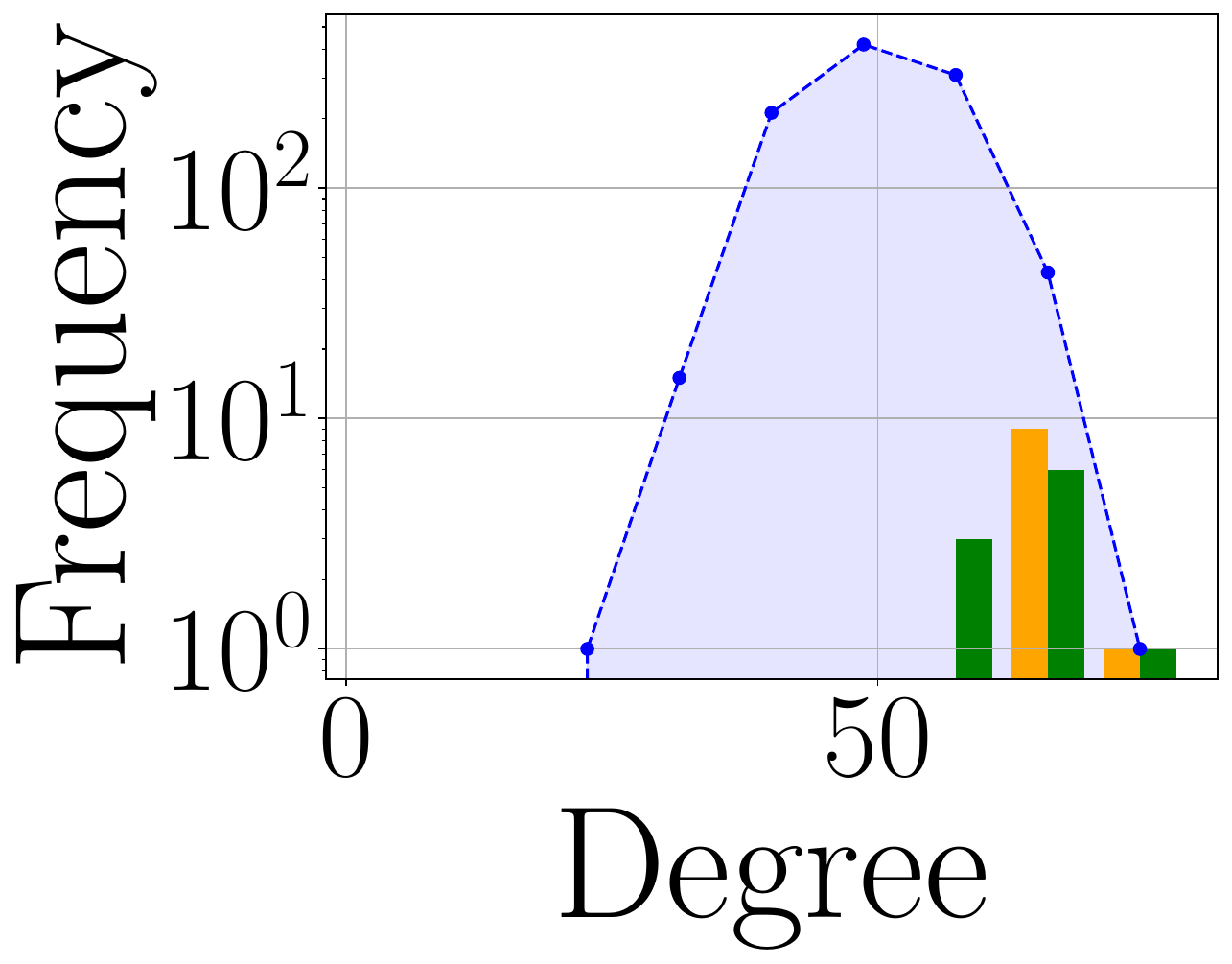}
    \caption{\er{} \\$IC(0.07,2)$}
    \end{subfigure}
   \begin{subfigure}{.32\linewidth} 
    \includegraphics[width=\linewidth]{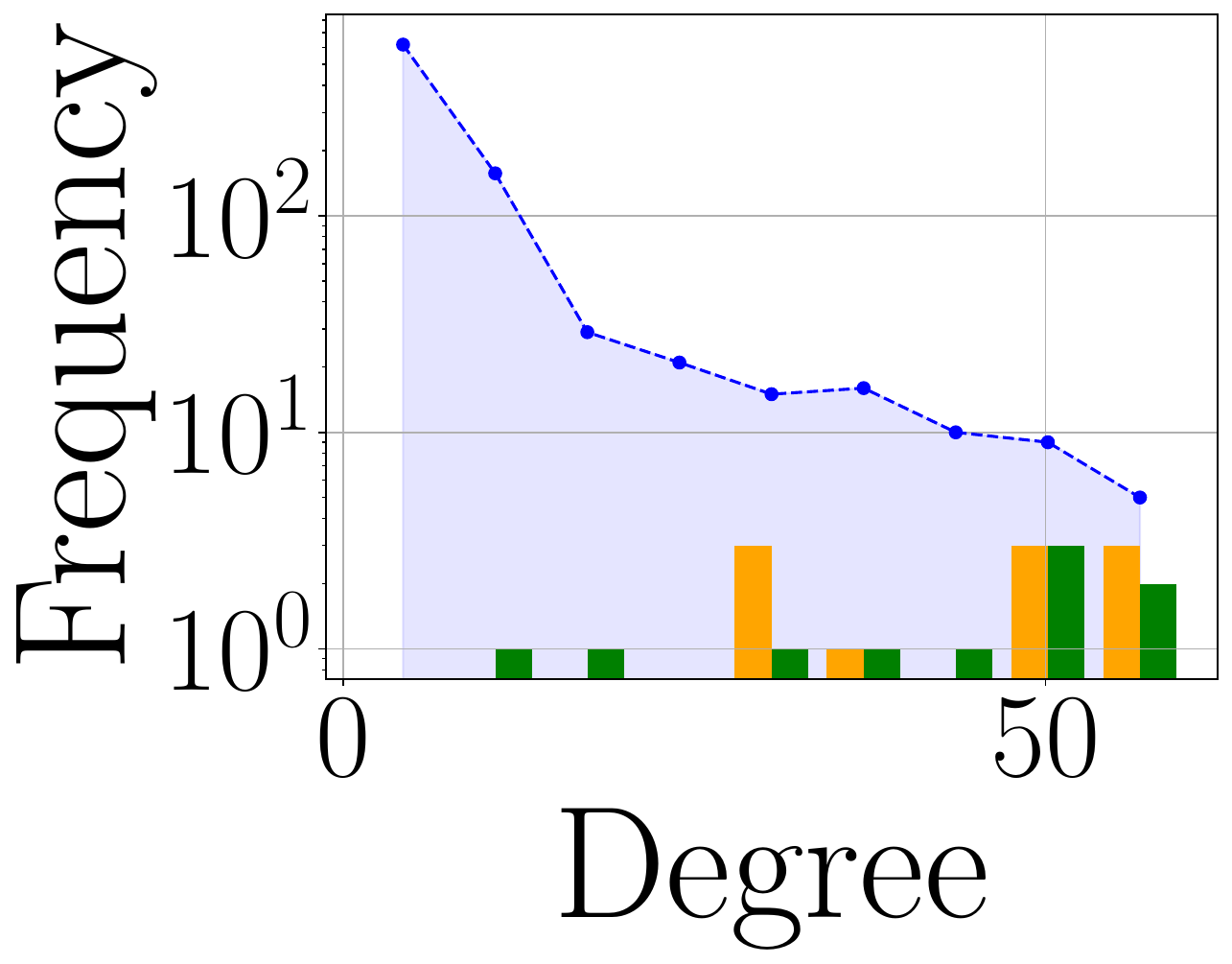}
    \caption{\icu\\$IC(0.2,4)$}
    \end{subfigure}
    \includegraphics[width=0.9\linewidth]{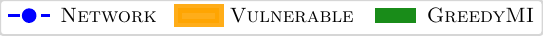}
    \caption{Comparison of the degree distributions of \textsc{GreedyMI} and \textsc{Vulnerable} in \rsource{} seeding (Y-axis in log-scale).}
    \label{fig:degree_dist_random}
\end{figure}

\begin{figure}[htb]
    \centering
    \begin{subfigure}{.32\linewidth}
    \includegraphics[width=\linewidth]{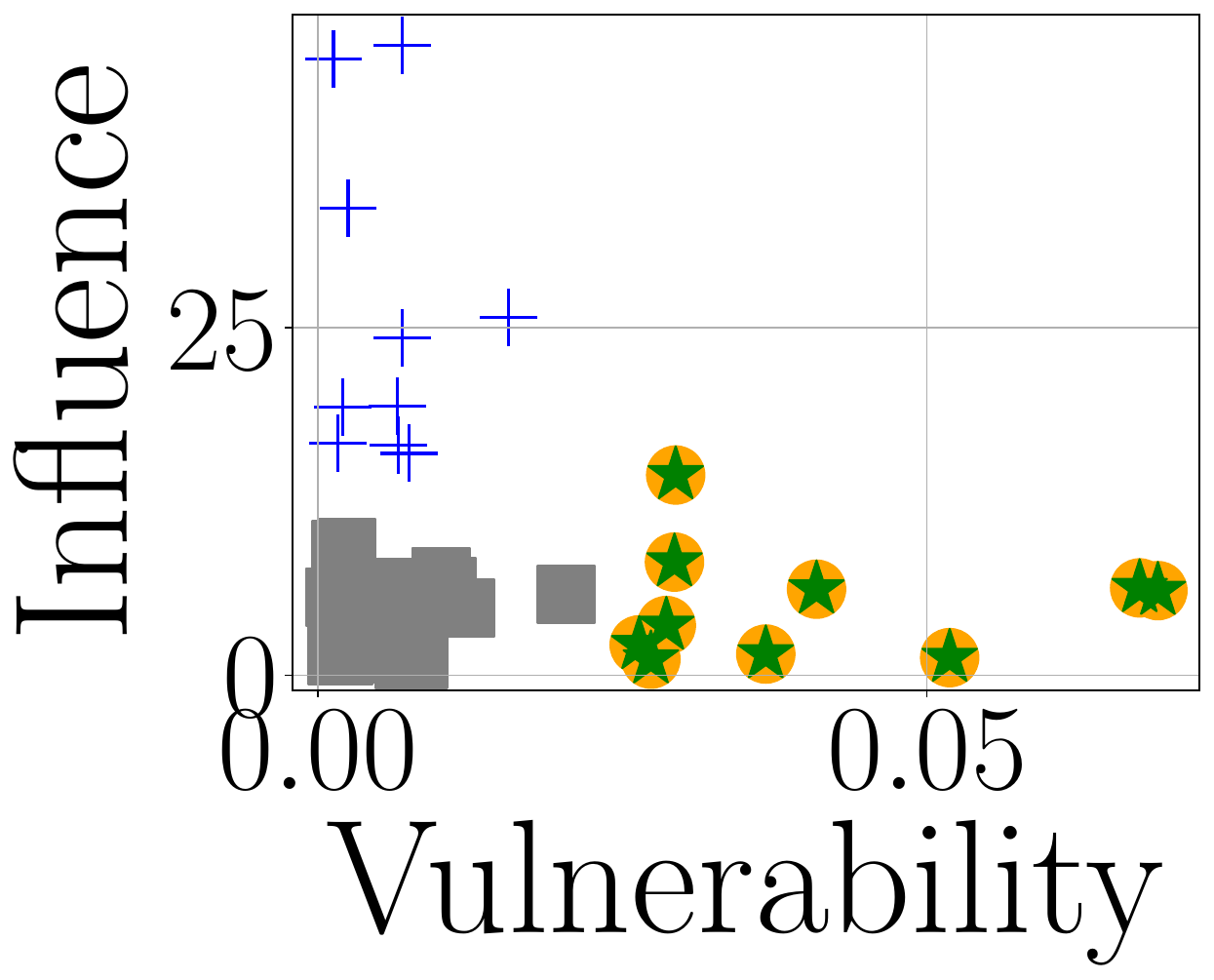}
    \caption{\pl{}\\$IC(0.1, 4)$}
    \end{subfigure}
    \begin{subfigure}{.32\linewidth}
    \includegraphics[width=\linewidth]{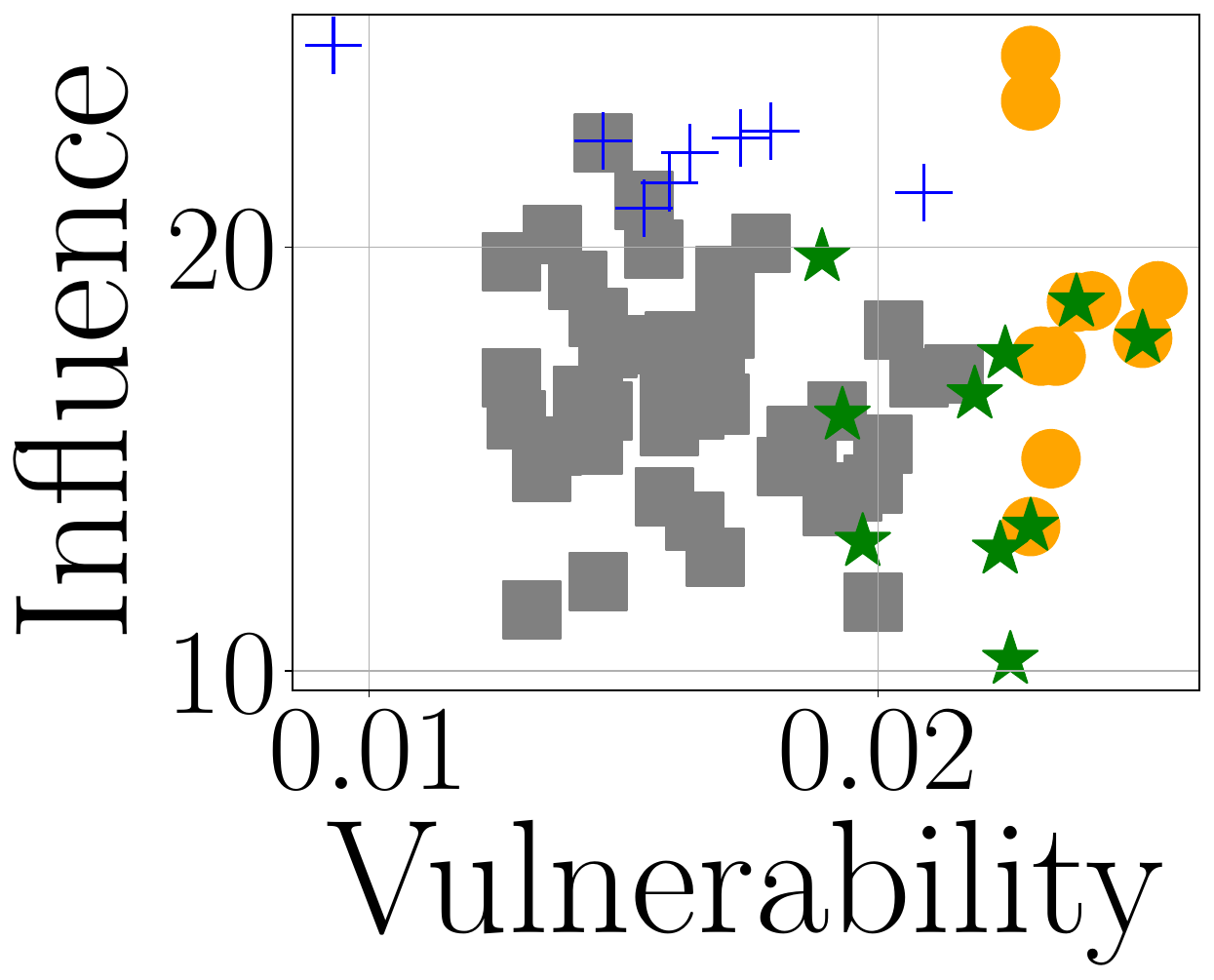}
    \caption{\er{} \\$IC(0.07,2)$}
    \end{subfigure}
    \begin{subfigure}{.32\linewidth}
    \includegraphics[width=\linewidth]{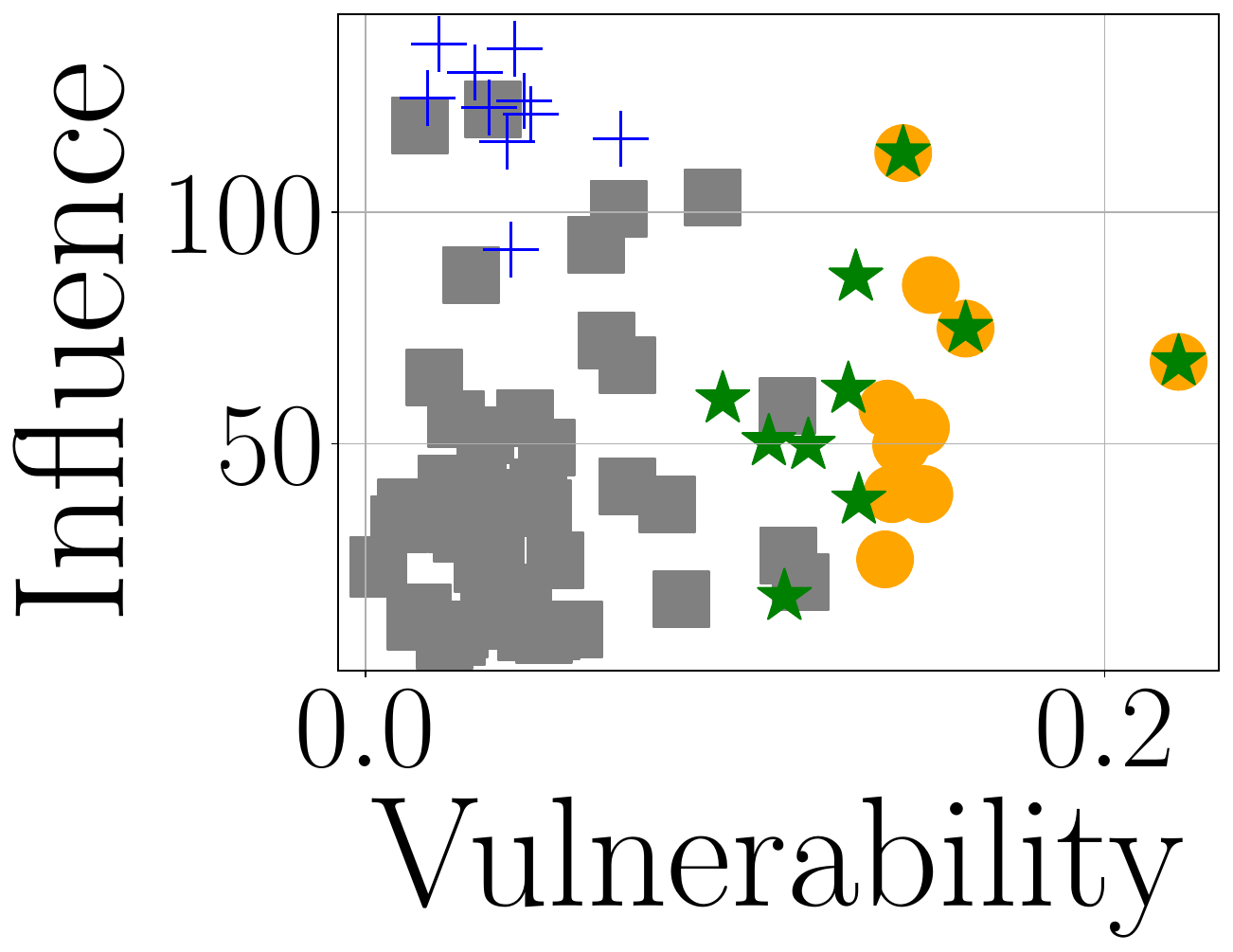}
    \caption{\icu\\$IC(0.2,4)$}
    \end{subfigure}
    \includegraphics[width=0.9\linewidth]{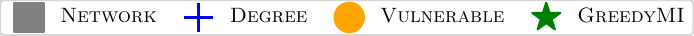}
    \caption{Vulnerability vs Influence of \textsc{GreedyMI} and \textsc{Vulnerable} in \rsource{} seeding. The ``Network'' corresponds to a sample of nodes that do not feature in any of the solution sets.}
    \label{fig:inf_vul_scatter}
\end{figure}

\paragraph{Sample size vs. performance.}
To assess the effect of sample size on the performance of \textsc{GreedyMI}, we adopt
the following progressive sampling-based approach. We choose~1000 samples in each
round. In iteration~$i$, we have a total of~$1000\cdot i$ samples. We find the
\textsc{GreedyMI} solution $A_{greedy}$ corresponding to these samples and note its 
conditional entropy score, $H(Z|X_{A_{greedy}})$.
Figure \ref{fig:num_samp} plots the empirical entropy of the prevalence conditioned on the \textsc{\textsc{GreedyMI}} solution of budget~$k=10$, $H(Z|X_{A_{greedy}})$, 
computed using the combined set of samples at the end of each round of sampling. Firstly, 
we observe that the entropy estimate increases with the number of samples as more of the
unseen joint distribution $P(Z;X_A)$ is sampled. 
We observe that the regimes in which large cascades are possible take more rounds of
sampling to converge. In each round of sampling, even a few large cascades can push up
the estimate due to the presence of the logarithm term in the entropy, triggering a new 
round of sampling.  For a network of size $n=1000$ and budget $k=10$ the maximum joint 
alphabet size $|Z|\times |X_A|= n2^k$ can be very large. Yet, the number of cascade
samples~(30,000) required for the conditional entropy to converge is about~34 times 
fewer in size. This shows that in many realistic problem setups, sampling-based methods
are not prohibitively expensive.
\begin{figure}[htb]
    \centering
    \includegraphics[width=0.94\linewidth]{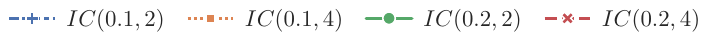}
    \begin{subfigure}{.36\linewidth}
    \includegraphics[width=\linewidth]{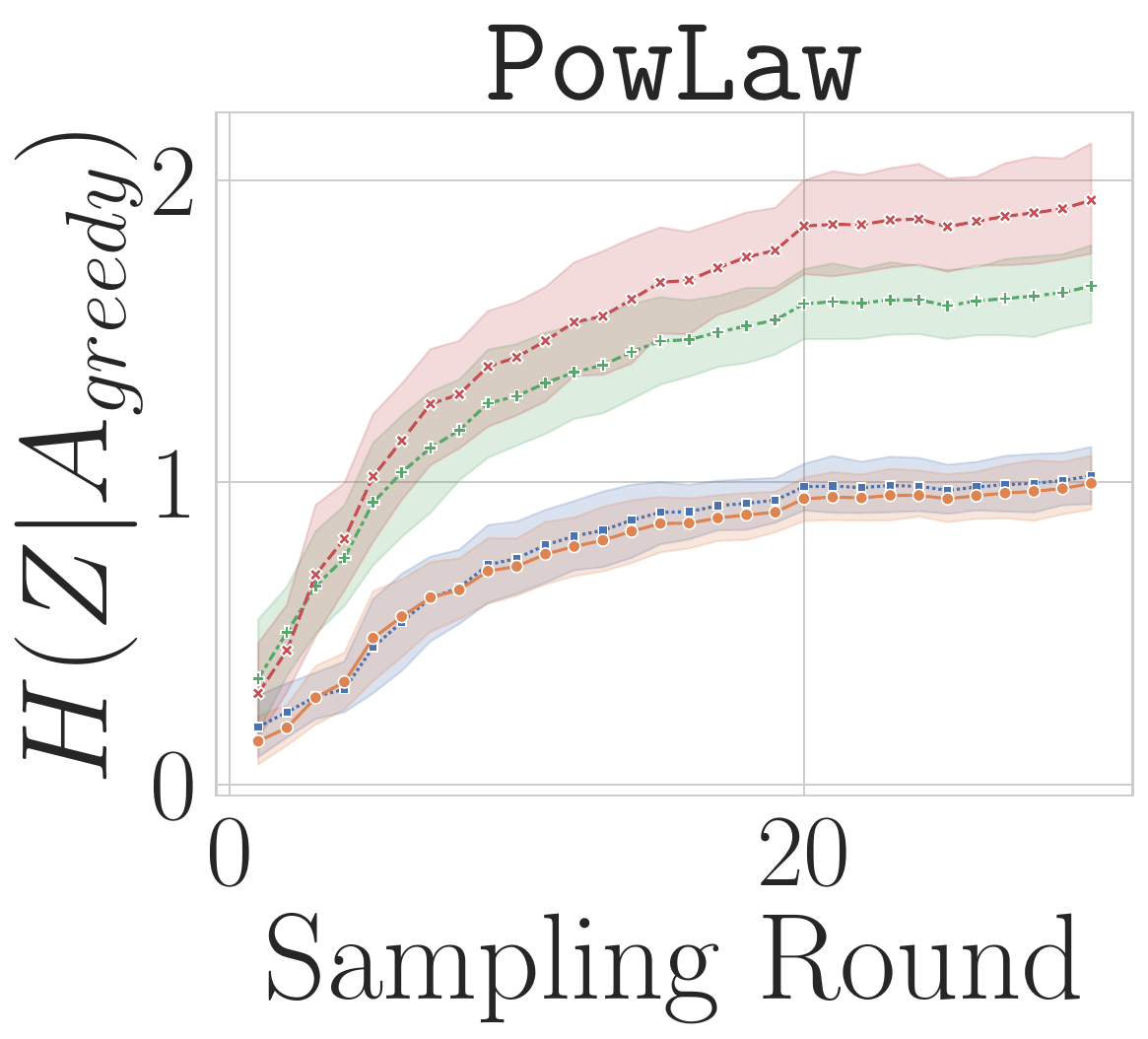}
    \end{subfigure}
    \begin{subfigure}{.36\linewidth}
    \includegraphics[width=\linewidth]{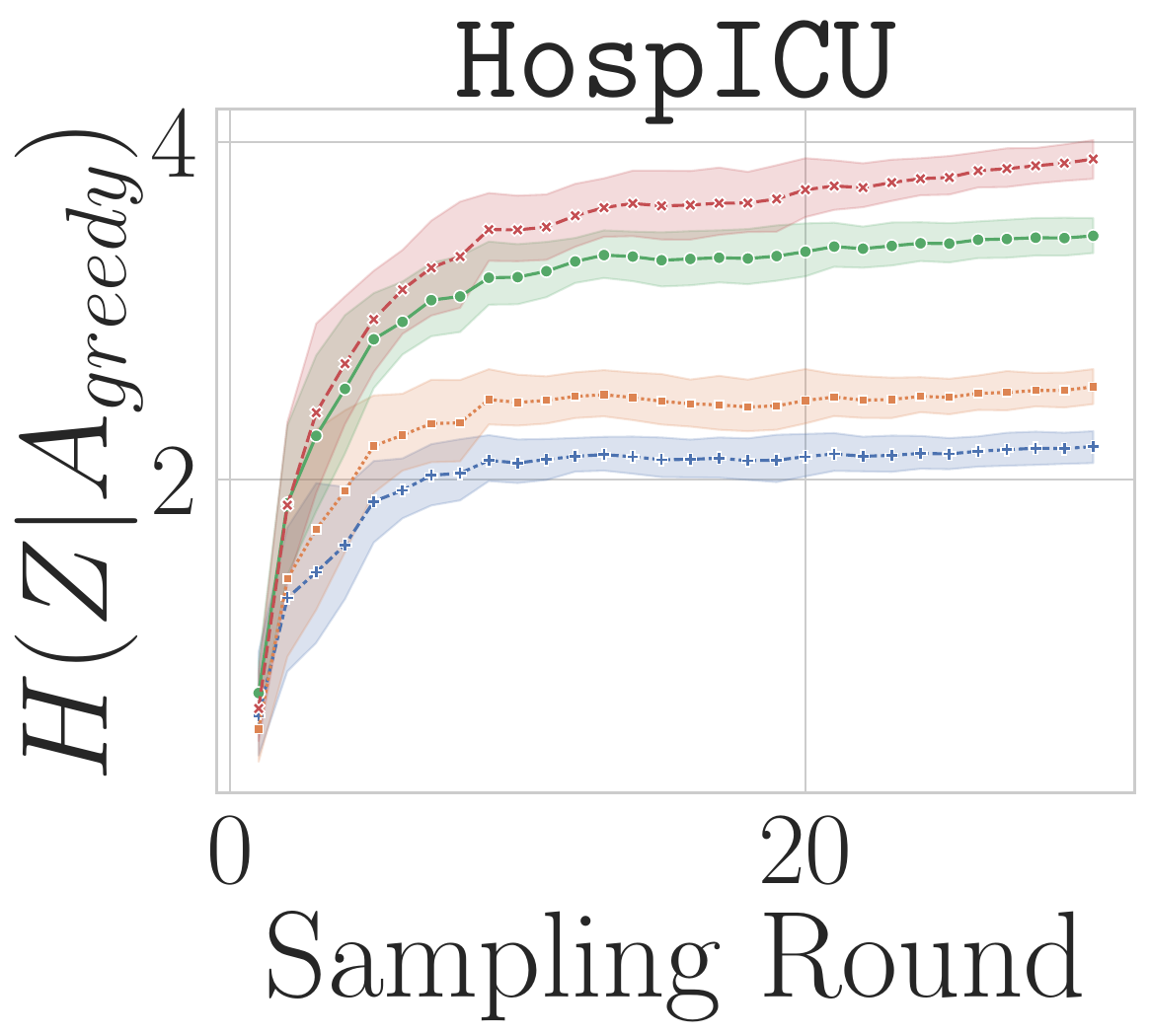}
    \end{subfigure}
    
    \caption{Convergence of the conditional entropy estimate $H(Z|A_{greedy})$. } 
    \label{fig:num_samp}
\end{figure}

\section{Conclusions}
\label{sec:concl}
Our work presents an information-theoretic framework for active disease surveillance,
offering -- for the first time in the literature -- novel, distribution-level insights
into outbreak size. As our results demonstrate, solutions to \prob{} yield results
with substantially lower variance than strategies such as node
degree, which have been frequently proposed in the literature~\cite{Browne-etal-2024,bai2017optimizing}.
At the same time, our findings underscore the challenges of adopting this approach.
First, developing algorithms with provable performance guarantees remains an open problem.
Unlike problems such as influence maximization, desirable 
properties such as submodularity may not hold, making algorithmic analysis more difficult.
Second, closed-form solutions for graph families beyond paths may lead to further insights. From a
practitioner’s perspective, there is also a need to speed-up the greedy method by, for example, reconstruction-aware methods  ~\cite{mishra2023reconstructing} for estimating $P(Z|X_A)$. Finally, extending this mutual information
framework to outbreak characteristics such as time to
peak, spatial spread offers a promising future direction.

\section*{Acknowledgments}
This research is partially supported by NSF grants CCF-1918656 and CNS-2317193, and 
DTRA award HDTRA1-24-R-0028, Cooperative Agreement number 6NU50CK000555-03-01
from the Centers for Disease Control and Prevention (CDC) and DCLS, Network
Models of Food Systems and their Application to Invasive Species Spread, grant
no. 2019-67021-29933 from the USDA National Institute of Food and Agriculture.
The work of R. Tandon is supported in part by NIH Award R01-CA261457-01A1, by the US Department of Energy, Office of Science, Office of Advanced Scientific Computing under Award Number DE-SC-ERKJ422, and US NSF under Grants CCF-2100013, CNS-2209951, CNS-2317192.

\bibliography{references}
\clearpage
\appendix
\onecolumn
\begin{center}
\fbox{{\textbf{Technical Supplement: Information Theoretic Optimal Surveillance for Epidemic Prevalence in Networks}}}
\end{center}

\bigskip 

\section*{Additional details for Section~\ref{sec:prob_def}}

\medskip 

\textbf{Lemma~\ref{lem:remain-sum}.} For any subset of variables $A \subseteq V$,
$H(Z|X_A) = H(Z_A^-| X_A)$. 
When the variables are all independent, $H(Z|X_A) = H(Z_A^-)$.

\begin{proof}

\begin{align*}
    H(Z  | X_A) &= \sum_{\mathbf{x} \in \mathcal{X}_A} P( X_A = \mathbf{x}) H(Z|X_A = \mathbf{x}) \\
    &= \sum_{\mathbf{x} \in \mathcal{X}_A} P( X_A = \mathbf{x}) H(Z_A^- + \mathbf{x}^T\mathbf{1} |X_A = \mathbf{x}) \\
    &= \sum_{\mathbf{x} \in \mathcal{X}_A} P( X_A = \mathbf{x}) H(Z_A^-|X_A = \mathbf{x})\\
    &= H(Z_A^-|X_A)
\end{align*}
When all the random variables $V$ are independent, $H(Z_A^-|X_A) = H(Z_A^-)$, as $X_A$ and the sum of the remaining variables $Z_A^-$ are also independent. 
\end{proof}

\bigskip

\section*{Additional details for Section~\ref{sec:analytical}}

\medskip 

\noindent
\textbf{Theorem~\ref{thm:gen_prob_hardness}.}
(a) \prob{} is NP-hard even when the
propagation is restricted to one hop.
(b) The problem remains
NP-hard even if the constraint 
on the weight of $A$
can be violated by a factor $(1-\epsilon)\log{n}$,
for any $\epsilon < 1$, where $n = |V|$.

\medskip 

\noindent
\begin{proof} 

\noindent
\textbf{Part (a):} We use a reduction from the
Minimum Set Cover (MSC) problem which is
defined as follows: given a base set
$T = \{t_1, t_2, \ldots, t_{\ell}\}$ with
$\ell$ elements, a collection 
$\Gamma = \{\gamma_1, \gamma_2, \ldots, \gamma_r\}$
of $r$ subsets of $T$ and an integer 
$\alpha \leq r$,
is there a subcollection $\Gamma' \subseteq \Gamma$
with $|\Gamma'| \leq \alpha$ such that the union
of the sets in $\Gamma'$ is equal to $T$?
MSC is known to be NP-complete~\cite{GJ-1979}.

Given an instance of MSC, we construct an instance
of \prob{} as follows.
\begin{description}
\item{(a)} The node set $V$ of the contact network
consists of two disjoint subsets $V_1$ and $V_2$.
For each set $\gamma_j \in \Gamma$ of MSC, we have a node
$a_j \in V_1$.
For each element $t_i \in T$ of MSC, we have a node
$b_i \in V_2$. 
Thus, $|V| = r + \ell$.
\item{(b)} The weight of each node in $V_1$ is chosen as 0 while the weight of each node in $V_2$ is chosen
as 1.
\item{(c)} The cost of each node in $V_1$ is chosen as 1 
and the cost of each node in $V_2$ is chosen as
$\alpha + 1$, where $\alpha$ is the bound from
the MSC problem.
\item{(d)} The edge set $E$ of $G$ is chosen as follows.
For each set $\gamma_j$ and element $t_i$ of MSC,
we add the directed edge $(a_j, b_i)$ to $E$.
The probability $\lambda_e$ for each 
edge $e \in E$ is set to 1.
\item{(e)} The constraint on the weight of $|A|$ is set to  $\alpha$.
\item{(f)} The marginal probability for each node in $V_1$ is $1/r$ while that for each node in $V_2$
is $1/\ell$.
\item{(g)} The required conditional entropy
$H(Z|X_A)$ is set to 0.
\end{description}
This completes the construction which can clearly
be done in polynomial time.
We note that the constructed contact network
allows only a 1-hop propagation.

Suppose there is a solution $\Gamma'$ with
$|\Gamma'| = \alpha$ to MSC.
Without loss of generality, let the solution be
$\Gamma' = \{\gamma_1, \gamma_2, \ldots, 
\gamma_{\alpha}\}$.
We choose the nodes in $V_1$ corresponding to
the sets in $\Gamma'$ as the set $A$;
that is, $A$ = $\{a_1, a_2, \ldots, a_{\alpha}\}$.
Thus, the cost of $A$ = $\alpha$.
Since the sets in $\Gamma'$ cover all the elements
of $T$ and all the edge probabilities in $G$ are 1,
the chosen set $A$ causes all the nodes $V_2$
to be infected in one step.
In other words, the only possible value for the weighted prevalence in this case is $\ell$.
As we have complete information about
the system, the conditional entropy $H(Z|X_A)$
is 0. Thus, we have a solution to the constructed
\prob{} instance.

For the converse, suppose there is a subset 
$A \subseteq V$
with cost at most $\alpha$ such that 
$H(Z|X_A)$ = 0.
We first note that $A$ cannot contain any node
from $V_2$ since the cost of each node in $V_2$
is $\alpha+1$.
Since the cost of each node in $V_1$ is 1, 
$A$ contains $\beta{} \leq \alpha$
nodes from $V_1$.
Without loss of generality, let 
$A = \{a_1, a_2, \ldots, a_{\beta}\}$.
Consider the subcollection 
$\Gamma' = \{\gamma_1, \gamma_2, \ldots,
\gamma_{\beta}\}$ containing the sets
corresponding to the nodes in $A$.
For the sake of contradiction,
suppose $\Gamma'$ does not cover the base set $T$; that is, the union of all the sets in $\Gamma'$
does not include all the elements 
of $T$. 
Specifically, let $t_j \in T$ be an element
which is not covered by $\Gamma'$.
This implies that in the network $G$,
there is no directed edge from any of the nodes
in $A$ to the node $b_j \in V_2$.
In other words, the random variable $X_{b_j}$ associated with $b_j$ is \emph{independent} of the random variables associated with any of the nodes in $A$.
Since $b_j$ has an infection probability of $1/\ell$, where $0 < 1/\ell < 1$,
the entropy $H(X_{b_j})$ of
the binary random variable $X_{b_j}$ is
\emph{greater than} 0.
As a consequence, the conditional entropy
$H(Z|X_A)$ is also \emph{greater than} 0.
This contradicts the assumption that
$H(Z|X_A) = 0$. Hence $\Gamma'$ is a solution
to the MSC instance, and this completes
our proof of Part~(a) of 
Theorem~\ref{thm:gen_prob_hardness}. 

\medskip 

\noindent
\textbf{Part (b):} We note that the reduction
presented in Part~(a) preserves approximations;
that is, from any solution $A$ with cost $\alpha$ to the instance of the \prob{} constructed above, a solution with $\alpha$ sets can be
constructed for MSC.
Hence, if there is a polynomial time algorithm that produces a test set $A$ whose cost is within a factor $\rho \geq 1$ of the minimum for \prob{}, then the algorithm can be used to obtain a similar
approximation for the MSC problem.
It is known that it is NP-hard
to approximate MSC to within a factor 
$(1-\epsilon)\log{n}$ for 
any $\epsilon < 1$~\cite{Feige-1998}. 
Hence, the same
inapproximability result also holds 
for \prob{}.
\end{proof} 

\bigskip 

\noindent
\textbf{Observation~\ref{obs:CZK}.}
    There exist instances in which optimizing the Caselton-Zidek criterion $\mathcal{K}(A) = I(X_A; X_{V\setminus A})$, does not optimize for mutual information with prevalence $\obj(A) = I(X_A;Z)$.

\begin{proof}
Consider the graph in Figure \ref{fig:ex_indirect}. Node $a$ is the only source and the disease probability is $\p = 0.5$ over all edges. First, we will find the optimal node for maximizing the CZK objective. Recall that $h(.)$ is the binary entropy function. For simplicity, we use the node id as the variable for its status as well; that is, $u$ refers to $X_u$ in the following expressions.
\begin{align*}
    H(b) &= h(\p) = 1 \\
    H(b|c,d,e) &= P(c,d,e=0,0,0) H(b| c,d,e=0,0,0) \\
    &= (1-\p^2)h\left(\frac{\p(1-\p)}{1-\p^2}\right)\\
    &= 0.75*h(1/3) = 0.689\\
    I(b;\{c,d,e\}) &= H(b) - H(b|c,d,e) = 1 - 0.689 =  0.311  \\
    H(c) &= h(\p^2) = h(0.25) = 0.811 \\
    H(c|b,d,e) &= P(b,d,e=1,0,0) H(c|b,d,e=1,0,0) \\
    &= [\p^2(1-\p)^2 + \p(1-\p)]\\
    & \quad h\left(\frac{\p^2(1-\p)^2)}{\p^2(1-\p)^2 + \p(1-\p)}\right) \\
    &= 0.3125 \times h(1/5) =  0.226\\
    I(c;\{b,d,e\}) &= H(c) - H(c|b,d,e) = 0.585
\end{align*}
Hence $c$ maximizes the objective of the form $I(u; V\setminus \{u\})$. 

Now, we will find the node which maximizes $\obj$ which targets prevalence. We use the fact that on this network with the source at $a$, knowing $Z$ makes the states of nodes $b$ and $c$ deterministic.
\begin{align*}
    I(Z;b) &= H(b) - H(b| Z) = H(b) - 0 = 1 \\
    I(Z;c) &= H(c) - H(c| Z) = H(c) - 0 = 0.811
\end{align*}
Here, $b$ maximizes mutual information with prevalence of the form $I(\{u\}; Z)$. This shows that the optimal node to query can be different under these two objectives.

\end{proof}

\begin{observation}
    There exist instances where the CZK criterion is not informative whereas the mutual information with prevalence  $\obj(.)$ can inform the choice of a surveillance set.
\end{observation}
\begin{proof}
     Consider a network with a set of nodes $V= V_1 \cup V_2$ where $V_1$ and $V_2$ are disjoint and $|V_1|=|V_2| = n$ and no edges. Each of the nodes has a probability of being infected independent of others. The probability
     values are chosen as follows:
     \begin{align*}
     p_i =
        \begin{cases}
            0 & \text{ if } i\in V_1 \\
            0.5 & \text{ if } i\in V_2
        \end{cases}
     \end{align*}
      Under the CZK criterion, we want to find a subset $A\subseteq V$ within budget $k$ which maximizes $I(X_A; X_{V\setminus A})$. For any set $A$,
     \begin{align*}
         &I(X_A;X_{V\setminus A}) = H(X_A) - H(X_A|X_{V\setminus A}) \\
         &= H(X_A) - H(X_A) = 0
     \end{align*}
     due to independence between $X_A$ and $X_{V\setminus A}$. Thus the CZK criterion makes no recommendation regarding how to choose nodes from among $V_1$ and $V_2$ even though $X_{V_1}$ is deterministic.

     Under the prevalence information criterion, we want to maximize $I(Z;X_A) = H(Z) - H(Z|X_A)$ which is equivalent to minimizing $H(Z|X_A)$. By Lemma \ref{lem:remain-sum}, $H(Z|X_A) = H(Z_A^-)$. 
     Let us suppose we choose $A_1$ and $A_2$ to query among $V_1$ and $V_2$ respectively, and $|A_1|+|A_2|=k$. 
     \begin{align*}
        &H(Z_{A_1\cup A_2}^-) = H(\sum_{i\in V_1\setminus A_1} 0 + \sum_{i\in V_2\setminus A_2} X_i )\\
        &= H(\sum_{i\in V_2\setminus A_2}X_i) \\
        &= H(\sum_{i\in V_2\setminus A_2}X_i) \\
        &= H(\text{Binomial}(n-|A_2|, 0.5))
     \end{align*}
     Here, Binomial$(n,p)$ denotes the binomial distribution with parameters $n$ and $p$. 
    Thus, $\obj(A)$ is non-zero allowing the comparison of different surveillance sets $A$.
    Since we would like to choose $A_1,A_2$  within budget $k$ so as to minimize $H(Z|X_A)$, the prevalence information criterion suggests maximizing the size of $A_2$; that is, $|A_2|=k$.
\end{proof}
\begin{figure}
    \centering
    \includegraphics[width=0.4\linewidth]{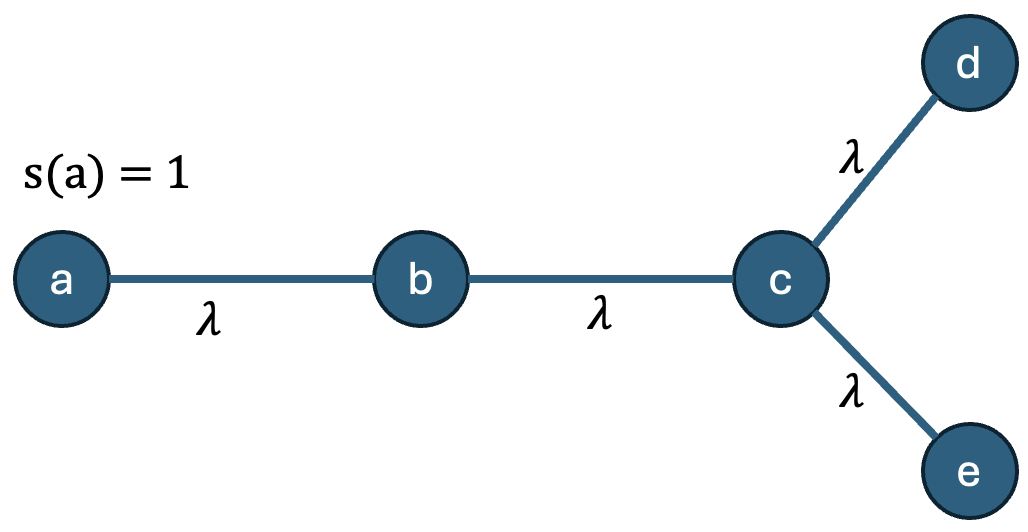}
    \caption{Example showing the CZK solution is not optimal for \prob{}}
    \label{fig:ex_indirect}
\end{figure}

\begin{figure}
    \centering
    \includegraphics[width=0.3\linewidth]{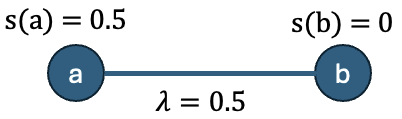}
    \caption{Example showing $\obj(A) $ is submodular}
    \label{fig:ex_submod}
\end{figure}

\bigskip

\noindent
\textbf{Observation~\ref{obs:ind-supmod}:}
  Given a set of nodes $V$ whose states are mutually independent random variables, the function $\obj{}(A) = I(Z_V;X_A)$ is supermodular in $A \subseteq V$.

\begin{proof}
     We would like to compare the marginal gain in the function $\obj(A)$ upon adding a new node to $A$. Let $A\subset B \subset V$ be any subsets and let $w \in V\setminus B$ whose state is $X_w$. Let $Z_{A}^-$ represent the sum of node state variables excluding those in $A$.
    \begin{align}
        \obj(A\cup \{w\}) - \obj(A)
        &= I(Z;X_{A,w}) - I(Z;X_A) \\
        &= H(Z|X_A) - H(Z|X_{A,w}) \\
        &= H(Z_{A}^-) - H(Z_{A,w}^-) \\
        \obj(B\cup \{w\}) - \obj(B) &= H(Z_{B}^-) - H(Z_{B,w}^-)
    \end{align}
    From \cite{madiman2008entropy}, the following holds for any variables $Y_1, Y_2, Y_3$ which are mutually independent.
    \begin{align}
        H(Y_1 + Y_2) + H(Y_2 + Y_3) \geq H(Y_1+Y_2+Y_3) + H(Y_2) \label{eqn:madiman}
    \end{align}
    Set $Y_1 = Z_{B\setminus A}$, $Y_2 = Z_{B,w}^-$, $Y_3= X_w$. $Y_1, Y_2, Y_3$ are independent as they are sums of non-overlapping sets of random variables. 
    \begin{align}
        Y_1 + Y_2 &= Z_{B\setminus A} + Z_{B,w}^- = Z_{A,w}^- \\
        Y_2 + Y_3 &= Z_{B,w}^- + X_w = Z_B^- \\
        Y_1 + Y_2 + Y_3 &= Z_A^-\\
    \end{align}
    Using Inequality~\eqref{eqn:madiman},
    \begin{align}
        H(Z_{A,w}^-) + H(Z_B^-) &\geq H(Z_A^-) + H(Z_{B,w}^-) \\
        H(Z_B^-) - H(Z_{B,w}^-) &\geq H(Z_A^-) - H(Z_{A,w}^-) \\
        \obj(B\cup \{w\}) - \obj(B) &\geq \obj(A\cup \{w\}) - \obj(A)
    \end{align}
    Hence, the marginal gain from adding $w$ to a larger set $B$ is higher than adding to the smaller subset $A$, proving supermodularity.
\end{proof}

\bigskip

\noindent
\textbf{Observation~\ref{obs:submodular}.}
There exist instances in which the function $\obj$ is submodular.

\begin{proof}
    In Figure \ref{fig:ex_submod}, we have a graph with two nodes $\{a,b\}$ and seeding probabilities $(s_a=0.5, s_b=0)$, disease probability $\lambda=0.5$. The marginal gain $\mathcal{G}$ upon adding a node $w$ to a subset $A$ is
    \begin{align*}
        & \mathcal{G}(A, w) = \obj(A\cup \{w\}) - \obj(A) \\
        &= I(Z; X_{A\cup\{w\}}) - I(Z; X_A) \\
        &= H(Z) - H(Z|X_{A\cup\{w\}}) - H(Z) + H(Z|X_A) \\
        & = H(Z|X_A) - H(Z|X_{A\cup\{w\}})
    \end{align*}
    We have,
    \begin{align*}
        &H(Z) = -(1-s_a)\log (1-s_a) - s_a(1-\lambda)\log s_a(1-\lambda)  \\
        &\quad - s_a\lambda \log s_a\lambda \\
        &H(Z|X_a) = -s_a[(1-\lambda)\log(1-\lambda) + \lambda \log \lambda]\\
        &H(Z|X_b) = -(1-s_a)\log \frac{1-s_a}{1-s_a\lambda} \\
        &\quad - s_a(1-\lambda)\log \frac{s_a(1-\lambda)}{1-s_a\lambda} \\
        &H(Z|X_a, X_b) = 0 \\
        &\mathcal{G}(\phi, a) = H(Z) - H(Z|X_a) \\
         &\mathcal{G}(\phi, b) = H(Z) - H(Z|X_b) \\
          &\mathcal{G}(\{a\}, b) = H(Z|X_a) - H(Z|X_a, X_b) = H(Z|X_a) \\
          &\mathcal{G}(\{b\}, a) = H(Z|X_b) - H(Z|X_a, X_b) = H(Z|X_b)
    \end{align*}
    For submodularity, we need,
    \begin{align*}
        &\mathcal{G}(\phi, a) - \mathcal{G}(\{a\}, b) \geq 0 \\
        &\Rightarrow H(Z) - 2H(Z|X_a) \geq 0 \\
        &\text{ and } \\
        &\mathcal{G}(\phi, b) - \mathcal{G}(\{b\}, a) \geq 0  \\
        &\Rightarrow H(Z) - 2H(Z|X_b) \geq 0
    \end{align*}
    With $s_a=0.5$ and $\lambda =0.5)$, $H(Z) = 1.5, H(Z|X_a) = 0.5, H(Z|X_b) = 0.688$, which satisfies the above conditions for submodularity.
    
\end{proof}

\noindent
\textbf{Observation~\ref{obs:detprob}.}
There exist instances in which node selection for maximizing detection likelihood can lead to solutions with objective value for \prob{} less than the optimal by $\Theta(n)$.

\begin{proof}
    Consider a network with a set of nodes $V=(V_1,V_2), |V_1|=|V_2|=n$ and no edges. Each of the nodes has an independent probability of being infected, $p_i = 1$ if $i \in V_1$, $p_i=0.5$ if $i \in V_2$. Let the set of nodes to query be $A=(A_1,A_2)$ such that $A_1\subseteq V_1, A_2 \subseteq V_2$. Given a budget $k < n$, let's say we want to maximize detection likelihood $P(X_A =1) = \prod_{i\in A} P(X_i = 1)$. The objective is maximized when we choose $A = A_1$ such that $P(X_A) = 1$. The mutual information of this subset with the prevalence is, 
    \begin{align*}
    &I(X_{A_1};Z) = H(Z) - H(Z|X_{A_1}) \\
    &=H(Z) - H(Z_{A_1}^-) = H(Z) - H(n-k + \sum_{i\in V_2} X_i) \\
    &=H(Z) - H(\sum_{i\in V_2} X_2) = H(Z) - H(\text{Binomial}(n,0.5))
    \end{align*}
    However, if we query with the goal of maximizing mutual information with prevalence we have,
    \begin{align*}
        &I(X_A;Z) = H(Z) - H(Z|X_A) \\
        &= H(Z) - H(Z_{A}^-) = H(Z) - H(\sum_{i\in V_1\setminus A_1} X_i + \sum_{j\in V_2\setminus A_2}X_j) \\
        &= H(Z) - H(n-|A_1| + \sum_{j\in V_2\setminus A_2}X_j) \\
        &= H(Z) - H(\sum_{j\in V_2\setminus A_2}X_j) \\
        &= H(Z) - H(\text{Binomial}(n-|A_2|, 0.5))\\
    \end{align*}
    Here, Binomial($n,p$) denotes the binomial distribution with parameters $n$ and $p$.
    This is maximized when the size of $A_2$ is maximized, i.e., $|A_2|=k$ such that $A=A_2$.
    Since the entropy of Binomial($n,0.5$) $>$ Binomial($n-k, 0.5$), sensor placement for maximizing detection likelihood is suboptimal for \prob{}.
\end{proof}

\section*{Additional details for Section 5: Our Approach}

\medskip 

\begin{algorithm}
\textbf{Input:} A network $G=(U\cup W, E)$, a budget $k$, $IC(\boldsymbol{\p}, 1)$\\
\textbf{Output:} Query-set $A$
\begin{algorithmic}[1]
\State Initialize $A$ to $\emptyset$. 
\State Compute the probability of infection for each node, $\{p_i\}_{i\in W}$.
\For{$j=1$ to $k$}
\For{each $v\in W\setminus A$}
\State $ \delta_v \gets H(Z|X_{A\cup \{v\}})$
\EndFor
\State $v^* \gets \argmin_v \delta_v$
\State $A \gets A \cup \{v^*\}$ 
\EndFor
\State \textbf{return} A
\end{algorithmic}

\caption{\textsc{$1$-hopGreedySelection}}
\label{alg:1hop_greedy}
\end{algorithm}

\begin{algorithm}
    \caption{\textsc{EntropyOnTree}}
    \label{alg:entr_tree}
    \textbf{Input:} Rooted tree $T_r$, subset of nodes $A$, $IC(\boldsymbol{\p}, \infty)$ \\
    \textbf{Output:} Conditional entropy of the prevalence $H(Z|X_A)$
    \begin{algorithmic}[1]
        \State Initialize arrays $S,H$ set to zeros
        \State $\mathcal{X}_A \gets \{0,1\}^{|A|}$
        \For{$i=0$ to $2^k-1$}
            \State $\mathbf{x} = \mathcal{X}_A[i]$
            \If{\Call{Feasible}{$T_r,\mathbf{x}$}}

            \State $\tilde{T}_r \gets T_r$ 
            \State $S[i] \gets 1$
            \For{$j=1$ to $k$}
            \State $x_v = \mathbf{x}[j]$
            \State $p(v) \gets \prod_{e \in \Gamma_{rv}} \p_{e}$
            \If{$x_v == 1$}
            \State $S[i] \gets S[i] p(v)$
            \State $\tilde{T}_r \gets$ \Call{Contract}{$\tilde{T}_r,\Gamma_{rv}$}
            \Else
            \State $S[i] \gets S[i](1-p(v))$
            \State $\tilde{T}_r \gets$ \Call{Remove}{$\tilde{T}_r,v$}
            \EndIf
            \EndFor
            \State $R \gets$ \Call{MessagePassing}{$\tilde{T}_r, \{\p_e\}, r$}
            \State $H[i] \gets -R \cdot \log(R)$
        \EndIf
        \EndFor
        \State dot\_prod $\gets 0$
        \For{$j=0$ to $len(S)-1$}
            \State dot\_prod $\gets$ dot\_prod + $S[j]*H[j]$
        \EndFor
        \State \textbf{return} dot\_prod
    \end{algorithmic}
\end{algorithm}

\begin{algorithm}

\textbf{Input:} Rooted tree $T_r$, $IC(\boldsymbol{\p}, \infty)$, node $v$\\
\textbf{Output:} Probability distribution $S=[s_0, s_1, \dots, s_N], s_i = P_{Z_v}(i)$ where $Z_v$ is the subtree prevalence.
\caption{\textsc{MessagePassing}($T_r, \boldsymbol{\p}, v)$}
\label{alg:tree_mp}
\begin{algorithmic}[1]
    \State Initialize messages array $M$
    \For{each $u$ in $v$.children}
        \State $R$ = \Call{MessagePassing}{$T_r, \boldsymbol{\p}, v)$}
        \State $ M_u \gets R$
        \State $M.append(M_u)$
    \EndFor
    \If{$v$.parent $\neq$ NULL}
        \State $p(v) \gets \prod_{e\in \Gamma_{rv}} \p_e$
        \State $S[0] \gets 1 - p(v)$
    \Else
        \State $S[0] \gets 0$
        \State $p(v) \gets 1$
    \EndIf
    \For{$j=1$ to $n$}
        \State $S[j] \gets M_1* M_1*\dots* M_{c_l(v)} [j-1]$
        \State $S[j] \gets p(v)S[j]$
    \EndFor
    \State \textbf{return} $S$
\end{algorithmic}
\end{algorithm}

\noindent
\textbf{Lemma~\ref{lem:tree}.}
For any subset $A\subseteq T_r$,~\textsc{EntropyOnTree} exactly computes the conditional entropy of the prevalence $H(Z|X_A)$.

\begin{proof}
By definition,
\begin{align*}
    H(Z|X_A) = \sum_{\mathbf{x} \in \mathcal{X}_A}P(X_A = \mathbf{x})H(Z|X_A = \mathbf{x})
\end{align*}
On a tree with the source at the root, the probability of a node being infected is $p(v) = \prod_{e\in \Gamma_{rv}} \p_e$ where $\Gamma_{rv}$ is the path from $r$ to $v$. Given an infection vector $\mathbf{x}$ over nodes $A$, let the infected subset be $A_1(\mathbf{x})$ and the uninfected subset be $A_0(\mathbf{x})$. The probability of $\mathbf{x}$ is $P(X_A = \mathbf{x})= \prod_{v\in A_1(\mathbf{x})}p(v) \prod_{u \in A_0(\mathbf{x})}(1-p(v))$. 

\noindent
Next, we show how to use \textsc{MessagePassing} to exactly compute  $H(Z|X_A = \mathbf{x})$ for each $\mathbf{x}$. Note that this algorithm is a special case of the Subtree Distribution Propagation algorithm in \cite{burkholz2021cascade}. Given $\mathbf{x}$ over nodes $A$, we transform $T_r$ by contracting the known live-edge paths and removing the uninfected nodes to obtain $\tilde{T}_r$. From $Z_{\tilde{T}_r} = Z_{T_r} + \sum_{u \in A_1(\mathbf{x})} |\Gamma_{ru}|$, it follows that $H(Z_{\tilde{T}_r}) = H(Z_{T_r}|X_A)$. Now, we need to compute the unconditional entropy of the prevalence on the transformed tree.

\noindent
Consider a node $v\neq r$ and its subtree $\tilde{T}_v$ with $l$ children $c_1, \dots, c_l$. The prevalence in this subtree is $Z_v = X_v + \sum_{i\in c(v)} Z_{i}$. Here, the prevalences in the children subtrees are independent given that $v$ is infected. Hence the distribution of their sum, $Z_v$ can be computed by convolution. 
We have $P_{Z_v}(0) = 1-p(v)$. For $n > 0$,
\begin{align*}
    P_{Z_v}(n) &= P(X_v = 1) P(Z_v = n | X_v = 1) \\
    &= p(v)\left(P_{Z_{c_1}} * \dots * P_{Z_{c_l}}[n-1]\right)
\end{align*}
Here, $*$ denotes convolution. The subtree prevalence probability vector $P_{Z_v}$ is thus passed as messages from each node to its parent starting from the leaves and ending at the root. At the root, the messages are combined slightly differently as $P_{Z_r}(0)= 0$ to obtain the prevalence probability distribution for the entire tree. The final step is to compute the entropy of this distribution $H(Z_{\tilde{T}_r}) = -\sum_n P_{Z_r}(n)\log P_{Z_r}(n)$.

\end{proof}

\begin{algorithm}
\caption{\textsc{Feasible}($T_r, \mathbf{x}$)}
\label{alg:feasible}
    \begin{algorithmic}[1]
        \State let $X[1..n]$ be a binary array denoting infection status
        \For{$j = 1$ to $k$}
            \State $x_v \gets \mathbf{x}[j]$
            \If{$x_v == 1$}
            \For{$u$ in $\Gamma_{rv}$}
                \If{$X[u] == 0$}
                    \State \textbf{return} False
                \Else
                    \State $X[u] \gets 1$
                \EndIf
            \EndFor
            \Else
                \State $X[v] \gets 0$
            \EndIf
        \EndFor
\end{algorithmic}
\end{algorithm}

\begin{algorithm}
\textbf{Input:} A tree network $T_r$ with root $r$, a budget $k$, disease parameters $\{\p_e\}$\\
\textbf{Output:} Query-set $A$
\begin{algorithmic}[1]
\State Initialize $A$ to $\emptyset$
\For{$j=1$ to $k$}
\For{each $v\in V(T_r)\setminus (A\cup \{r\})$}
\State $ \delta_v \gets$ \Call{EntropyOnTree}{$T_r,A\cup \{v\}, \{\p_e\}$}
\EndFor
\State $v^* \gets \argmin_v \delta_v$
\State $A \gets A \cup \{v^*\}$
\EndFor
\State \textbf{return} A
\end{algorithmic}

\caption{\textsc{TreeGreedySelection}}
\label{alg:tree_greedy}
\end{algorithm}

\textbf{Theorem~\ref{lem:path}. }
For budget $k$, a path of sufficiently large length $n > -\log (k+1)/\log \p$, and $IC(\p,\infty)$ with homogeneous disease probability $\p\in (0,1)$, the \prob{}-optimal separation  without integrality constraints is $\{g_j = \log (\frac{k+1-j}{k+2-j}) / \log \lambda; j=1,\dots, k\}$. 

\begin{proof}
    We want to find  the optimal separations $\{g_j; j=1, \dots k\}$  such that $\obj(\{i_1, \dots, i_k\}) = \obj({\{g_1,\dots, g_k\}})$ is maximized. Recall that $h$ is the binary entropy function, defined as $h(p) = -p\log p - (1-p)\log (1-p)$. We will use two simplifying properties of this path setting: (a) once prevalence is given, all the nodes states are deterministic, (b) the node states on the path form a Markov chain.
    \begin{align*}
    &\obj(\{i_1, \dots, i_k\}) = I(Z; X_{i_1}, X_{i_2}, \dots, X_{i_k}) \\
    &= H(X_{i_1}, X_{i_2}, \dots, X_{i_k}) - H(X_{i_1}, X_{i_2}, \dots, X_{i_k} | Z) \\
    &= H(X_{i_1}, X_{i_2}, \dots, X_{i_k}) \\
    &= H(X_{i_1}) + H(X_{i_2} | X_{i_1}) + \dots  + H(X_{i_k} | X_{i_{k-1}}, \dots, X_{i_2}, X_{i_1}) \\
    &= H(X_{i_1}) + H(X_{i_2} | X_{i_1}) + \dots + H(X_{i_k} | X_{i_{k-1}}) \\
    &= H(X_{i_1}) + P(X_{i_1}=1) H(X_{i_2} | X_{i_1} = 1) + P(X_{i_2} = 1) H(X_{i_3} | X_{i_2} = 1)  + \dots + P(X_{i_{k-1}} = 1) H(X_{i_k} | X_{i_{k-1}} = 1) \\
    &= h(\lambda^{g_1}) + \lambda^{g_1}h(\lambda^{g_2}) + \lambda^{g_1}\lambda^{g_2}h(\lambda^{g_3}) + \dots + \lambda^{g_1+g_2\dots+g_{k-1}} h(\lambda^{g^k}) \\
    &= \obj(\{g_1, g_2, \dots, g_k\})
    \end{align*}

\paragraph{Base case $j = k$. } 
Taking first derivative w.r.t. $g_k$ and setting to zero,
\begin{align*}
\diffp{\obj}{{g_k}} &= \lambda^{g_1+g_2\dots+g_{k-1}} h'(\lambda^{g_k})\lambda^{g_k}\ln \lambda  = 0 \\
\Rightarrow h'(\lambda^{g_k}) &= 0 \\
\Rightarrow \log \frac{1-\lambda^{g_k}}{\lambda^{g_k}} &= \log 1 \\
\Rightarrow \lambda^{g_k} &= \frac{1}{2} \\
\Rightarrow g^k &= \log \frac{1}{2}/\log \lambda
\end{align*}

\paragraph{Induction hypothesis. } We suppose that the claim holds for $g_{m}, g_{m+1}, \dots, g_k$ where $1< m \leq k$. We will prove for the case $g_{m-1}$.

\paragraph{Induction step. } Taking first derivative w.r.t. $g_{m-1}$,
\begin{align*}
    &\diffp{\obj}{{g_{m-1}}}  = 0 \\
    &\Rightarrow h'(\lambda^{g_{m-1}}) + h(\lambda^{g_m})
     + \lambda^{g_m}h(\lambda^{g_{m+1}}) \\
     &\quad + \lambda^{g_m+g_{m+1}} h(\lambda^{g_{m+2}}) + \dots = 0\\
    &\quad + \lambda^{g_m+\dots+g_{k-1}} h(\lambda^{g_k}) \\
    &\Rightarrow 0 = h'(\lambda^{g_{m-1}}) - \frac{k-m+1}{k-m+2}\log (k-m+1) \\
    &\quad +\log ({k-m+2})
     + \frac{k-m+1}{k-m+2}[-\frac{k-m}{k-m+1}\log(k-m) \\
     &\quad + \log(k-m+1)]
     + \frac{k - m + 1}{k-m+2}\frac{k-m}{k-m+1} \\
     &\quad [- \frac{k-m-1}{k-m} \log (k-m-1) + \log(k-m)] + \dots \\
    &\Rightarrow 0 = h'(\lambda^{g_{m-1}}) + \log (k - m +2 ) \\
    &\Rightarrow \frac{1-\lambda^{g_{m-1}}}{\lambda^{g_{m-1}}} = \frac{1}{k-m+2} \\
    &\Rightarrow \lambda^{g_{m-1}} = \frac{k-m+2}{k-m+3} \\
    &\Rightarrow g_{m-1} = \log \frac{k-m+2}{k-m+3} / \log \lambda
\end{align*}

\end{proof}

\begin{corollary}
    The maximum value of the the objective $\obj$ is $\log (k+1)$ where $k$ is the budget.
\end{corollary}
\begin{proof}
\begin{align*}
    &\obj_{max} = h(\lambda^{g_1}) + \lambda^{g_1}h(\lambda^{g_2}) + \dots + \lambda^{g_1 + g_2 +\dots+g_{k-1}}h(\lambda^{g^k}) \\
    &= -\frac{k}{k+1}\log k + \log (k+1) + \frac{k}{k+1} ( - \frac{k-1}{k}\log (k-1) + \log k) \\
    &\quad+ \frac{k}{k+1}\frac{k-1}{k}(-\frac{k-2}{k-1}\log (k-2) + \log (k-1)) + \dots \\
    &= \log (k+1)
\end{align*}
\end{proof}

\begin{algorithm}
    \caption{PathSelection}
    \label{alg:path}
    \textbf{Input: }A path network with source at endpoint, disease parameter $\p$, budget $k$ \\
    \textbf{Output: }Optimal separation of sensors $(g_1, g_2, \dots, g_k)$
    \begin{algorithmic}[1]
    \State Initialize $G$
        \For{each $i = 1$ to $k$}
            \State $j^* \leftarrow \log \frac{k}{k+1}/\log \p$
            \State $j_{floor} \leftarrow \max(1, \lfloor j^*\rfloor)$
            \State $j_{ceil} \leftarrow \lceil j^*\rceil$
            \If{$|\p^{j_{floor}} - \frac{k}{k+1} | < |\p^{j_{ceil}} - \frac{k}{k+1}|$}
                \State $G[i] \leftarrow j_{floor}$
            \Else
                \State $G[i] \leftarrow j_{ceil}$
            \EndIf
        \EndFor
        \State \textbf{return} $G$
    \end{algorithmic}
\end{algorithm}

\noindent
\textbf{Experimental verification of Theorem~\ref{lem:path}. }
For various disease probabilities $\p$, we conduct an integer grid search over $[1,2,\dots]$ for the optimal separations (denoted by \texttt{Grid}) and compare it to the solution given by Theorem~\ref{lem:path} (denoted by \texttt{Theorem}) in Table \ref{tab:path_exp1}. We observe that rounding a fractional solution to the nearest integers is often optimal.

\begin{table}[h]
    \centering
    \begin{tabular}{|c|c|c|c|c|c|c|}
    \hline
         \multicolumn{1}{|c|}{$k$} & \multicolumn{2}{|c|}{$\lambda = $ 0.5} & \multicolumn{2}{|c|}{$\p=0.7$} & \multicolumn{2}{|c|}{$\p=0.8$}  \\
         \hline
         & \texttt{Grid} & \texttt{Theorem} & \texttt{Grid} & \texttt{Theorem} & \texttt{Grid} & \texttt{Theorem}\\
         \hline
         1 & 1 & 1 & 2 & 1.94 & 3 & 3.1 \\
         \hline
         2 & (1,1) & (0.58, 1) & (1,2) & (1.14, 1.94) & (2,3) & (1.82, 3.1) \\
         \hline
         3 & (1,1,1) & (0.42, 0.58, 1) & (1,1,2) & (0.8, 1.14, 1.94) & (1,2,3) & (1.29, 1.82, 3.1)  \\
         \hline
         4 & (1,1,1,1) & (0.32, 0.42, 0.58, 1) & (1,1,1,2) & (0.63, 0.8, 1.14, 1.94) & (1,1,2,3) & (1, 1.29, 1.82, 3.1) \\
         \hline
    \end{tabular}
    \begin{tabular}{|c|c|c|c|c|}
    \hline
         \multicolumn{1}{|c|}{$k$} & \multicolumn{2}{|c|}{$\lambda = 0.9$} & \multicolumn{2}{|c|}{$\p=0.95$}  \\
         \hline
         & \texttt{Grid} & \texttt{Theorem} & \texttt{Grid} & \texttt{Theorem} \\
         \hline
         1 & 7 & 6.58 & 14 & 13.51  \\
         \hline
         2 & (4,7) & (3.85, 6.58) & (8,14) & (7.9, 13.51)\\
         \hline
         3 & (3,4,7) & (2.73, 3.85, 6.58)  & (6,8,14) & (5.61, 7.9, 13.51) \\
         \hline
         4 & (2,3,4,7) & (2.11, 2.73, 3.84, 6.58) & (4,6,8,14) & (4.35, 5.61, 7.9, 13.51) \\
         \hline
    \end{tabular}
    \caption{Experimental verification of Theorem \ref{lem:path}}
    \label{tab:path_exp1}
\end{table}

\section*{Additional details for Section 6: Experiments}

\begin{figure}
    \centering
    \begin{subfigure}{0.5\columnwidth}
    \includegraphics[width=0.48\linewidth]{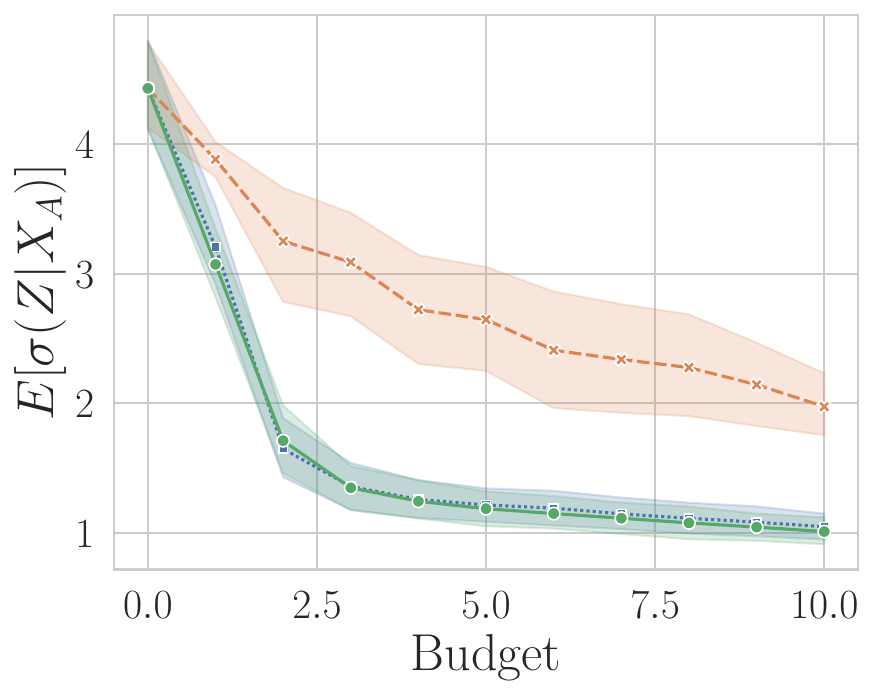}
    \includegraphics[width=0.48\linewidth]{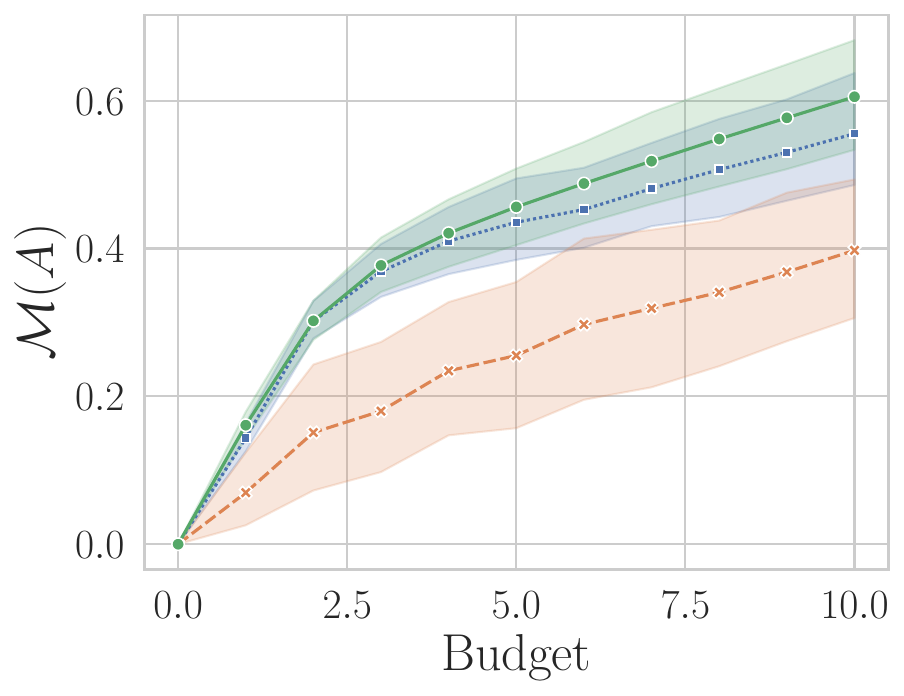}
    \caption{$IC(0.1,2)$}
    \end{subfigure}
    \begin{subfigure}{0.48\columnwidth}
        \includegraphics[width=0.48\linewidth]{performance_networkcl_n1000_g2.5_seed0_regime0.1_4_fixed_source_vary_budget_std_lineplot.pdf}
        \includegraphics[width=0.48\linewidth]{performance_networkcl_n1000_g2.5_seed0_regime0.1_4_fixed_source_vary_budget_mi_lineplot.pdf}
        \caption{$IC(0.1,4)$}
    \end{subfigure}
    \begin{subfigure}{0.48\columnwidth}
         \includegraphics[width=0.48\linewidth]{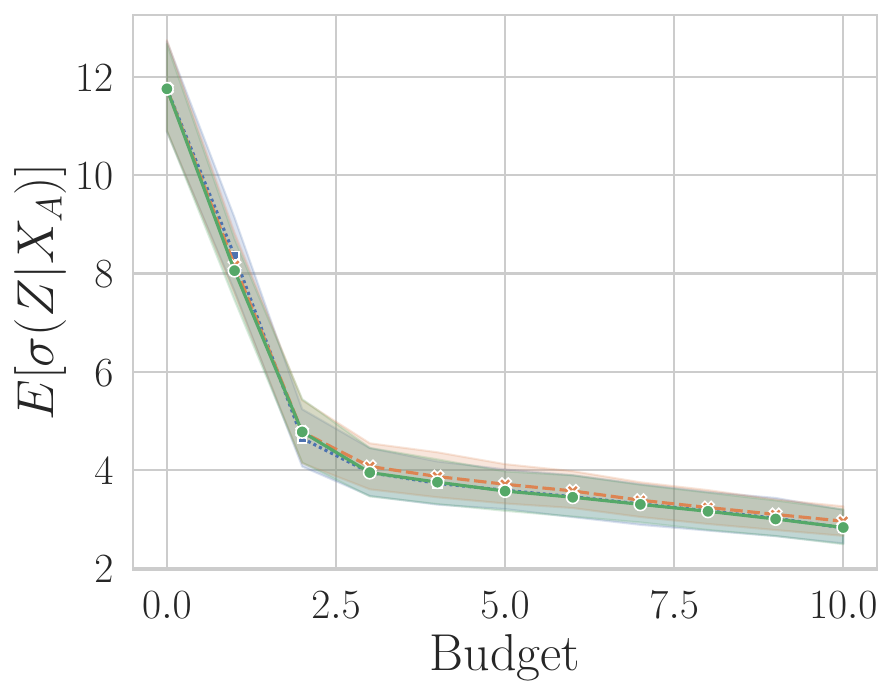}
        \includegraphics[width=0.48\linewidth]{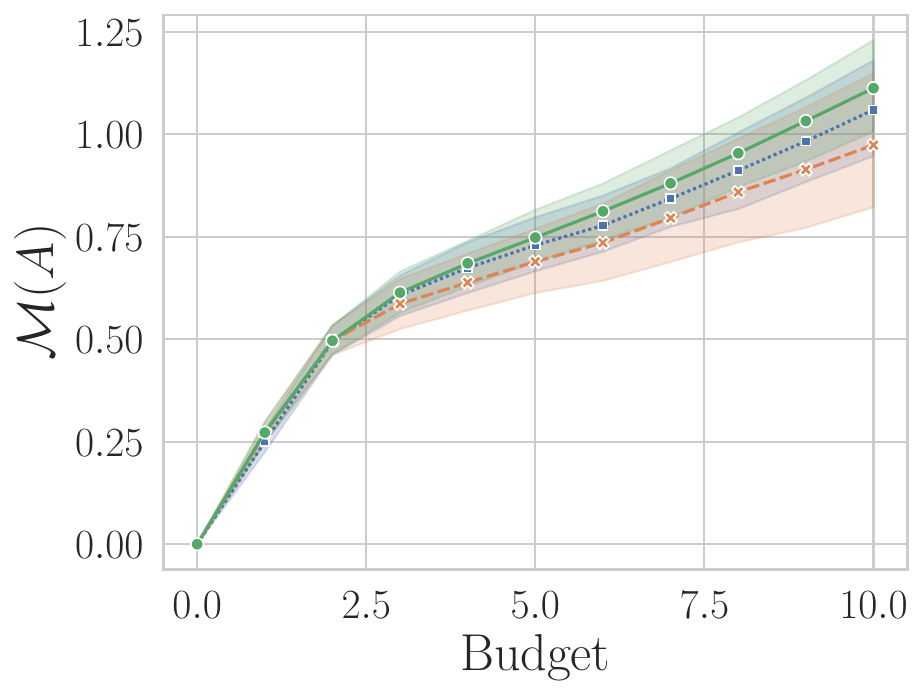}
        \caption{$IC(0.2,2)$}
    \end{subfigure}
    \begin{subfigure}{0.48\columnwidth}
        \includegraphics[width=0.48\linewidth]{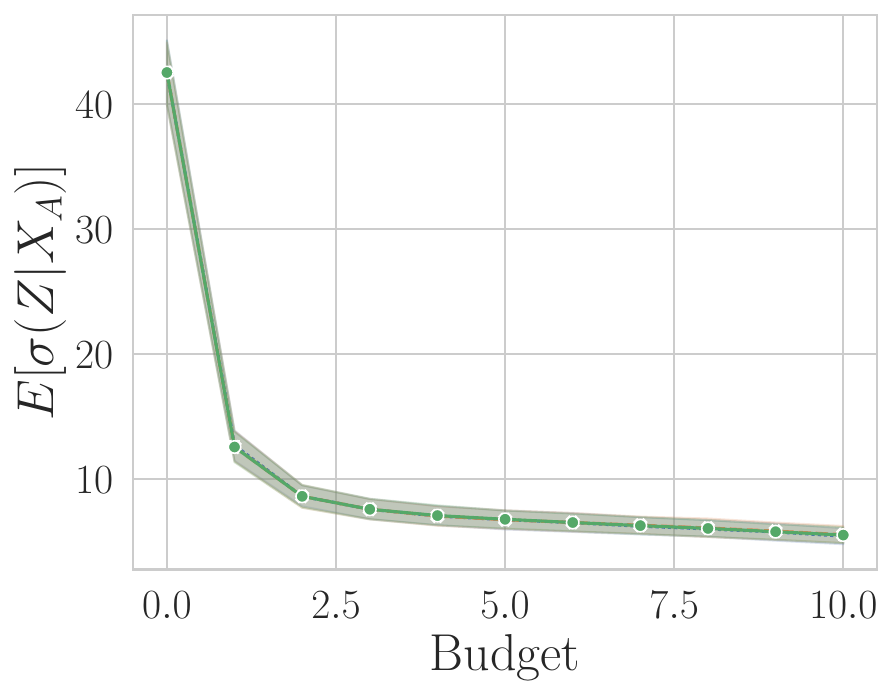}
        \includegraphics[width=0.48\linewidth]{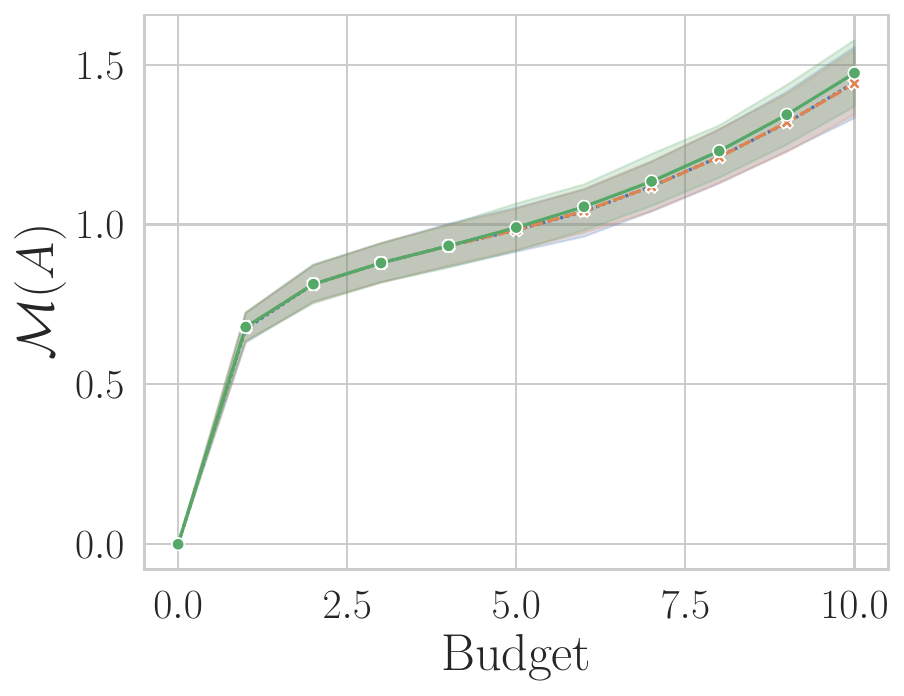}
        \caption{$IC(0.2,4)$}
    \end{subfigure}
    
    \includegraphics[width=0.5\linewidth]{methods_legend-2.pdf}
    
    \caption{Performance of \textsc{GreedyMI} vs baselines under  \texttt{known-source} seeding (averaged over 10 runs) on \pl{}}
    \label{fig:perf_greedy_fixed_addn_pl}
\end{figure}

\begin{figure}
    \centering
    \begin{subfigure}{0.48\columnwidth}
        \includegraphics[width=0.48\linewidth]{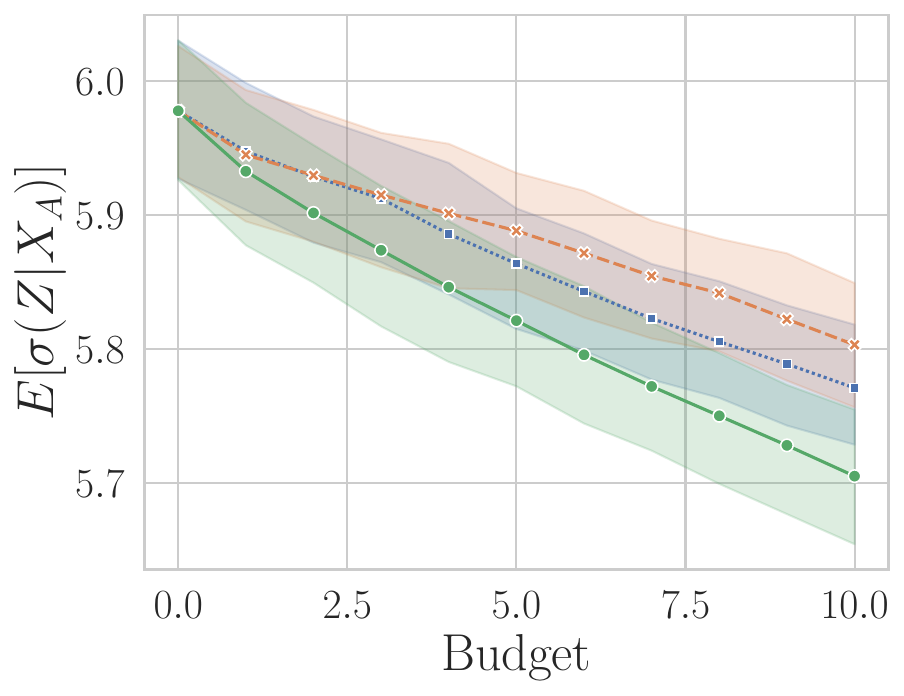}
        \includegraphics[width=0.48\linewidth]{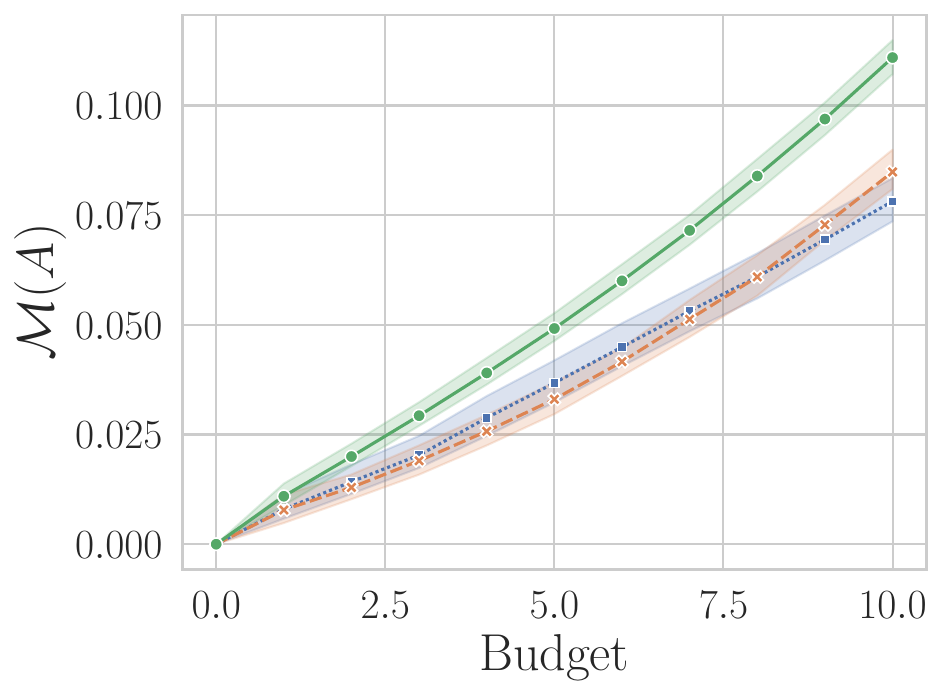}
        \caption{$IC(0.05,2)$}
    \end{subfigure}
    \begin{subfigure}{0.48\columnwidth}
        \includegraphics[width=0.48\linewidth]{performance_networker_n1000_q0.05_seed0_regime0.07_2_fixed_source_vary_budget_std_lineplot.pdf}
        \includegraphics[width=0.48\linewidth]{performance_networker_n1000_q0.05_seed0_regime0.07_2_fixed_source_vary_budget_mi_lineplot.pdf}
    \caption{$IC(0.07,2)$}
    \end{subfigure}
    
     \includegraphics[width=0.5\linewidth]{methods_legend-2.pdf}
    \caption{Performance of \textsc{GreedyMI} vs baselines under \texttt{known-source} seeding (averaged over 10 runs) on \er{}}
    \label{fig:perf_greedy_fixed_addn_er}
\end{figure}

\begin{figure}
    \centering
    \begin{subfigure}{0.48\columnwidth}
        \includegraphics[width=0.48\linewidth]{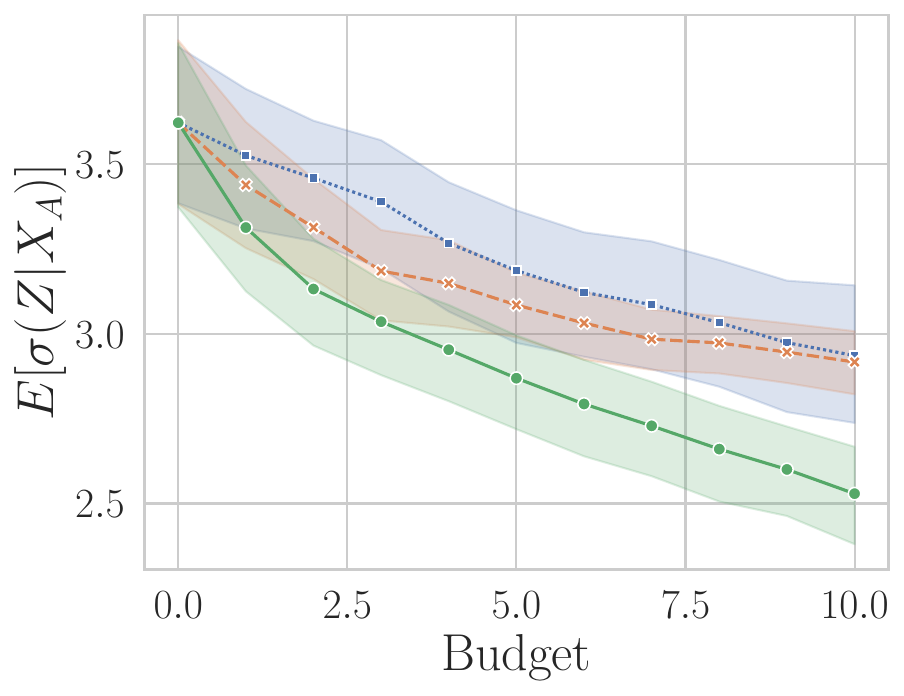}
        \includegraphics[width=0.48\linewidth]{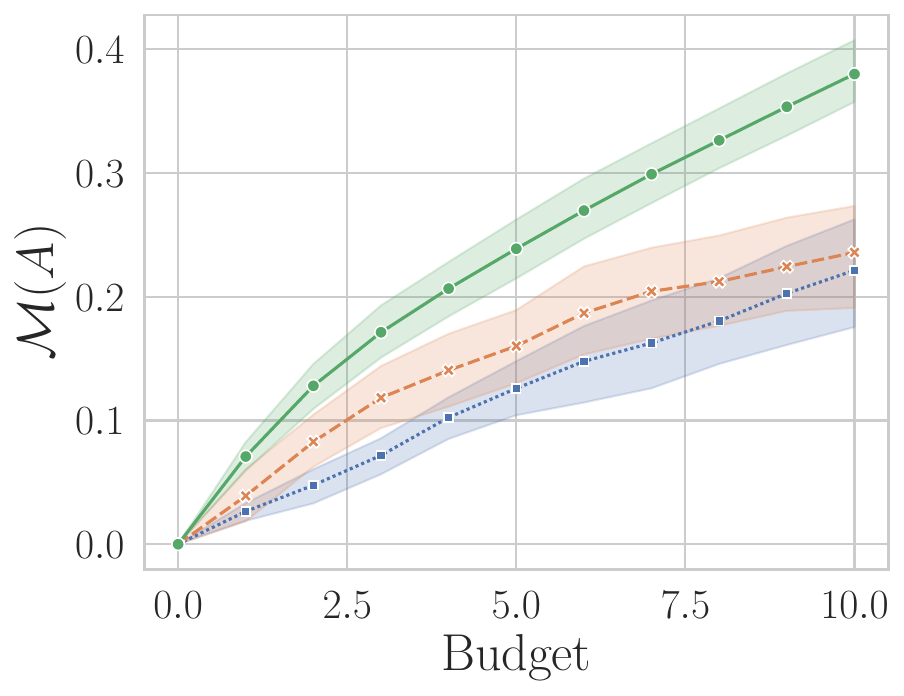}
        \caption{$IC(0.1,2$}
    \end{subfigure}
    \begin{subfigure}{0.48\columnwidth}
        \includegraphics[width=0.48\linewidth]{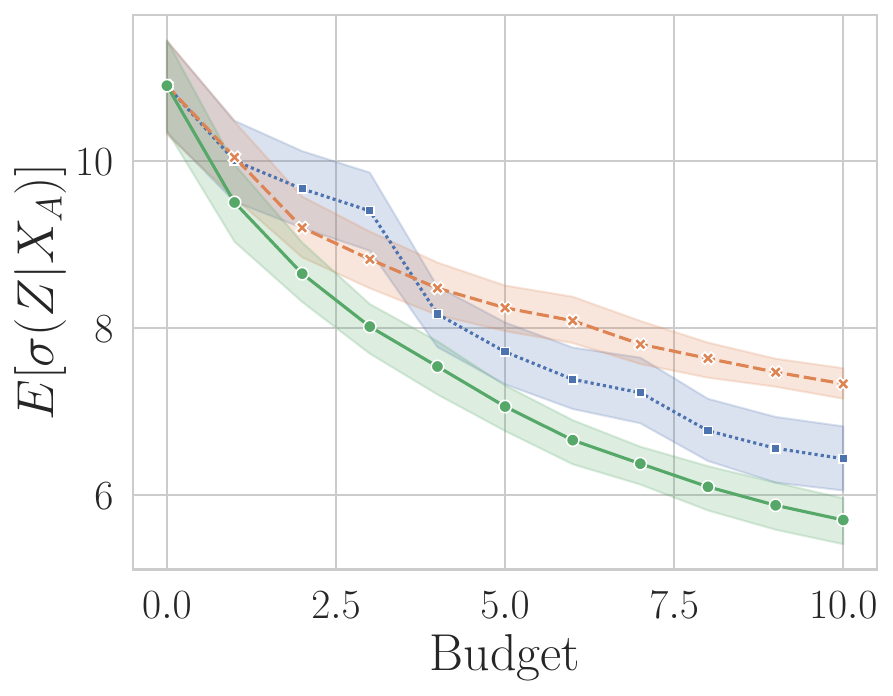}
        \includegraphics[width=0.48\linewidth]{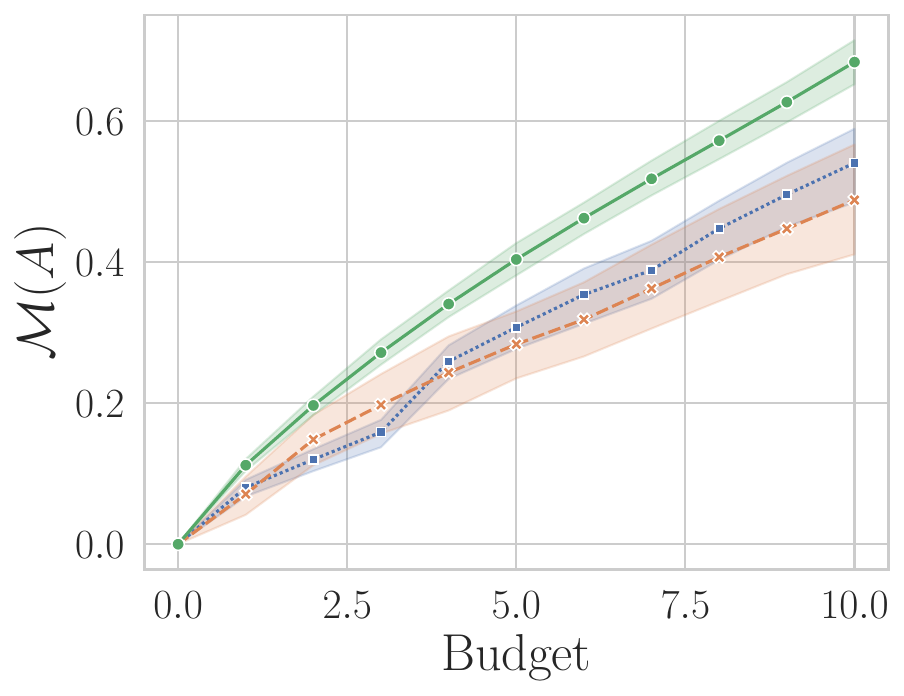}
        \caption{$IC(0.1,4)$}
    \end{subfigure}
    \begin{subfigure}{0.48\columnwidth}
        \includegraphics[width=0.48\linewidth]{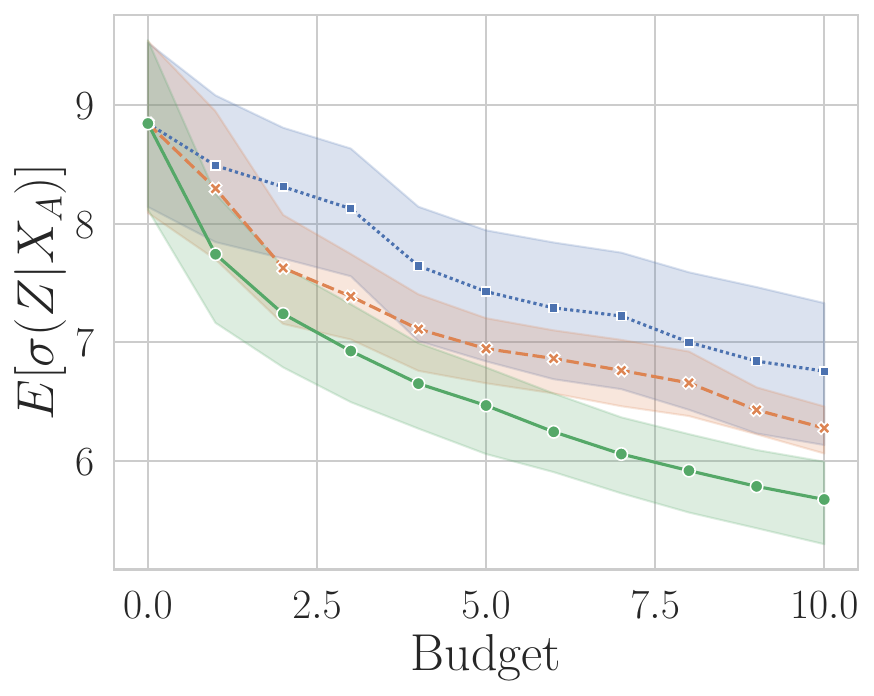}
        \includegraphics[width=0.48\linewidth]{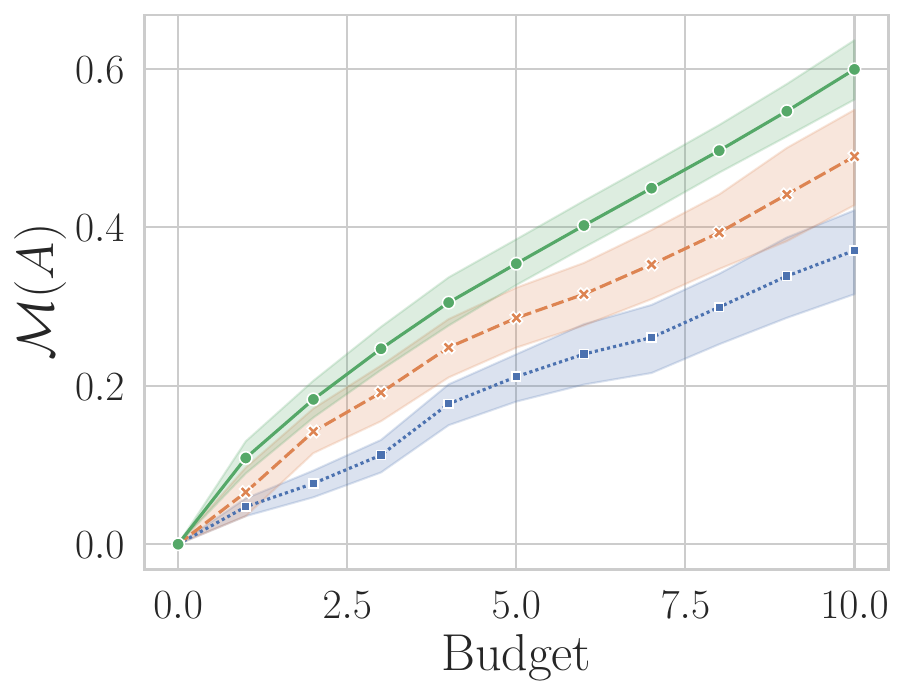}
        \caption{$IC(0.2,2)$}
    \end{subfigure}
    \begin{subfigure}{0.48\columnwidth}
        \includegraphics[width=0.48\linewidth]{performance_networkhospital_icu_contact_regime0.2_4_fixed_source_vary_budget_std_lineplot.pdf}
       \includegraphics[width=0.48\linewidth]{performance_networkhospital_icu_contact_regime0.2_4_fixed_source_vary_budget_mi_lineplot.pdf}
       \caption{$IC(0.2,4)$}
    \end{subfigure}
      
        \includegraphics[width=0.5\linewidth]{methods_legend-2.pdf}
    \caption{Performance of \textsc{GreedyMI} vs baselines under \texttt{known-source} seeding (averaged over 10 runs) on \icu{}}
    \label{fig:perf_greedy_fixed_addn_icu}
\end{figure}

\begin{figure}
    \centering
    \begin{subfigure}{0.48\columnwidth}
    \includegraphics[width=0.48\linewidth]{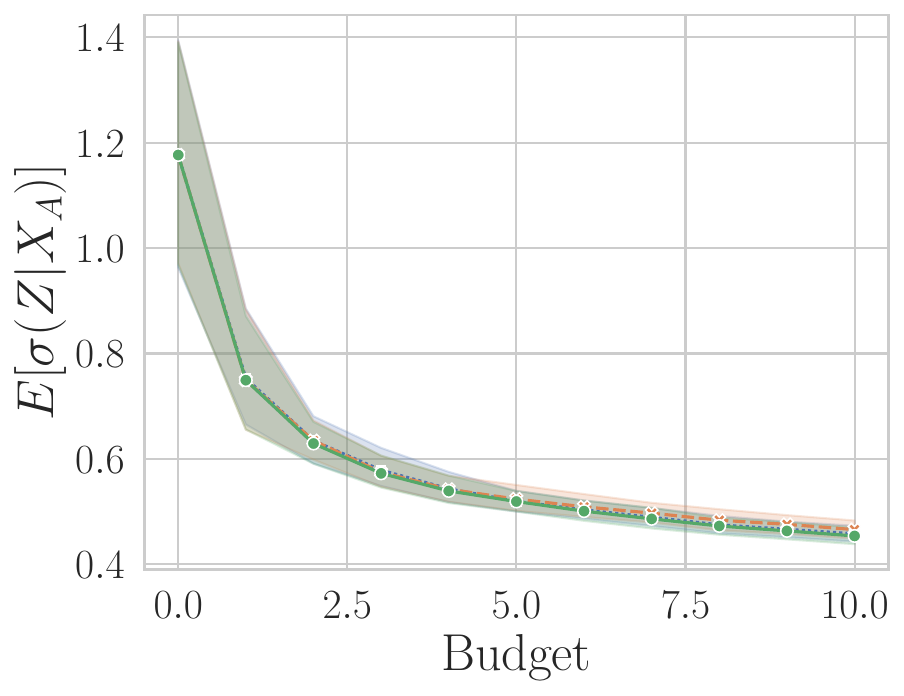}
    \includegraphics[width=0.48\linewidth]{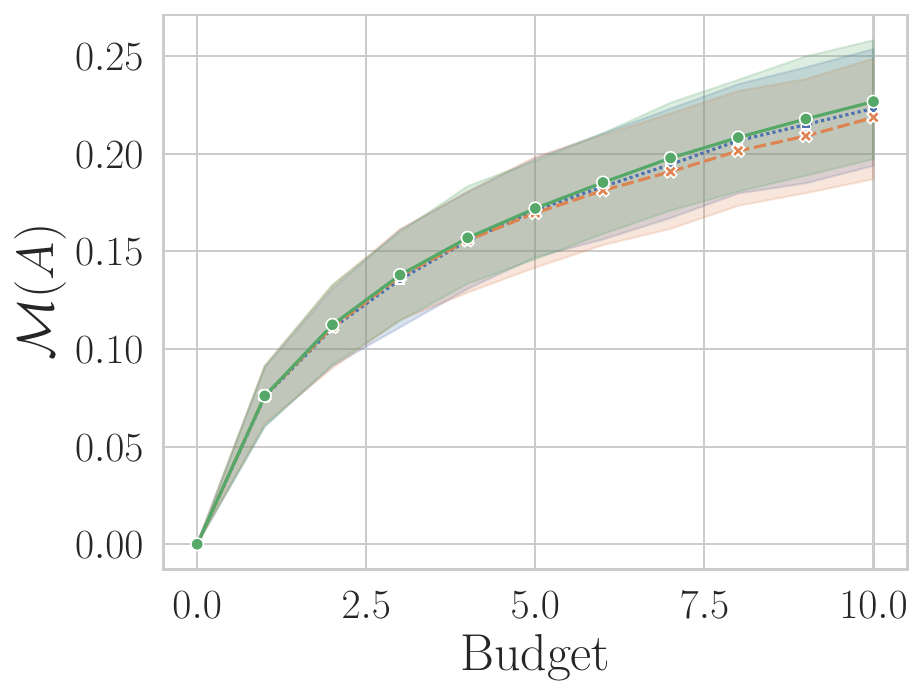}
    \caption{$IC(0.05,2)$}
    \end{subfigure}
    \begin{subfigure}{0.48\columnwidth}
    \includegraphics[width=0.48\linewidth]{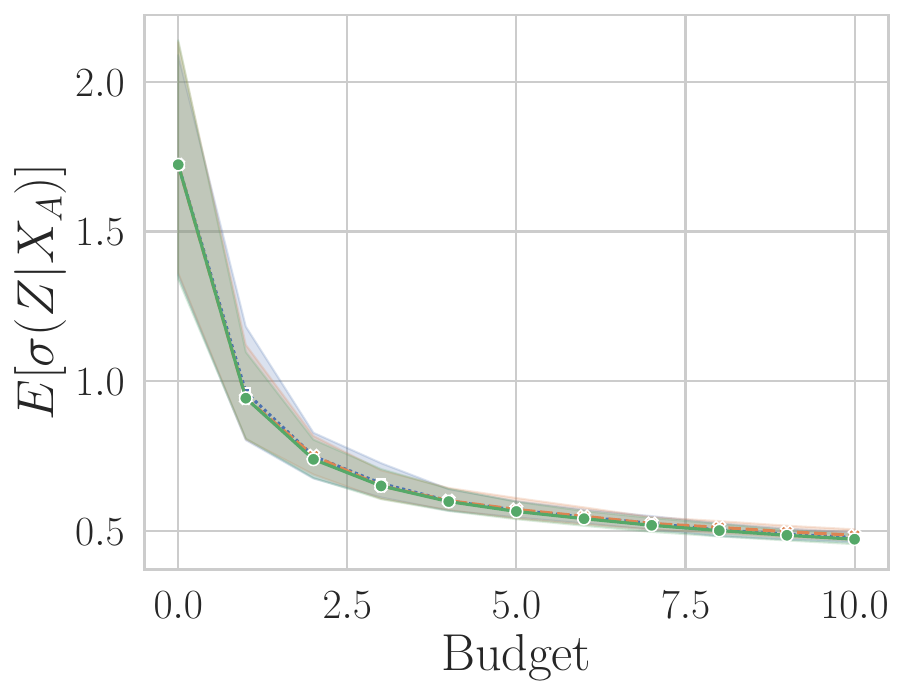}
    \includegraphics[width=0.48\linewidth]{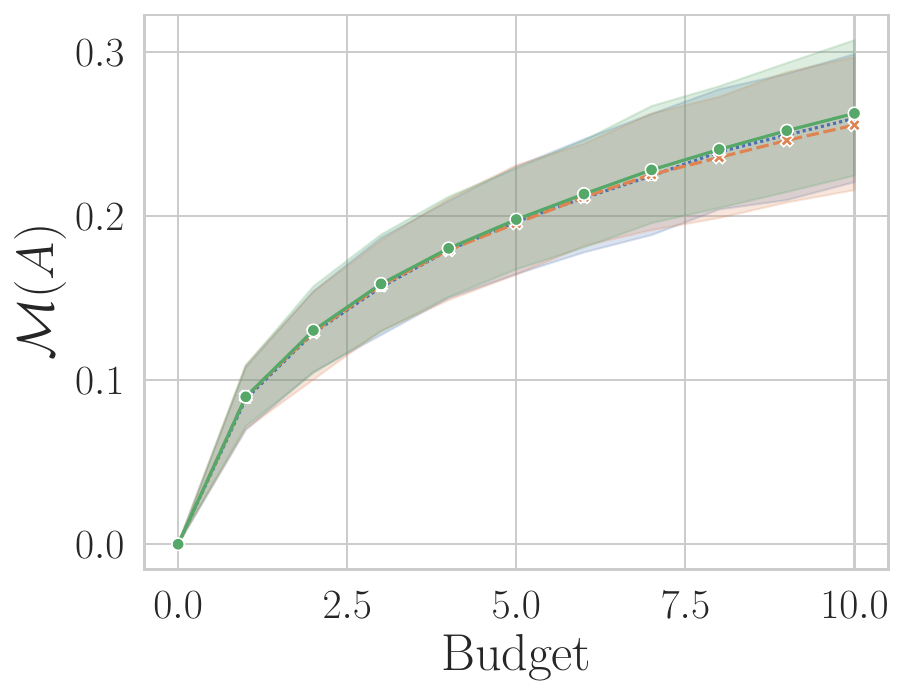}
    \caption{$IC(0.05,4)$}
    \end{subfigure}
    \begin{subfigure}{0.48\columnwidth}
    \includegraphics[width=0.48\linewidth]{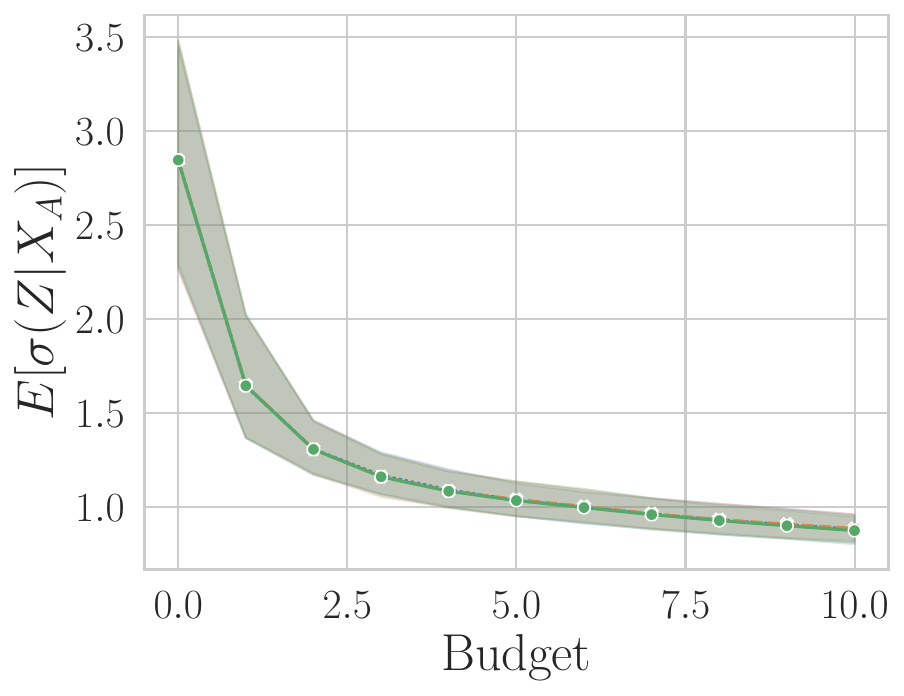}
    \includegraphics[width=0.48\linewidth]{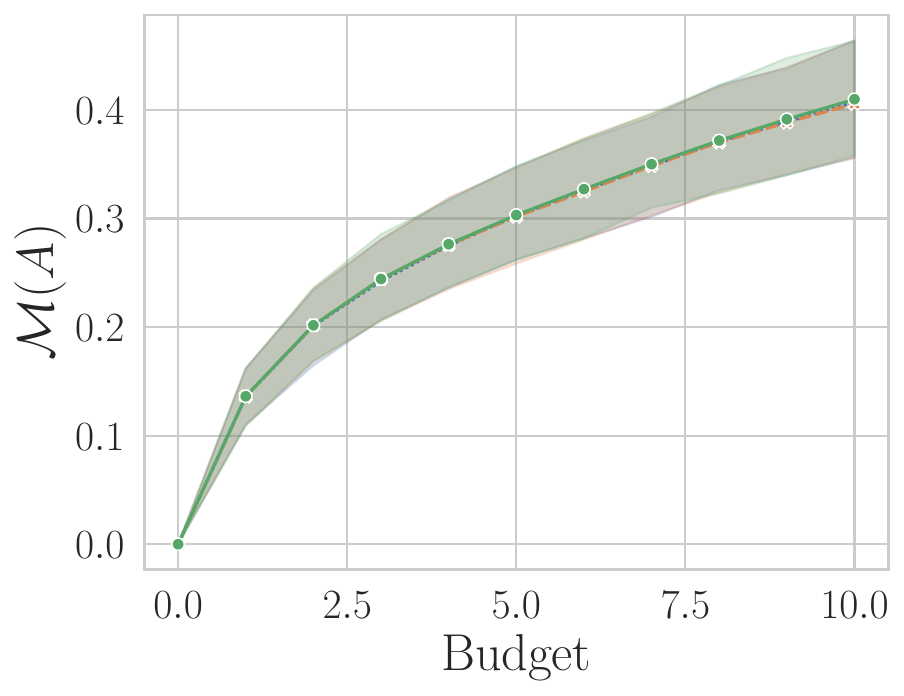}
    \caption{$IC(0.1,2)$}
    \end{subfigure}
    \begin{subfigure}{0.48\columnwidth}
    \includegraphics[width=0.48\linewidth]{performance_networkcl_n1000_g2.5_seed8_regime0.1_4_random_source_vary_budget_std_lineplot.pdf}
    \includegraphics[width=0.48\linewidth]{performance_networkcl_n1000_g2.5_seed8_regime0.1_4_random_source_vary_budget_mi_lineplot.pdf}
    \caption{$IC(0.1,4)$}
    \end{subfigure}
    \begin{subfigure}{0.48\columnwidth}
    \includegraphics[width=0.48\linewidth]{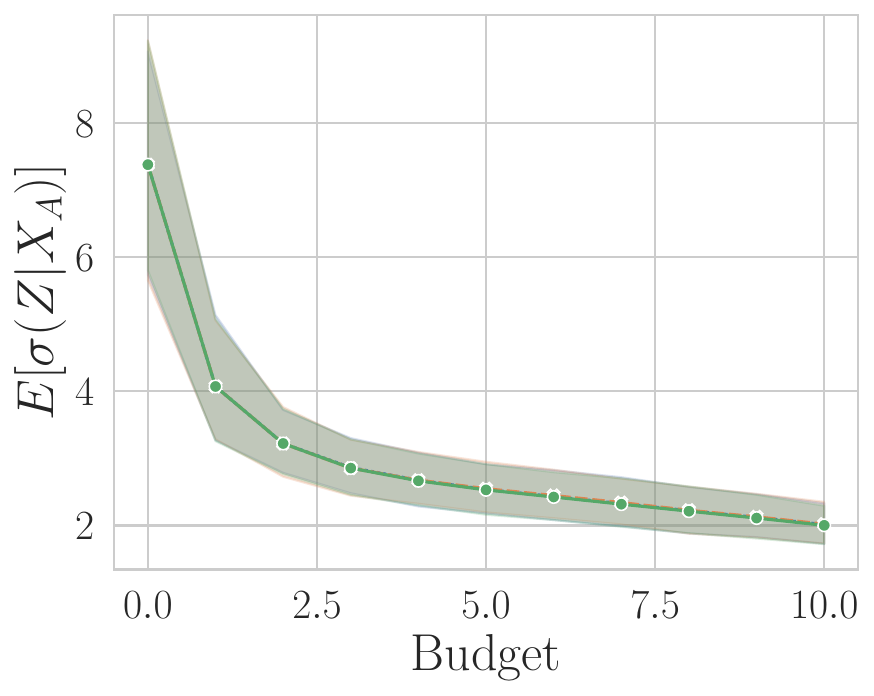}
    \includegraphics[width=0.48\linewidth]{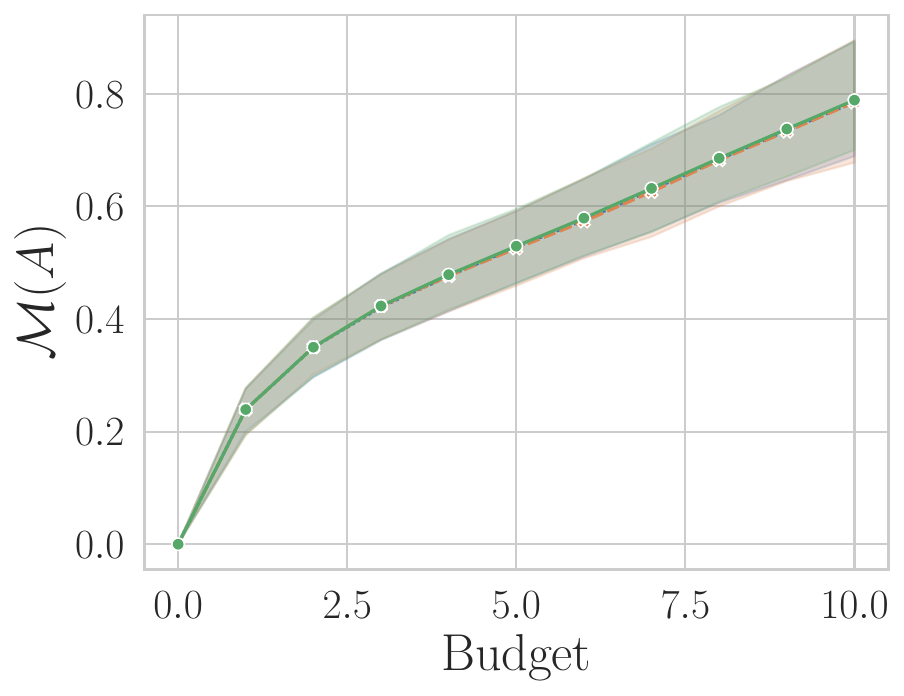}
    \caption{$IC(0.2,2)$}
    \end{subfigure}
    \begin{subfigure}{0.48\columnwidth}
    \includegraphics[width=0.48\linewidth]{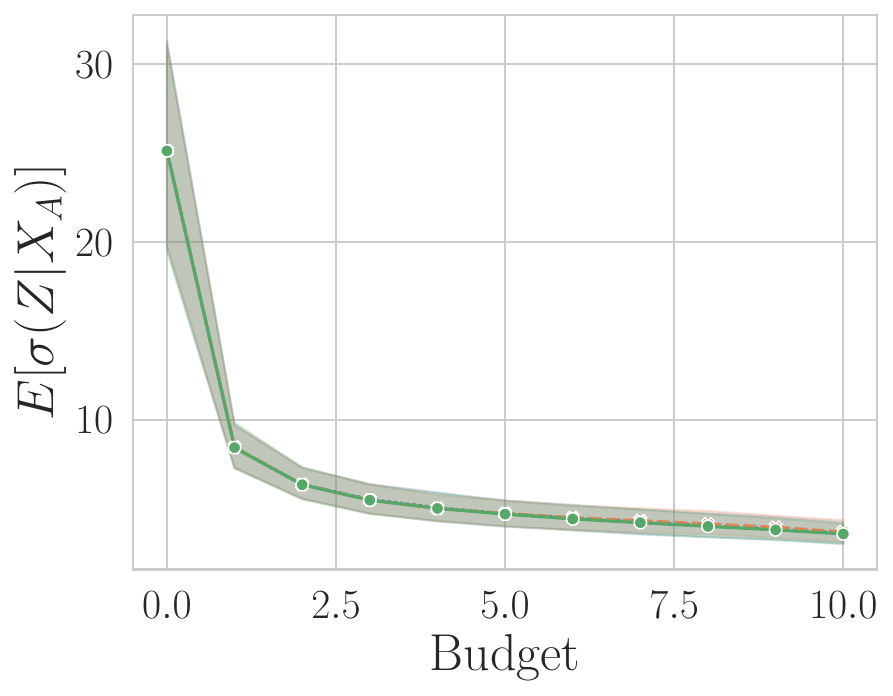}
    \includegraphics[width=0.48\linewidth]{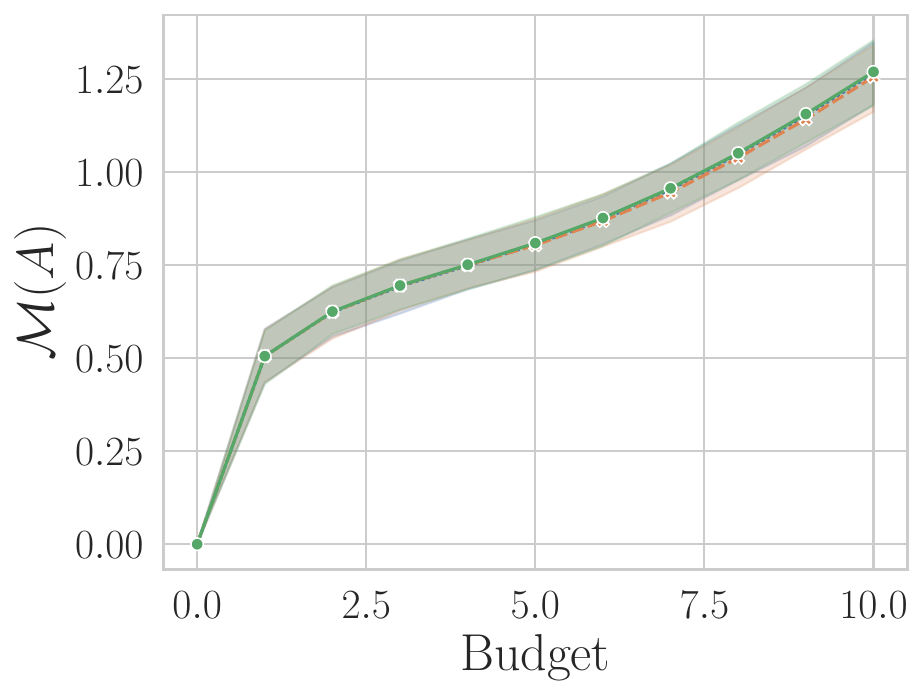}
    \caption{$IC(0.2,4)$}
    \end{subfigure}
    \includegraphics[width=0.5\linewidth]{methods_legend-2.pdf}
    \caption{Performance of \textsc{GreedyMI} vs baselines under \rsource{} seeding on \pl}
    \label{fig:perf_greedy_random_pl_addn}
\end{figure}

\begin{figure}
    \centering
    \begin{subfigure}{0.48\columnwidth}
    \includegraphics[width=0.48\linewidth]{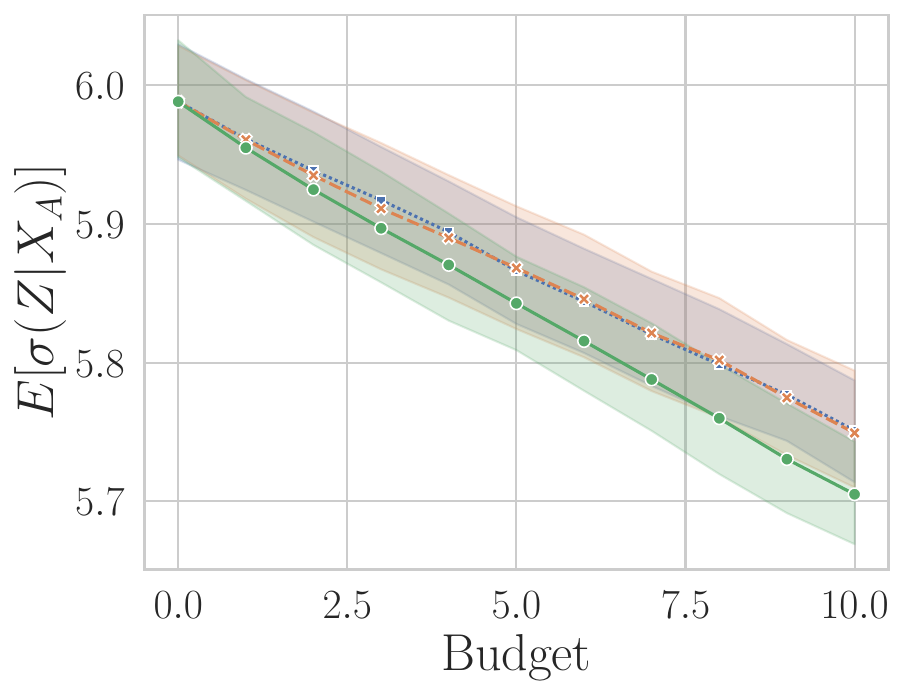}
    \includegraphics[width=0.48\linewidth]{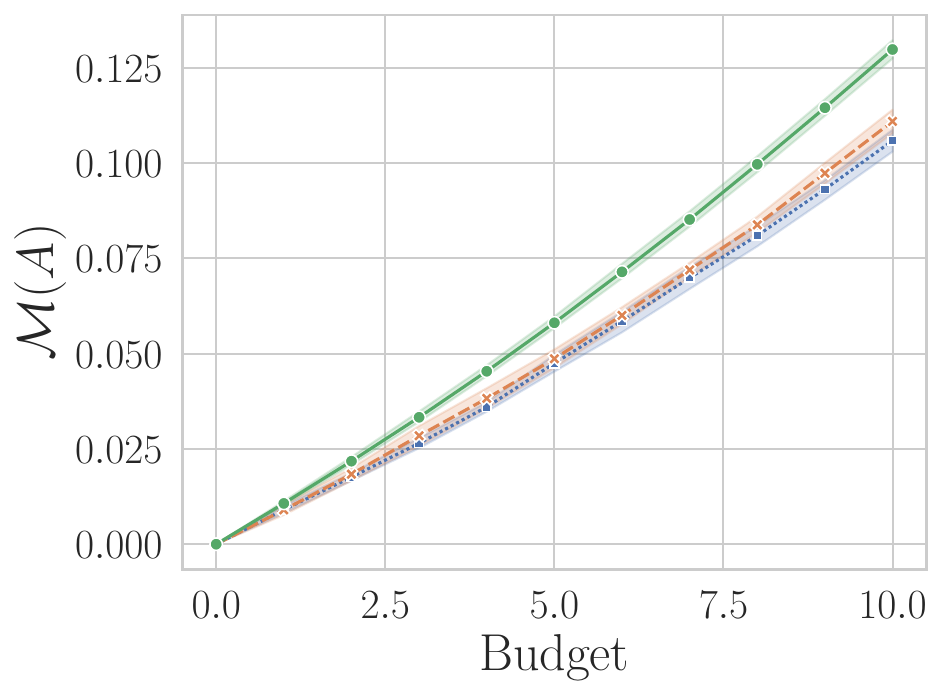}
    \caption{$IC(0.05,2)$}
    \end{subfigure}
    \begin{subfigure}{0.48\columnwidth}
        \includegraphics[width=0.48\linewidth]{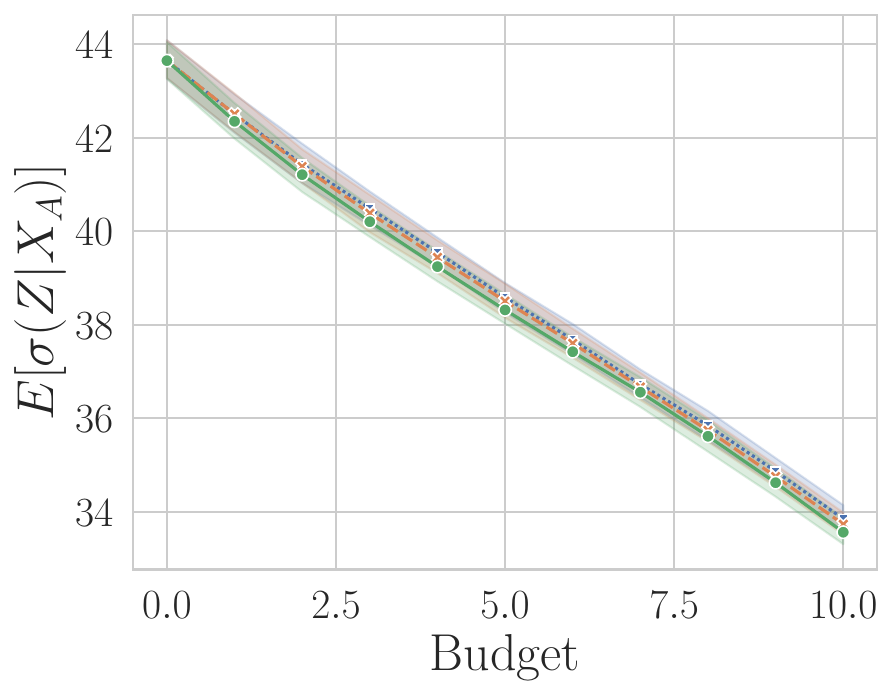}
        \includegraphics[width=0.48\linewidth]{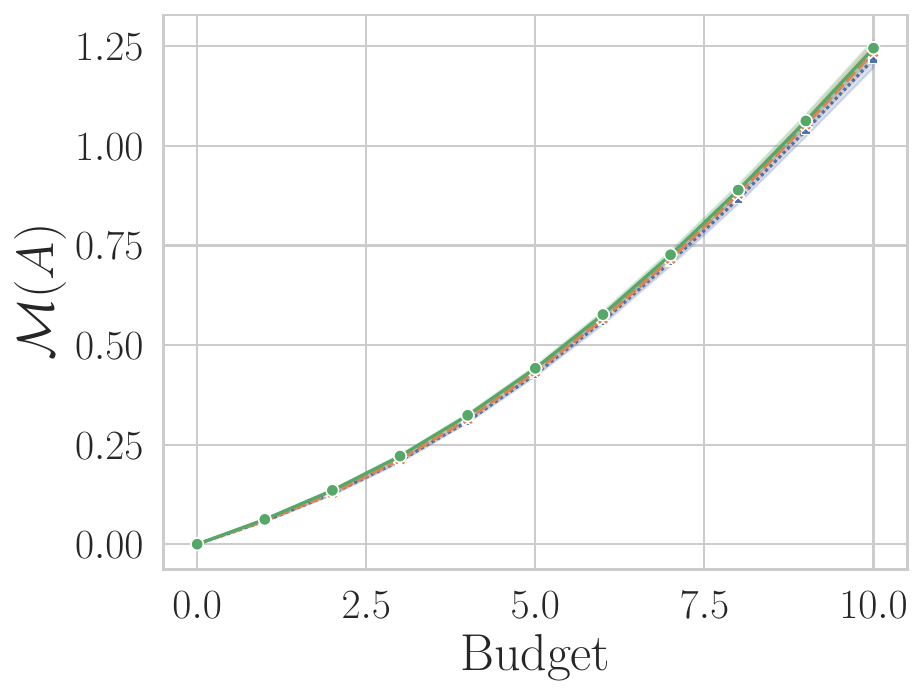}
        \caption{$IC(0.05,4)$}
    \end{subfigure}
    \begin{subfigure}{0.48\columnwidth}
    \includegraphics[width=0.48\linewidth]{performance_networker_n1000_q0.05_seed4_regime0.07_2_random_source_vary_budget_std_lineplot.pdf}
    \includegraphics[width=0.48\linewidth]{performance_networker_n1000_q0.05_seed4_regime0.07_2_random_source_vary_budget_mi_lineplot.pdf}
    \caption{$IC(0.07,2)$}
    \end{subfigure}
    \begin{subfigure}{0.48\columnwidth}
        \includegraphics[width=0.48\linewidth]{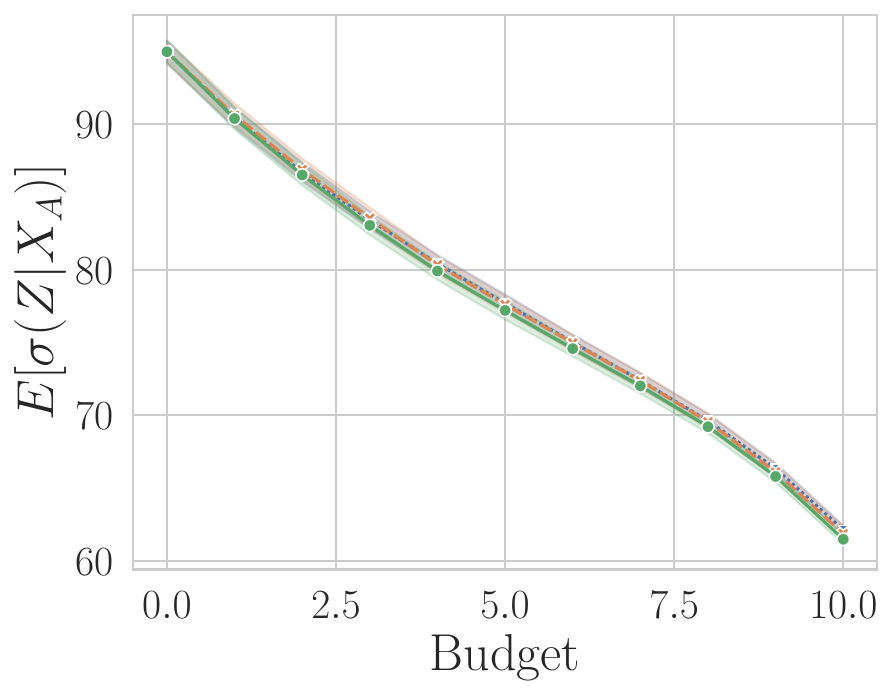}
        \includegraphics[width=0.48\linewidth]{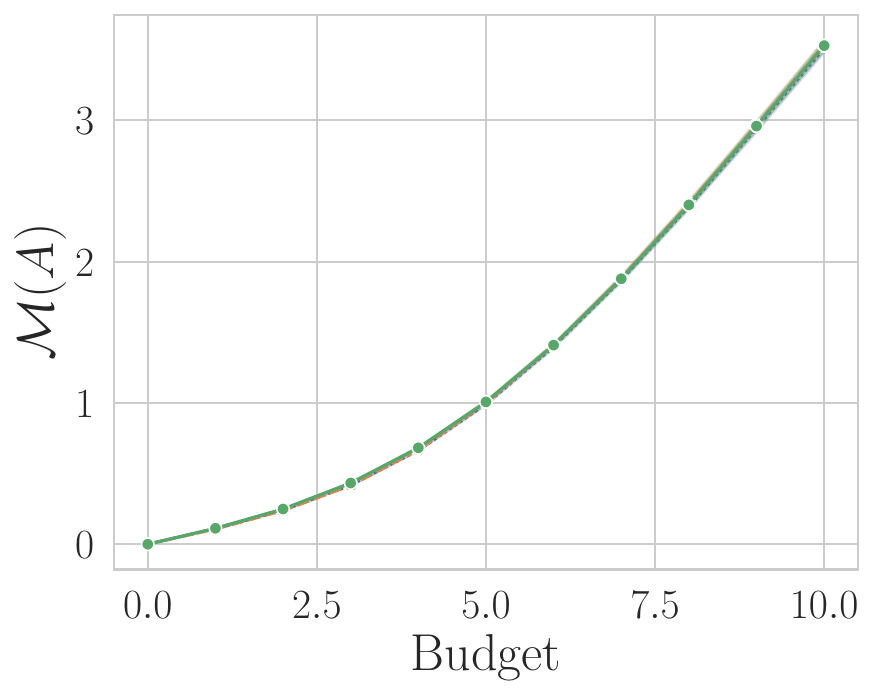}
    \caption{$IC(0.07,4)$}
    \end{subfigure}
    \caption{Performance of \textsc{GreedyMI} vs baselines under \rsource{} seeding on \er}
    \label{fig:perf_greedy_random_er_addn}
\end{figure}
\begin{figure}
    \centering
    \begin{subfigure}{0.48\columnwidth}
    \includegraphics[width=0.48\linewidth]{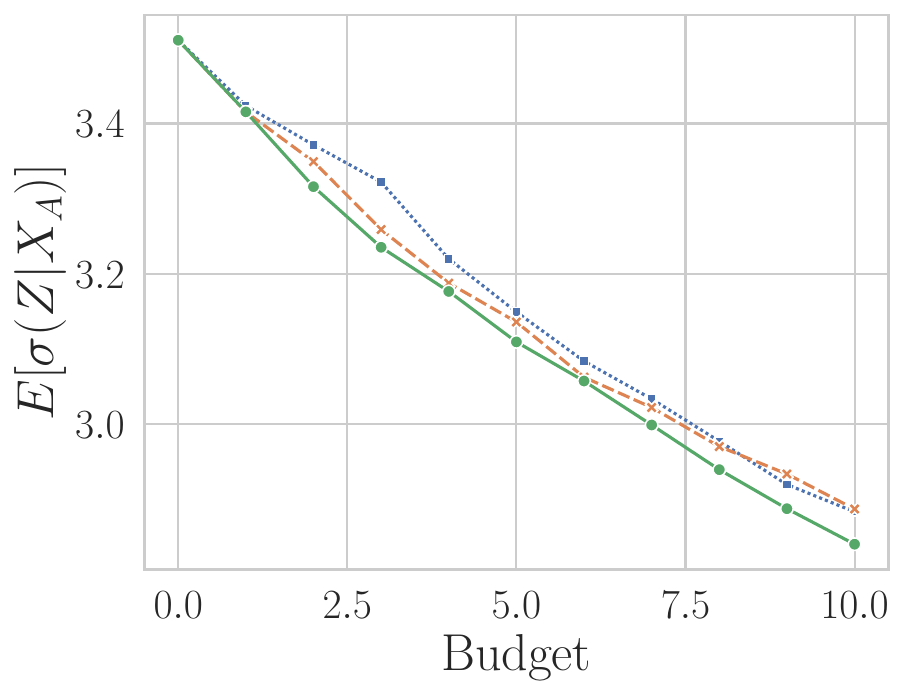}
    \includegraphics[width=0.48\linewidth]{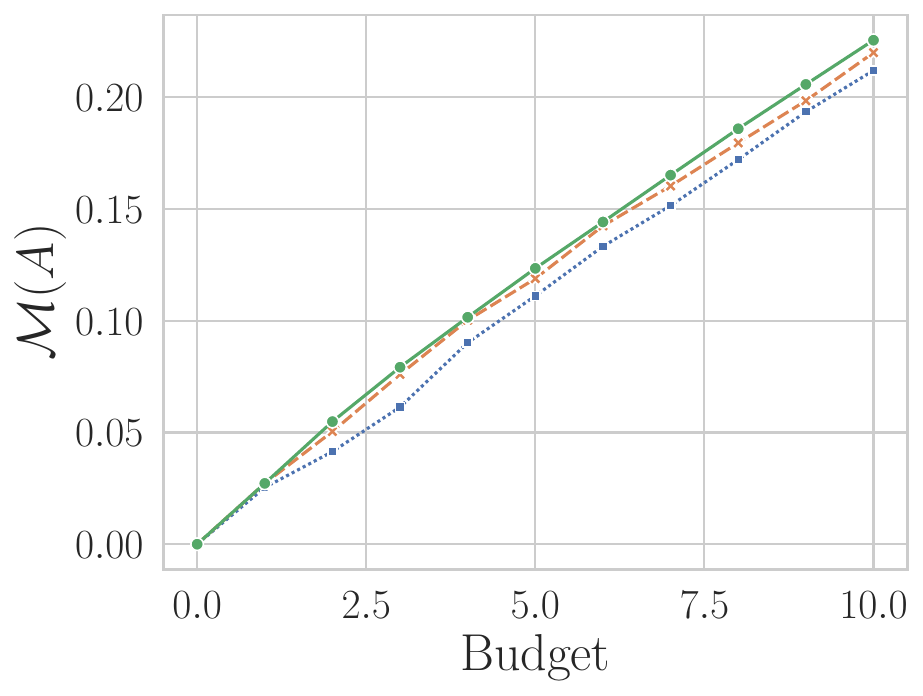}
    \caption{$IC(0.1,2)$}
    \end{subfigure}
    \begin{subfigure}{0.48\columnwidth}
    \includegraphics[width=0.48\linewidth]{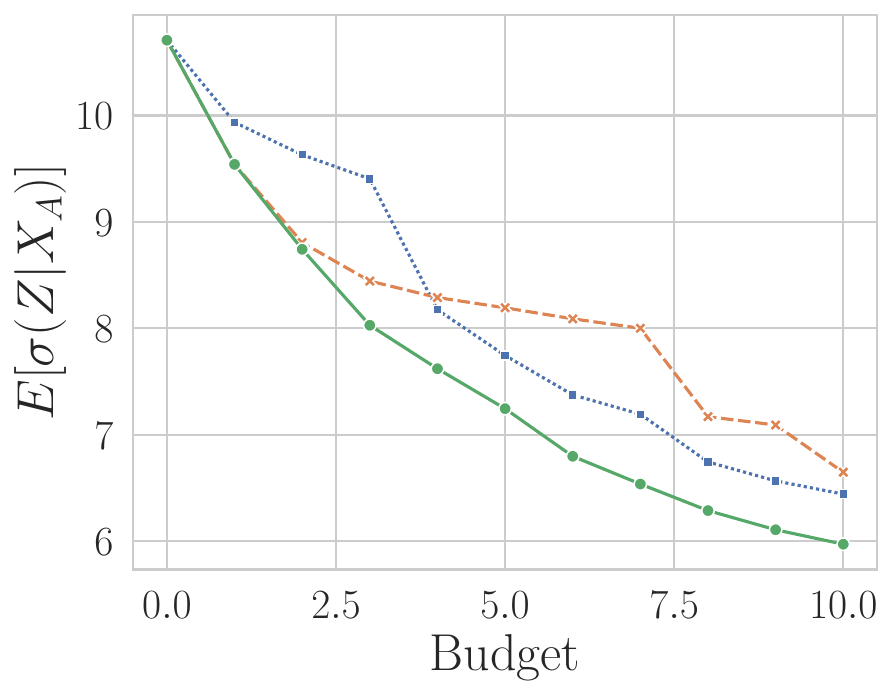}
    \includegraphics[width=0.48\linewidth]{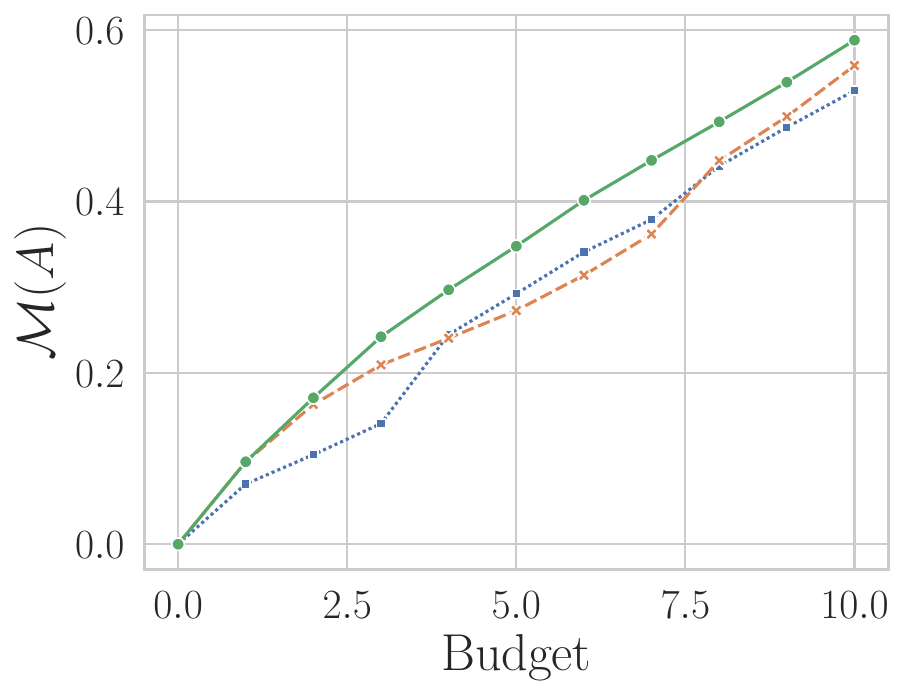}
    \caption{$IC(0.1,4)$}
    \end{subfigure}
    \begin{subfigure}{0.48\columnwidth}
    \includegraphics[width=0.48\linewidth]{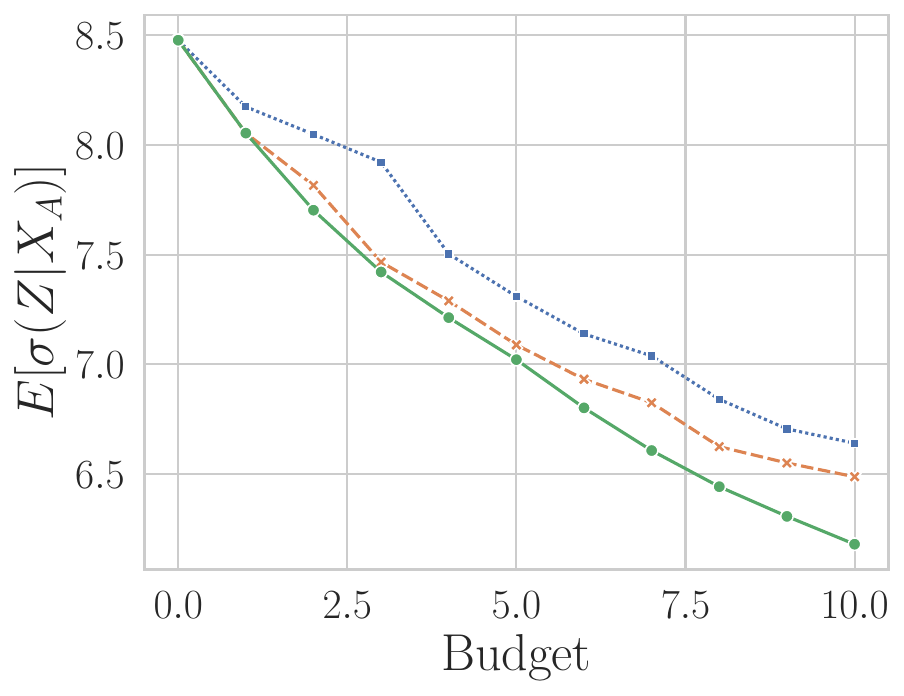}
    \includegraphics[width=0.48\linewidth]{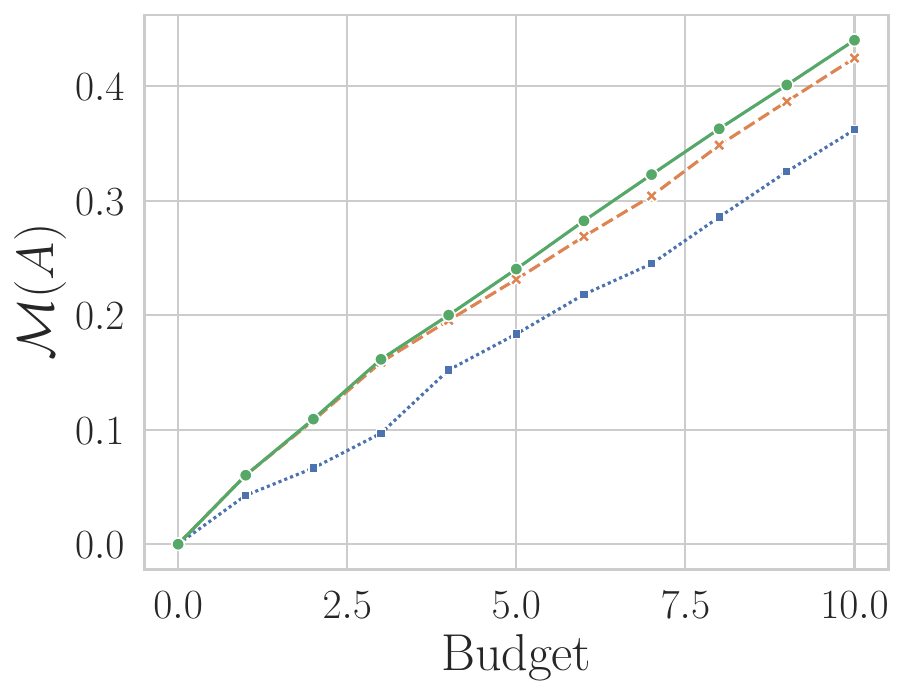}
    \caption{$IC(0.2,2)$}
    \end{subfigure}
    \begin{subfigure}{0.48\columnwidth}
    \includegraphics[width=0.48\linewidth]{performance_networkhospital_icu_contact_regime0.2_4_random_source_vary_budget_std_lineplot.pdf}
    \includegraphics[width=0.48\linewidth]{performance_networkhospital_icu_contact_regime0.2_4_random_source_vary_budget_mi_lineplot.pdf}
    \caption{$IC(0.2,4)$}
    \end{subfigure}
    \caption{Performance of \textsc{GreedyMI} vs baselines under \rsource{} seeding on \icu}
    \label{fig:perf_greedy_random_icu_addn}
\end{figure}

\begin{figure}
    \centering
    \begin{subfigure}{0.24\columnwidth}
    \includegraphics[width=0.99\linewidth]{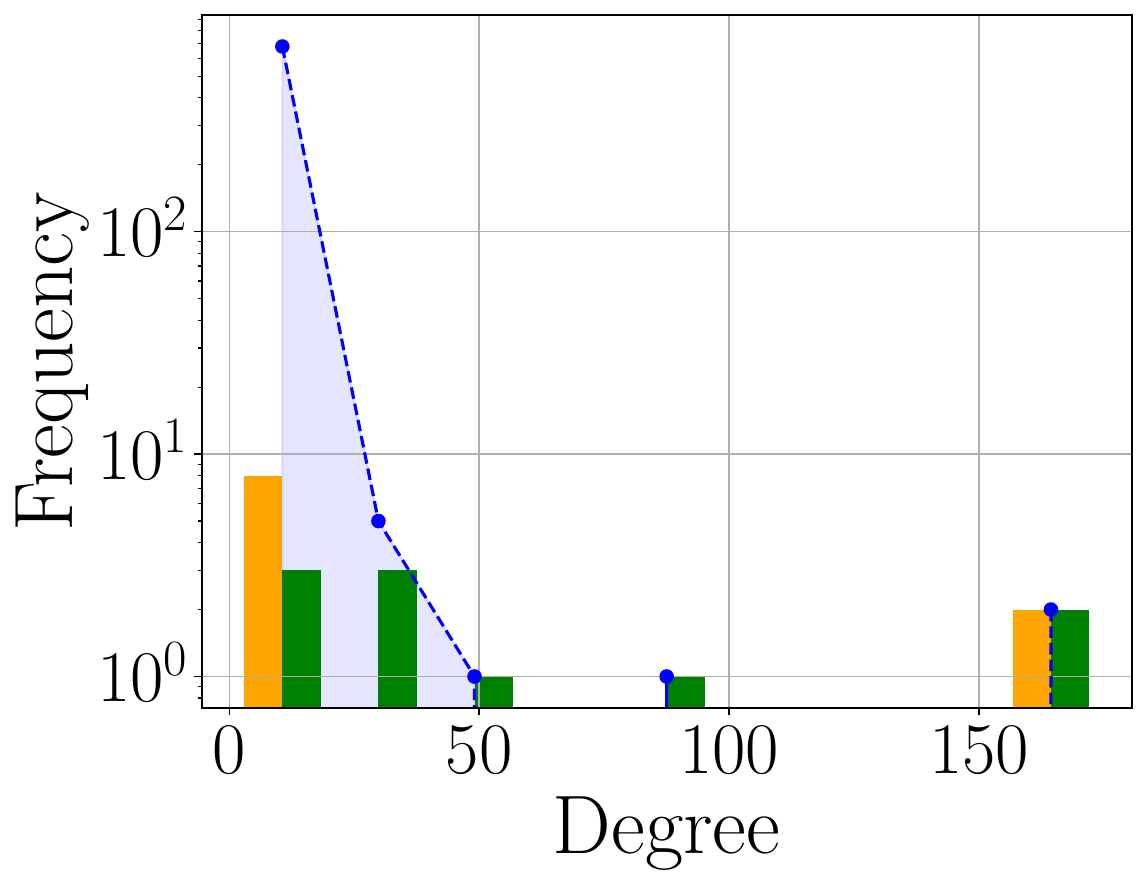}
    \caption{$IC(0.1,2)$}
    \end{subfigure}
     \begin{subfigure}{0.24\columnwidth}
    \includegraphics[width=0.99\linewidth]{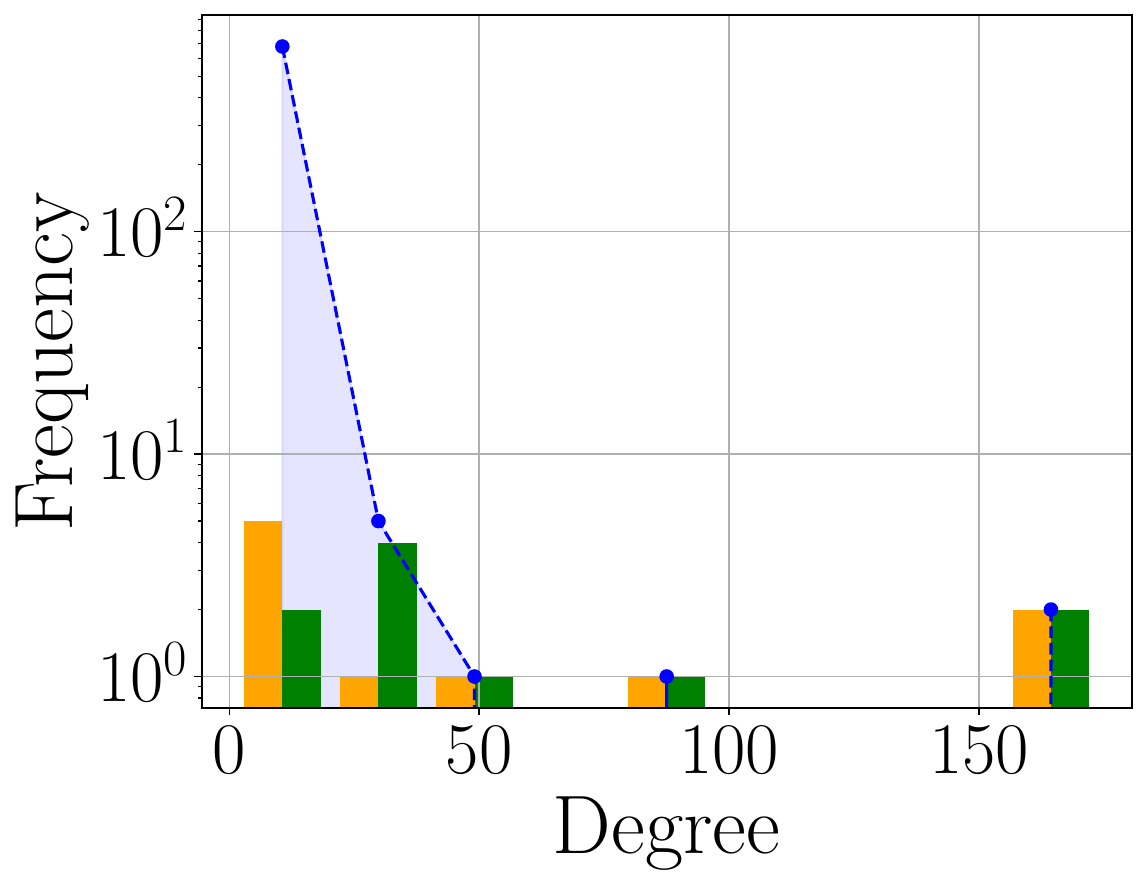}
    \caption{$IC(0.1,4)$}
    \end{subfigure}
     \begin{subfigure}{0.24\columnwidth}
    \includegraphics[width=0.99\linewidth]{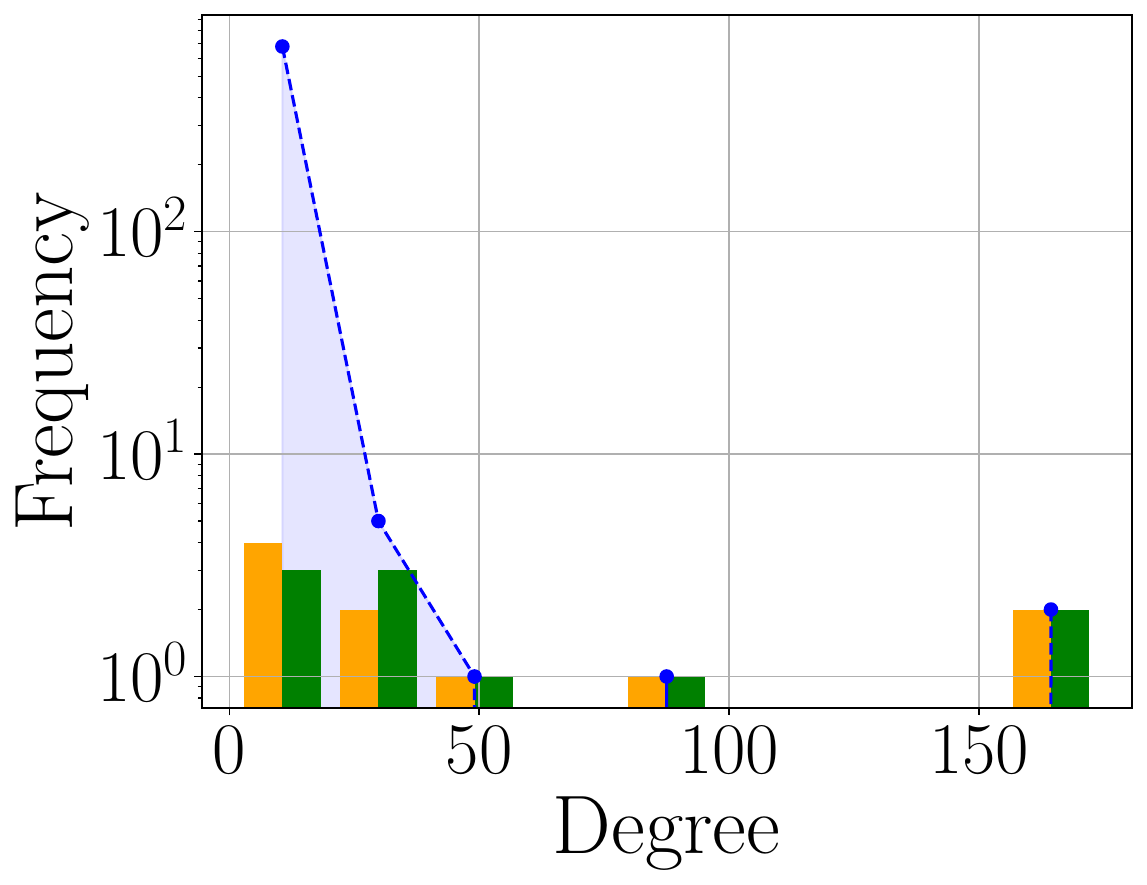}
    \caption{$IC(0.2,2)$}
    \end{subfigure}
     \begin{subfigure}{0.24\columnwidth}
    \includegraphics[width=0.99\linewidth]{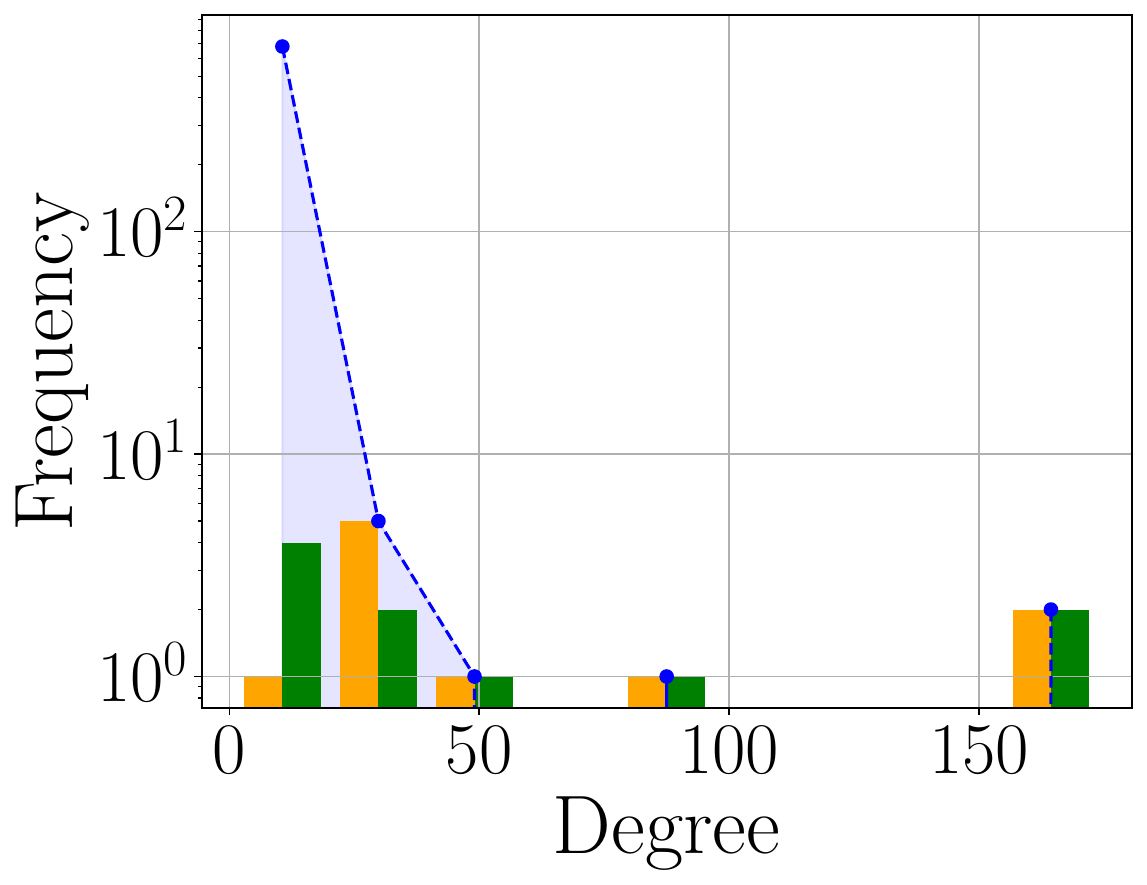}
    \caption{$IC(0.2,4)$}
    \end{subfigure}
    \includegraphics[width=0.5\linewidth]{degrees_legend.pdf}
    \caption{Comparison of the degree distribution of \textsc{GreedyMI} and \textsc{Vulnerable} under \texttt{known-source} seeding on \pl{}}
    \label{fig:degree_dist_addn_fixed_pl}
\end{figure}
\begin{figure}
    \centering
    \begin{subfigure}{0.28\columnwidth}
    \includegraphics[width=0.99\linewidth]{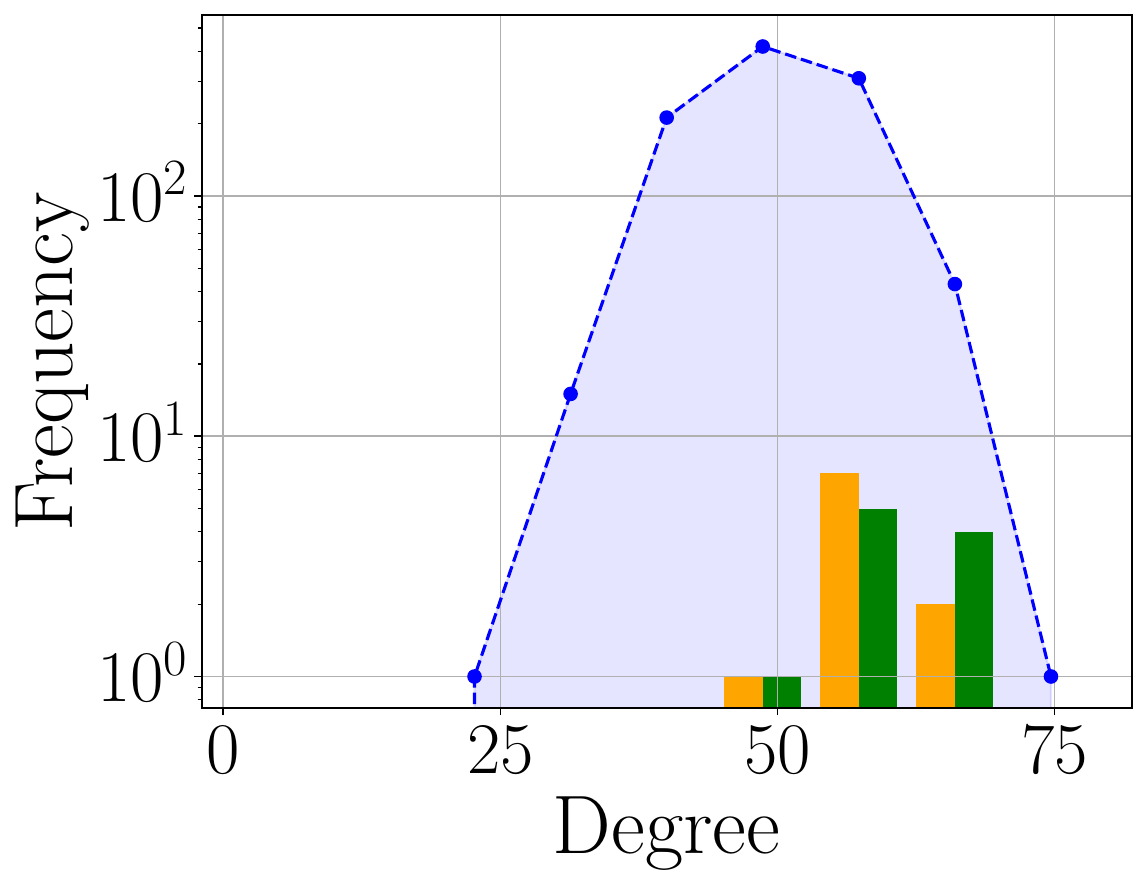}
    \caption{$IC(0.05,2)$}
    \end{subfigure}
     \begin{subfigure}{0.28\columnwidth}
    \includegraphics[width=0.99\linewidth]{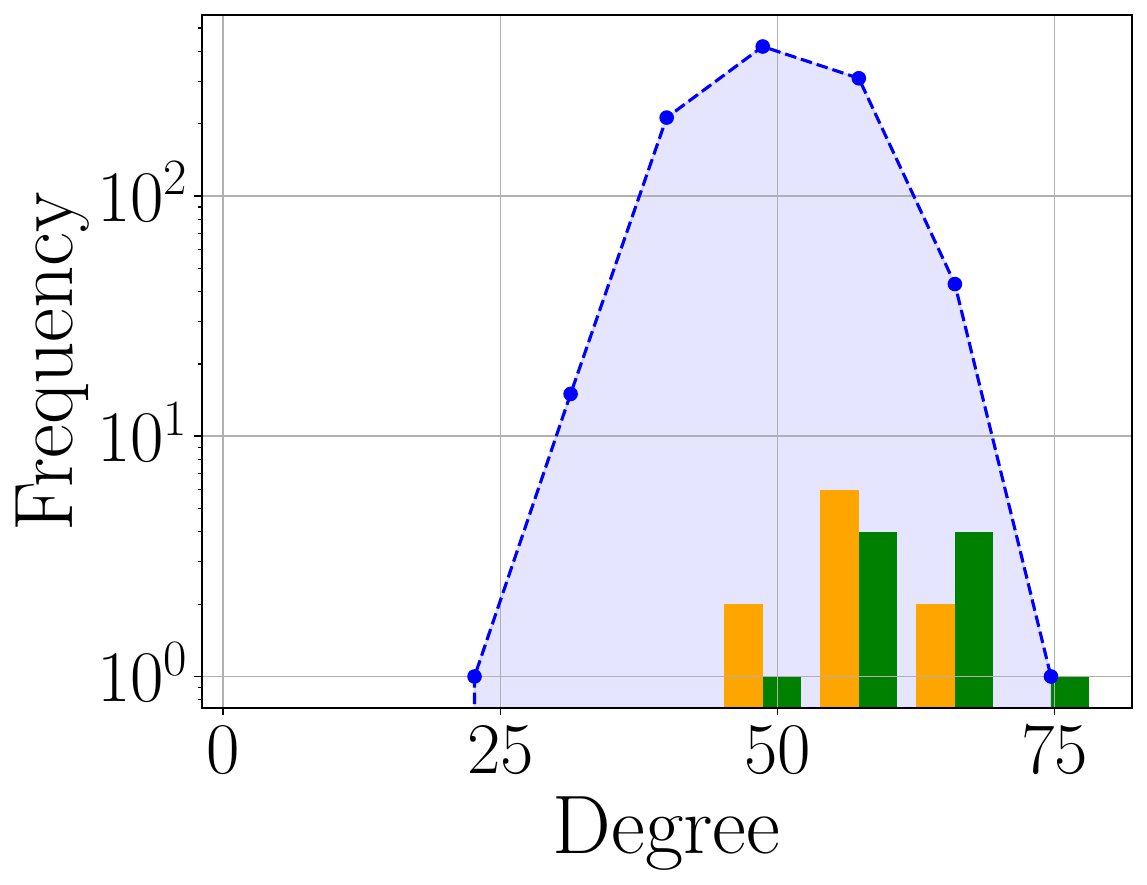}
    \caption{$IC(0.07,2)$}
    \end{subfigure}
    \includegraphics[width=0.5\linewidth]{degrees_legend.pdf}
    \caption{Comparison of the degree distrsc{Vulnerable} under \texttt{known-source} seeding on \er{}}
    \label{fig:degree_dist_addn_fixed_er}
\end{figure}

\begin{figure}
    \centering
    \begin{subfigure}{0.24\columnwidth}
    \includegraphics[width=0.99\linewidth]{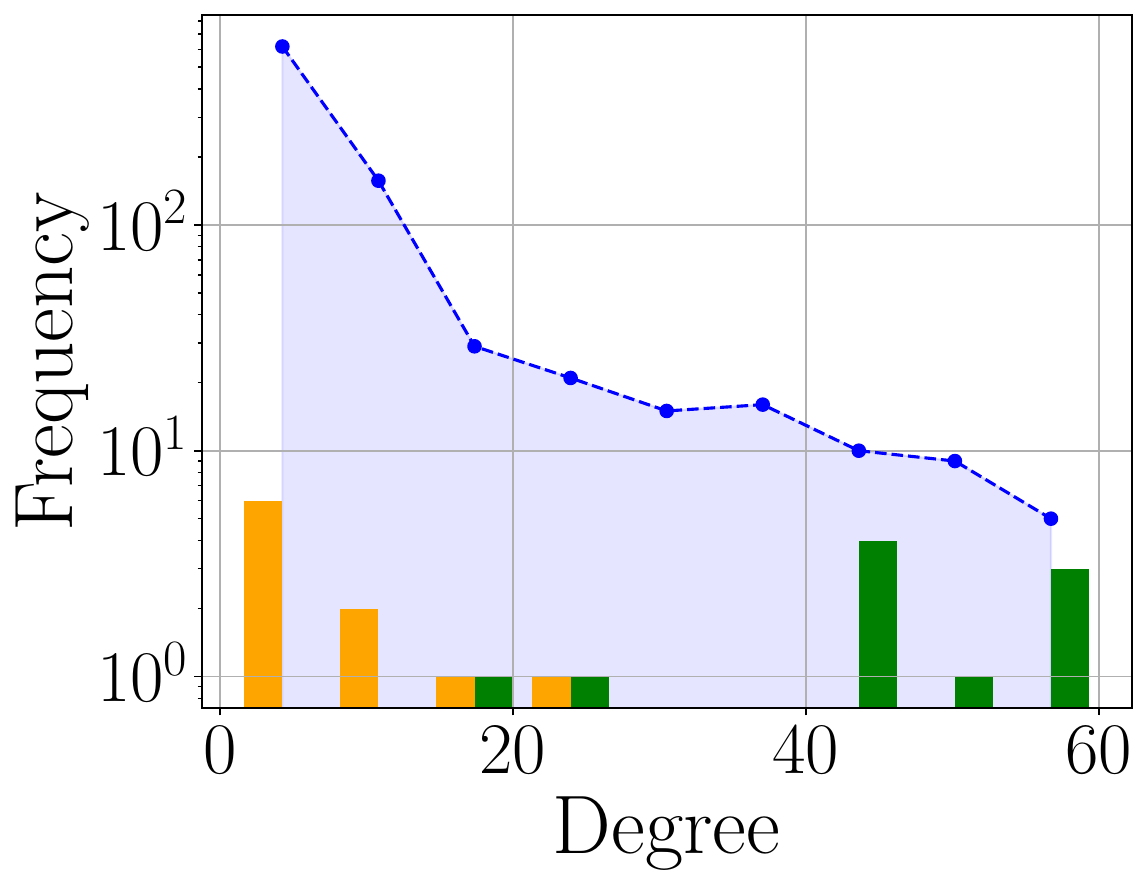}
    \caption{$IC(0.1,2)$}
    \end{subfigure}
     \begin{subfigure}{0.24\columnwidth}
    \includegraphics[width=0.99\linewidth]{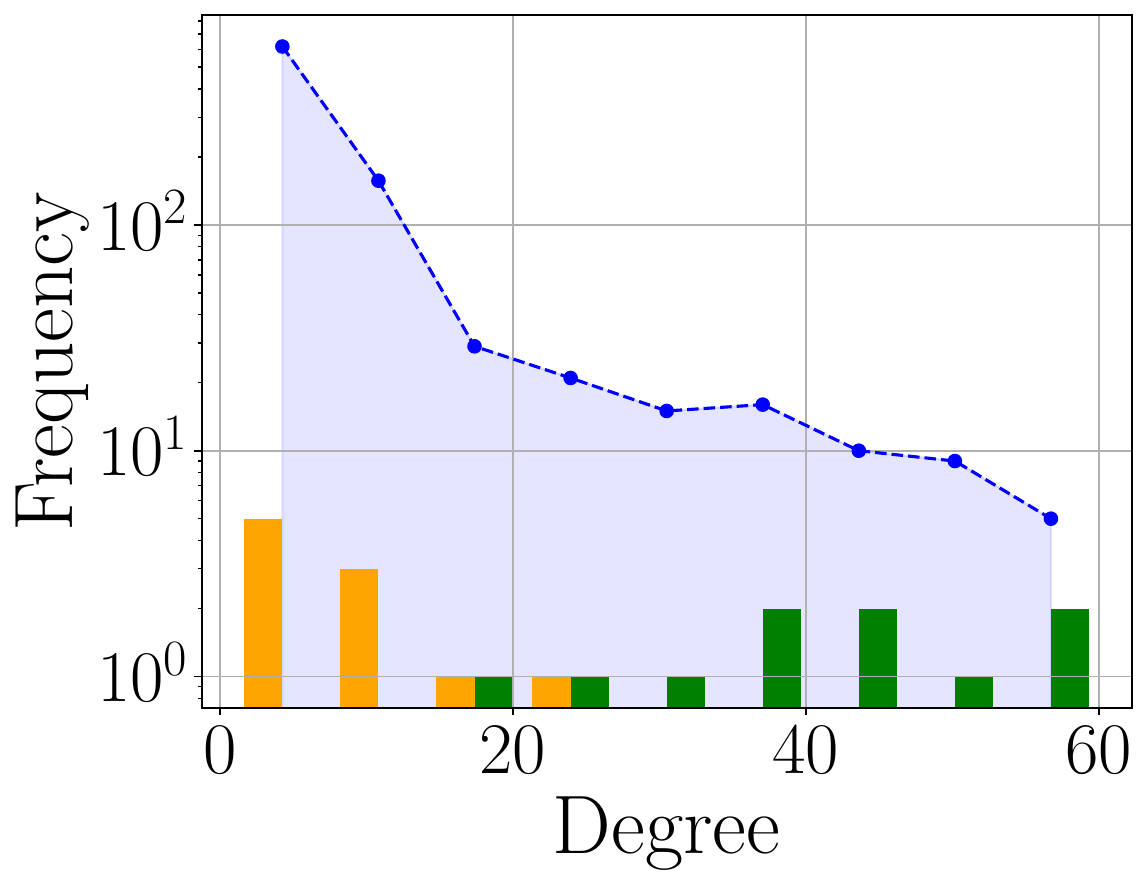}
    \caption{$IC(0.1,4)$}
    \end{subfigure}
    \begin{subfigure}{0.24\columnwidth}
    \includegraphics[width=0.99\linewidth]{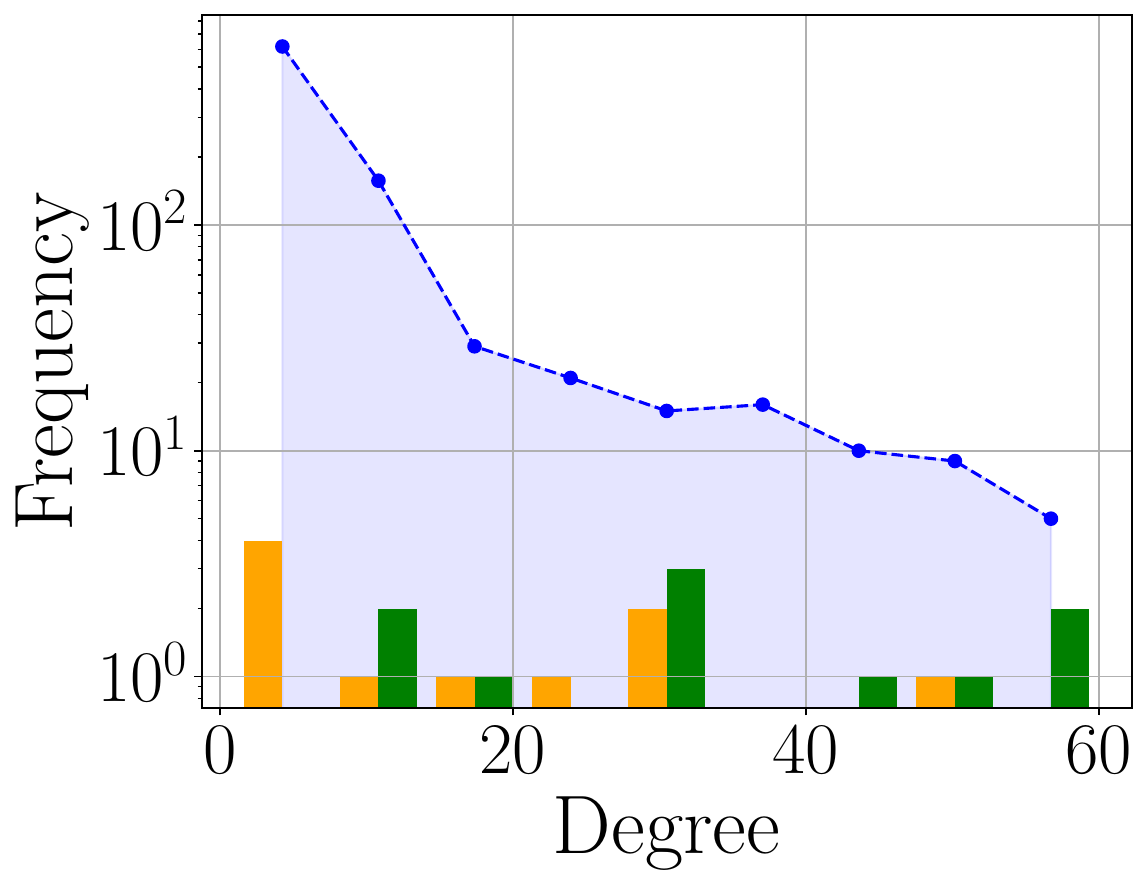}
    \caption{$IC(0.2,2)$}
    \end{subfigure}
     \begin{subfigure}{0.24\columnwidth}
    \includegraphics[width=0.99\linewidth]{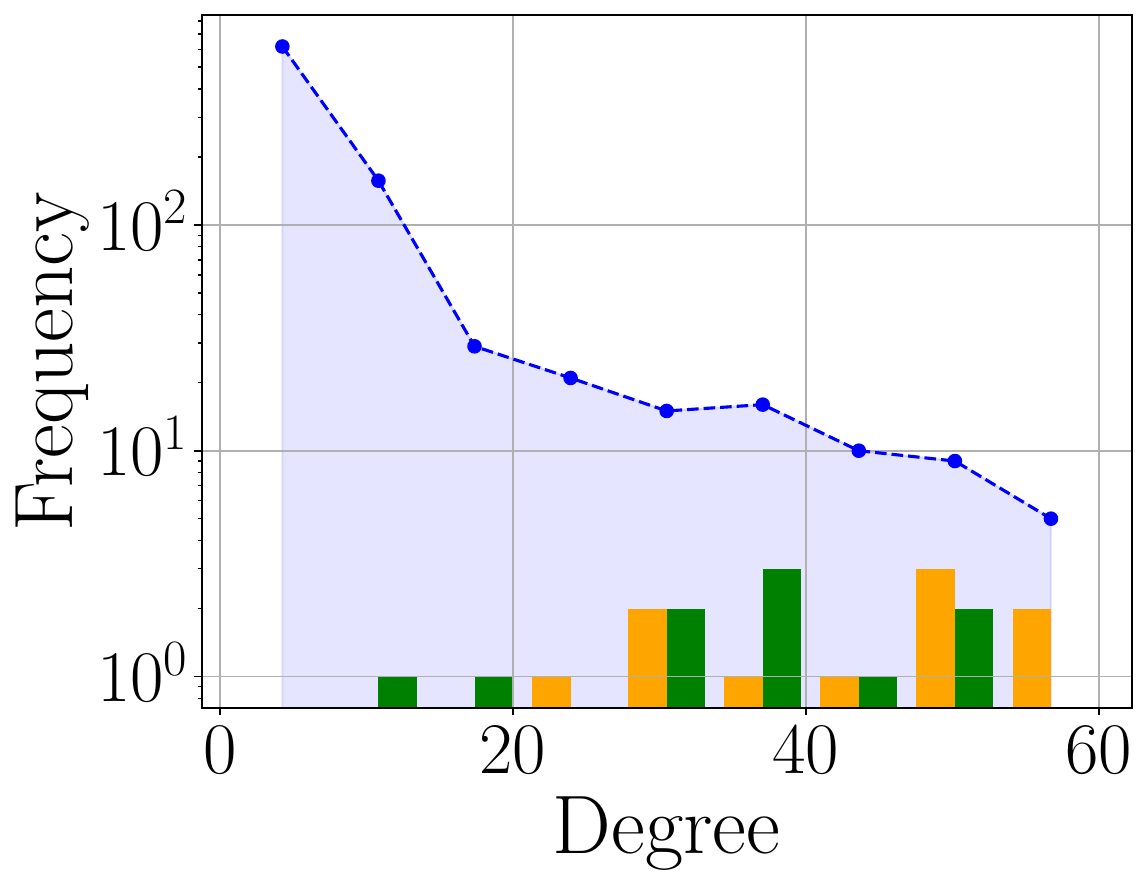}
    \caption{$IC(0.1,4)$}
    \end{subfigure}
    \includegraphics[width=0.5\linewidth]{degrees_legend.pdf}
    \caption{Comparison of the degree distribution of \textsc{GreedyMI} and \textsc{Vulnerable} under \texttt{known-source} seeding on \icu{}}
    \label{fig:degree_dist_addn_fixed_icu}
\end{figure}

\begin{figure}
    \centering
    \begin{subfigure}{0.24\columnwidth}
    \includegraphics[width=0.99\linewidth]{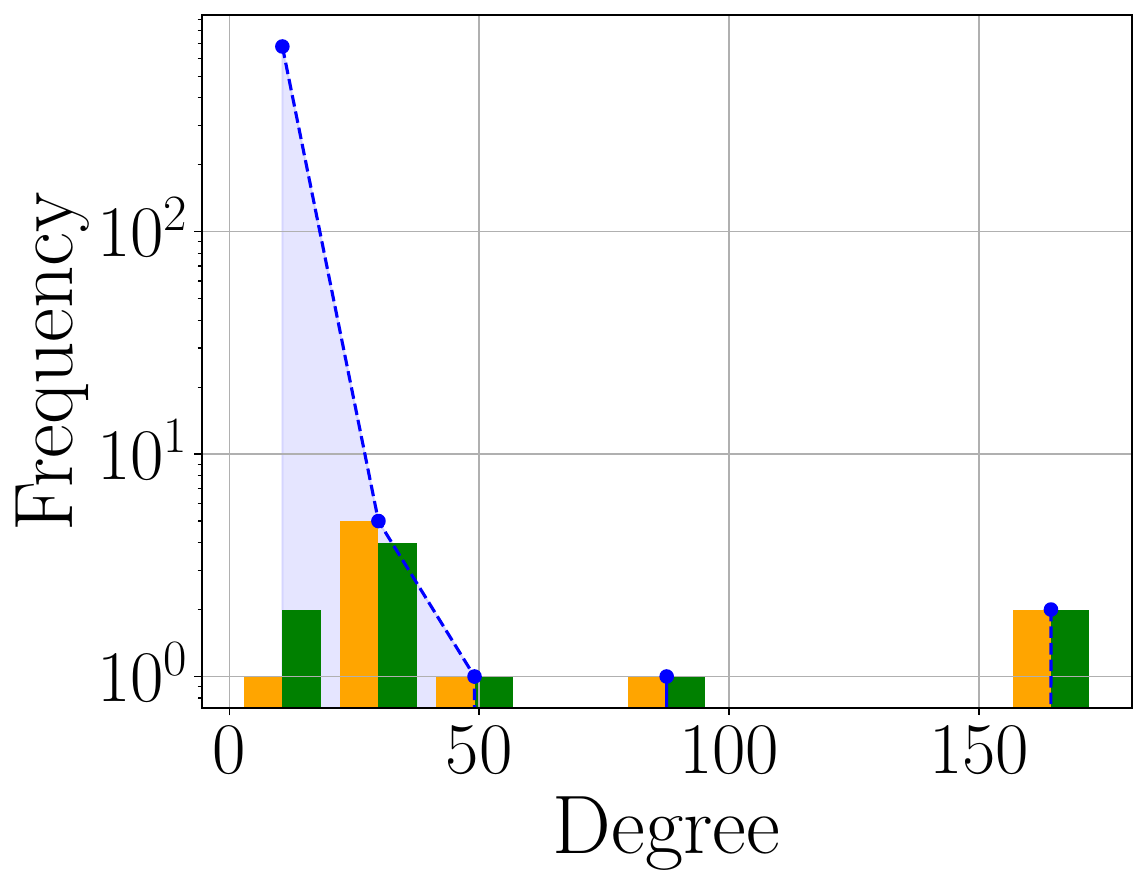}
    \caption{$IC(0.1,2)$}
    \end{subfigure}
     \begin{subfigure}{0.24\columnwidth}
    \includegraphics[width=0.99\linewidth]{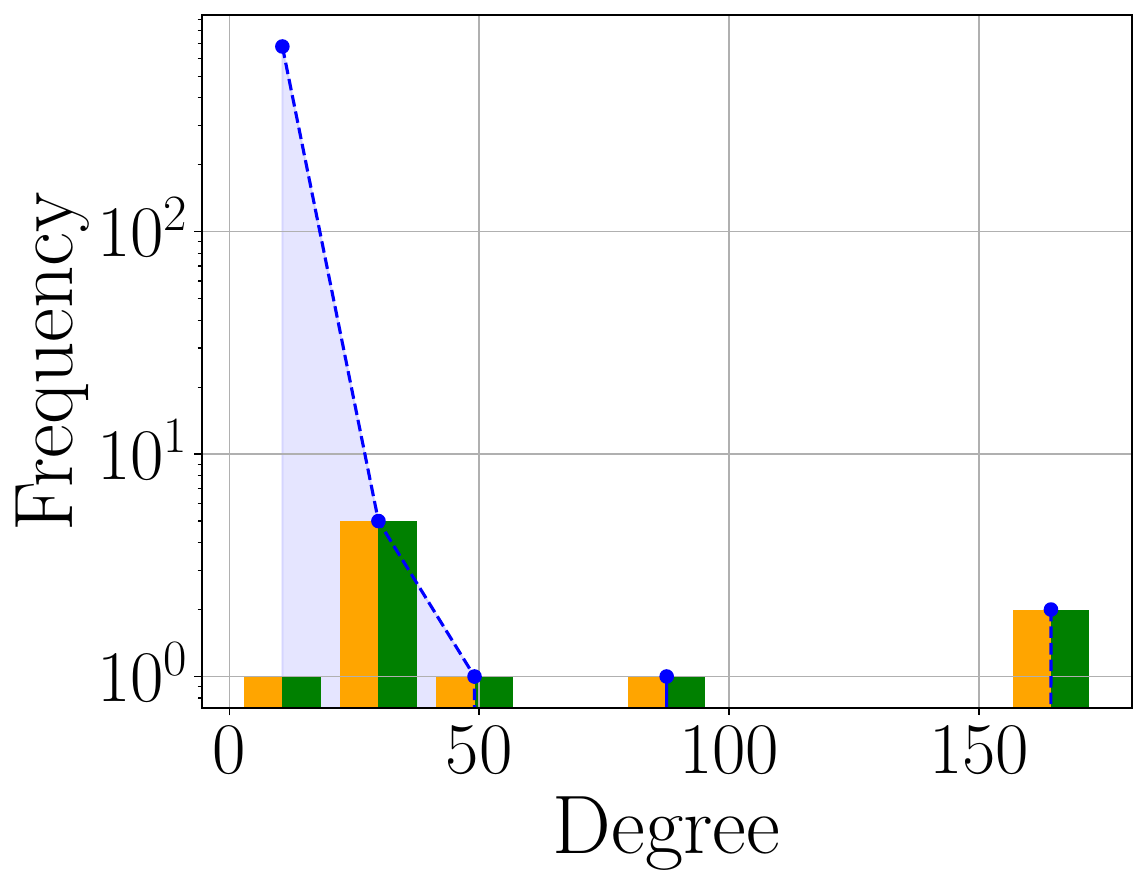}
    \caption{$IC(0.1,4)$}
    \end{subfigure}
     \begin{subfigure}{0.24\columnwidth}
    \includegraphics[width=0.99\linewidth]{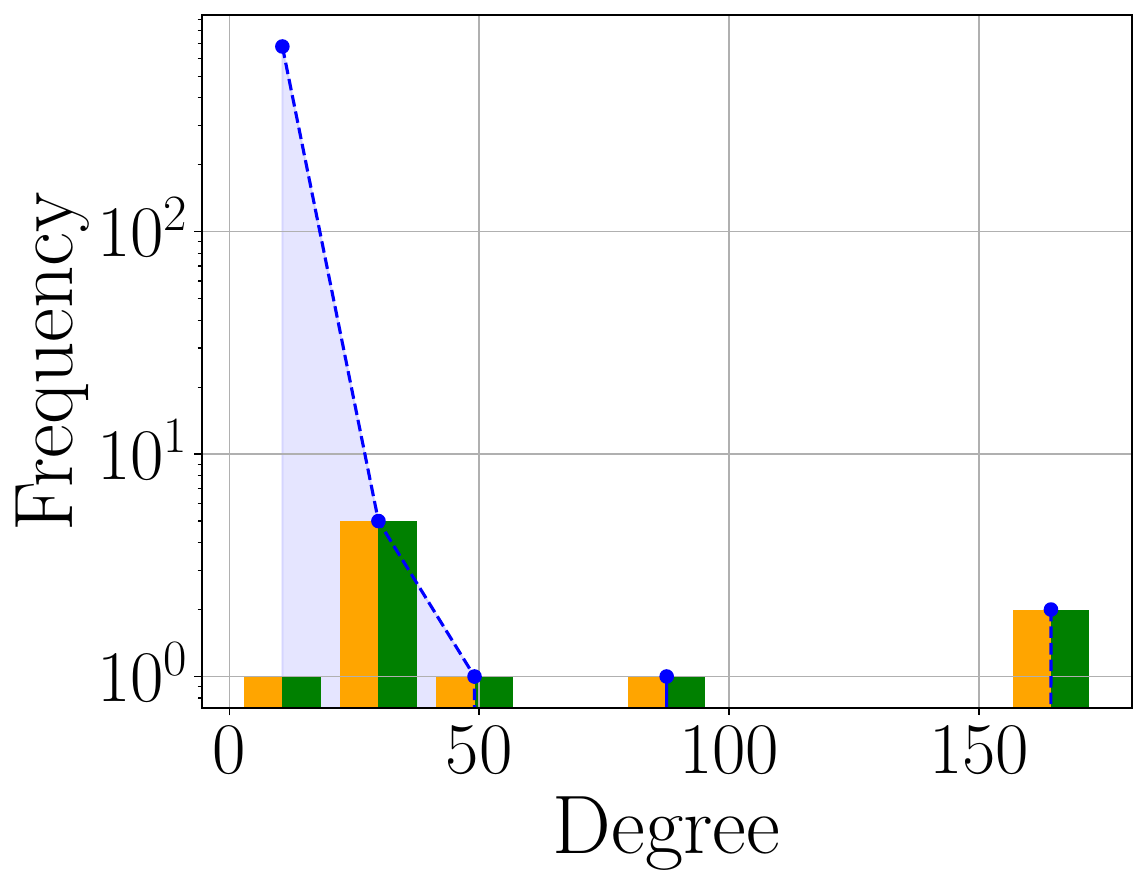}
    \caption{$IC(0.2,2)$}
    \end{subfigure}
     \begin{subfigure}{0.24\columnwidth}
    \includegraphics[width=0.99\linewidth]{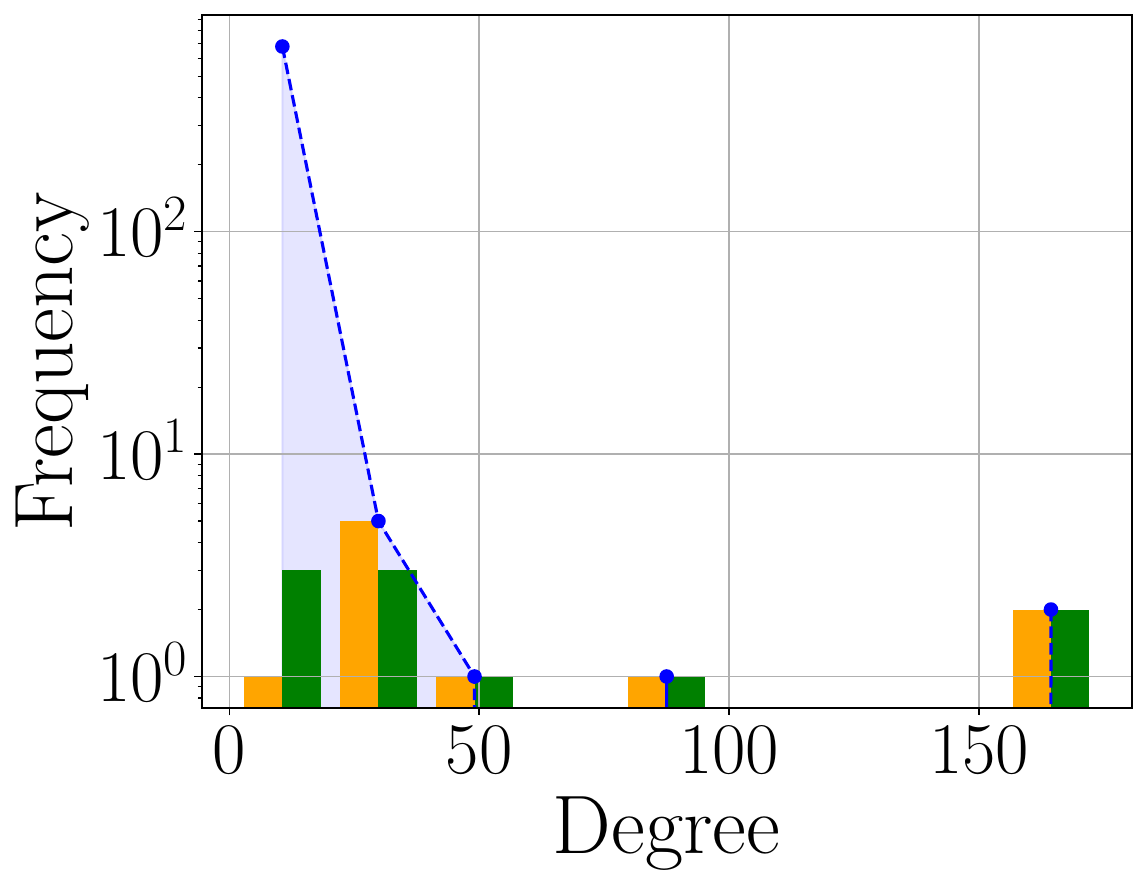}
    \caption{$IC(0.2,4)$}
    \end{subfigure}
    \includegraphics[width=0.5\linewidth]{degrees_legend.pdf}
    \caption{Comparison of the degree distribution of \textsc{GreedyMI} and \textsc{Vulnerable} under \rsource{} seeding on \pl{}}
    \label{fig:degree_dist_addn_random_pl}
\end{figure}

\begin{figure}
    \centering
    \begin{subfigure}{0.28\columnwidth}
    \includegraphics[width=0.99\linewidth]{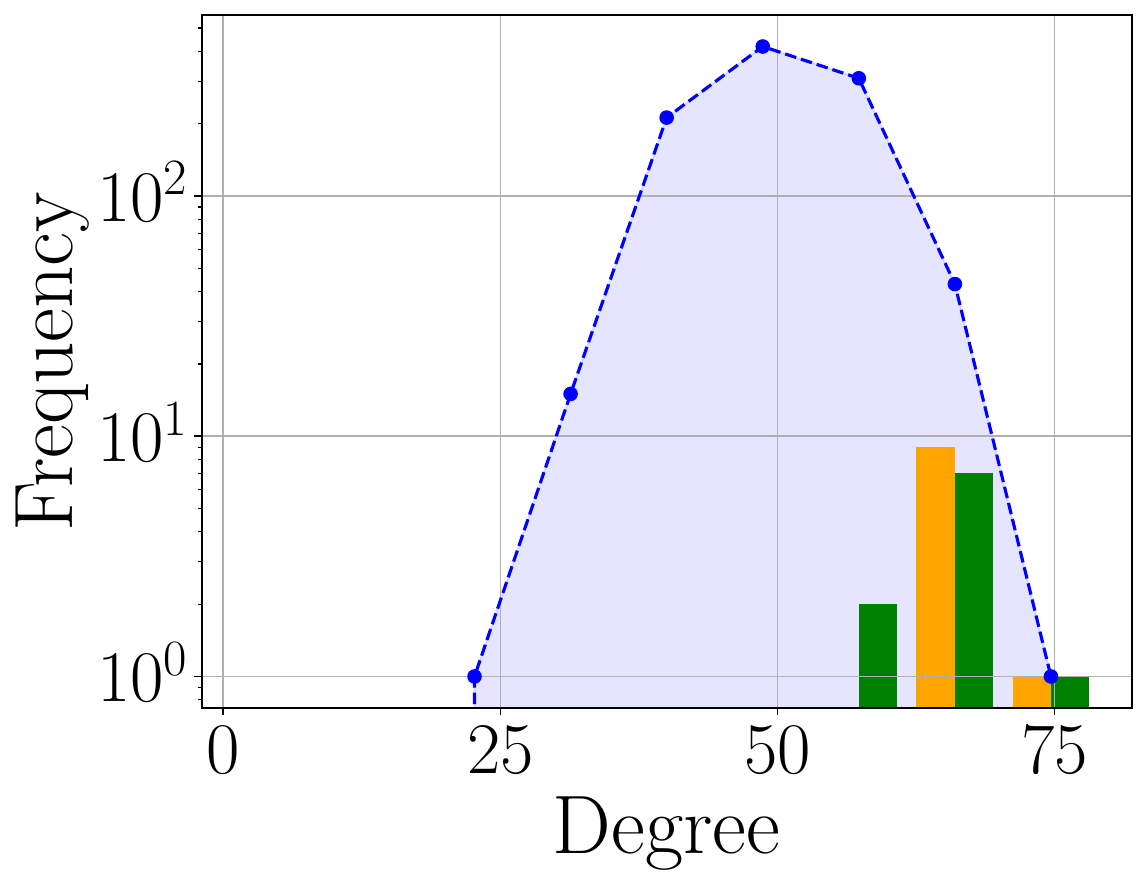}
    \caption{$IC(0.05,2)$}
    \end{subfigure}
     \begin{subfigure}{0.28\columnwidth}
    \includegraphics[width=0.99\linewidth]{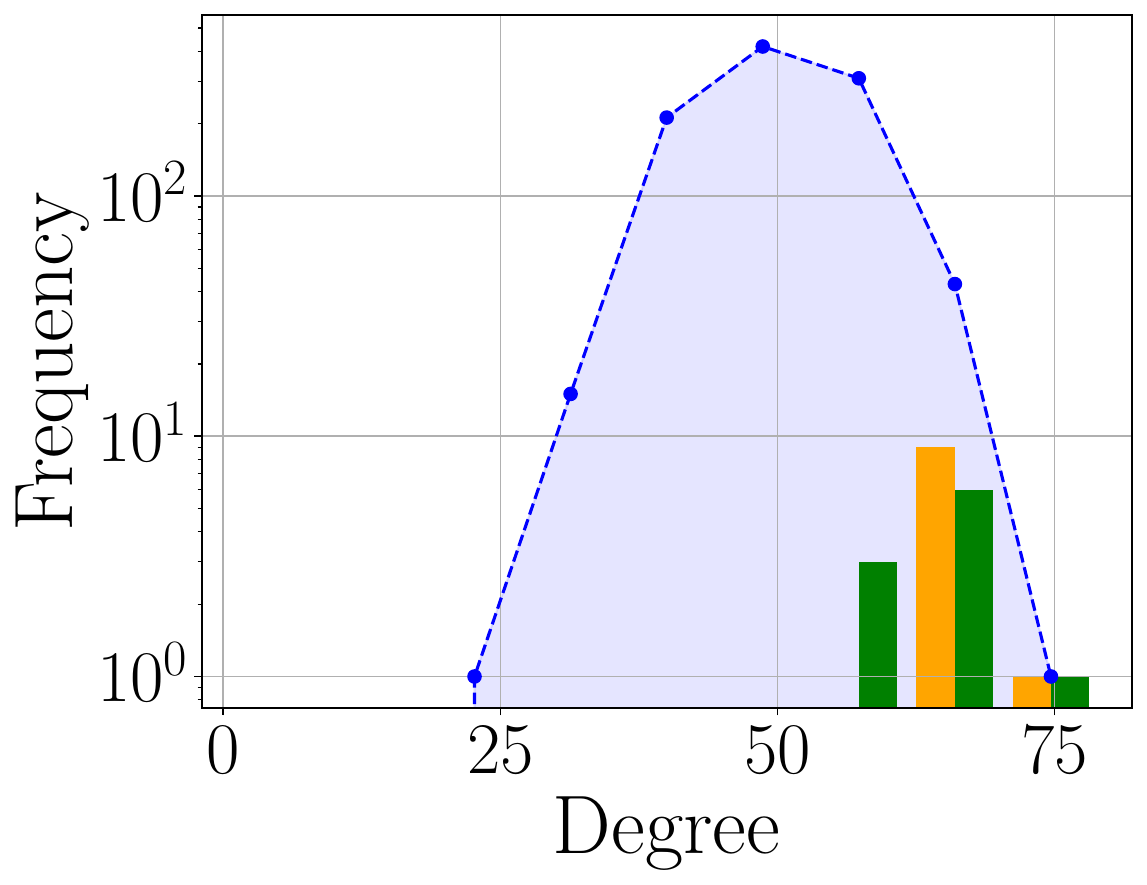}
    \caption{$IC(0.07,2)$}
    \end{subfigure}
    \includegraphics[width=0.5\linewidth]{degrees_legend.pdf}
    \caption{Comparison of the degree distribution of \textsc{GreedyMI} and \textsc{Vulnerable} in \rsource{} seeding on \er{}}
    \label{fig:degree_dist_addn_random_er}
\end{figure}

\begin{figure}
    \centering
    \begin{subfigure}{0.24\columnwidth}
    \includegraphics[width=0.99\linewidth]{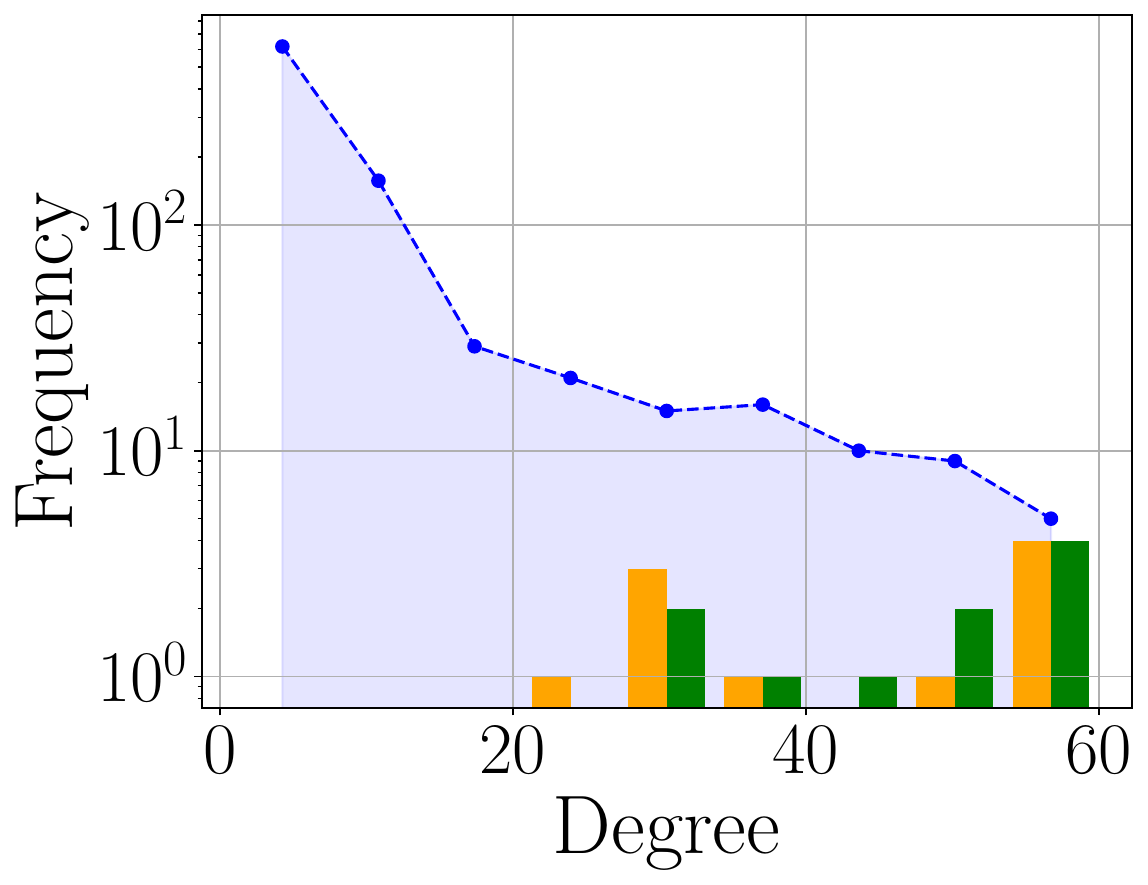}
    \caption{$IC(0.1,2)$}
    \end{subfigure}
     \begin{subfigure}{0.24\columnwidth}
    \includegraphics[width=0.99\linewidth]{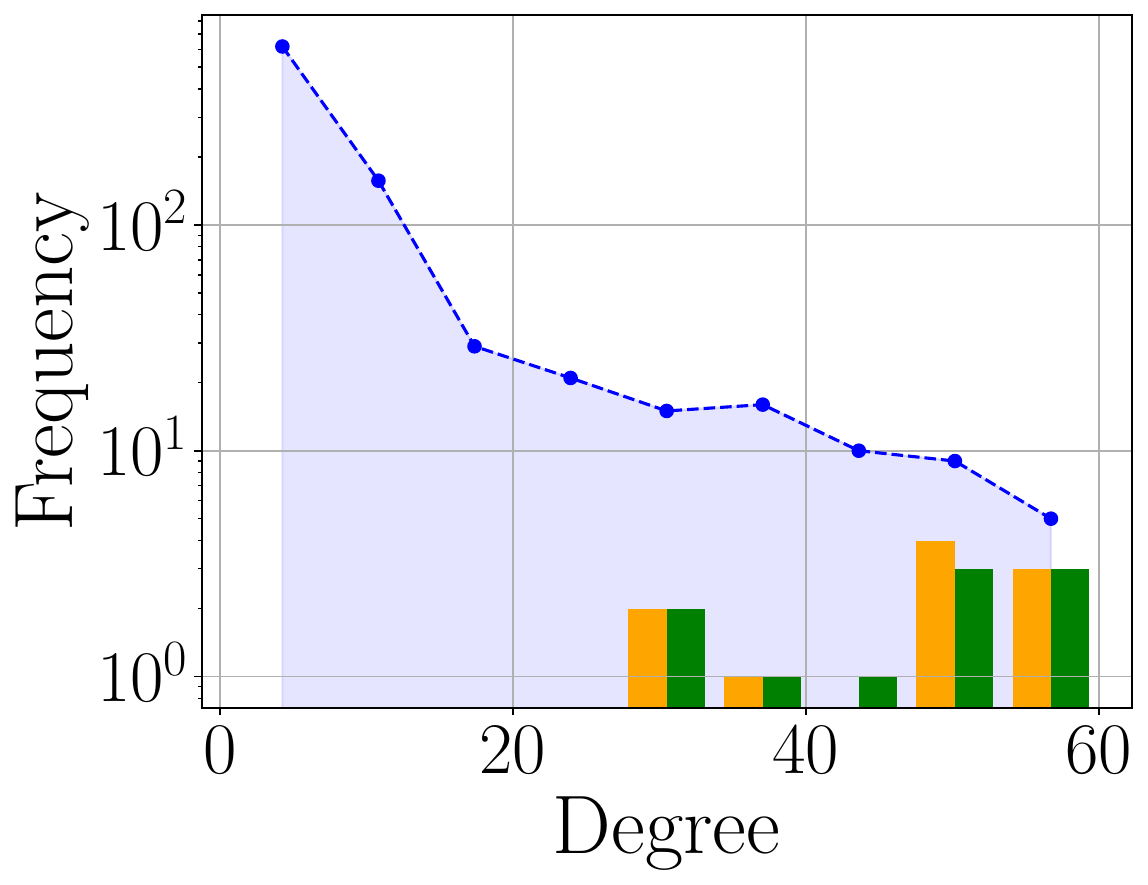}
    \caption{$IC(0.1,4)$}
    \end{subfigure}
    \begin{subfigure}{0.24\columnwidth}
    \includegraphics[width=0.99\linewidth]{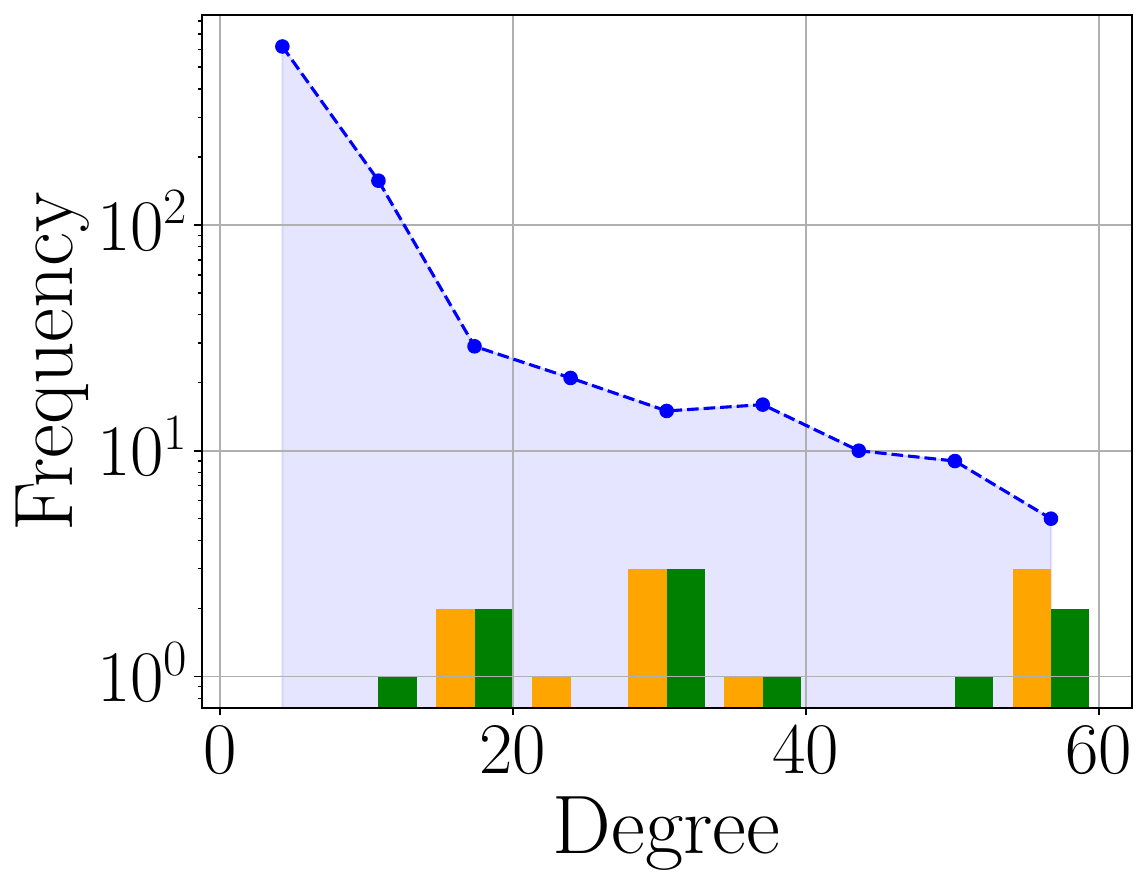}
    \caption{$IC(0.2,2)$}
    \end{subfigure}
     \begin{subfigure}{0.24\columnwidth}
    \includegraphics[width=0.99\linewidth]{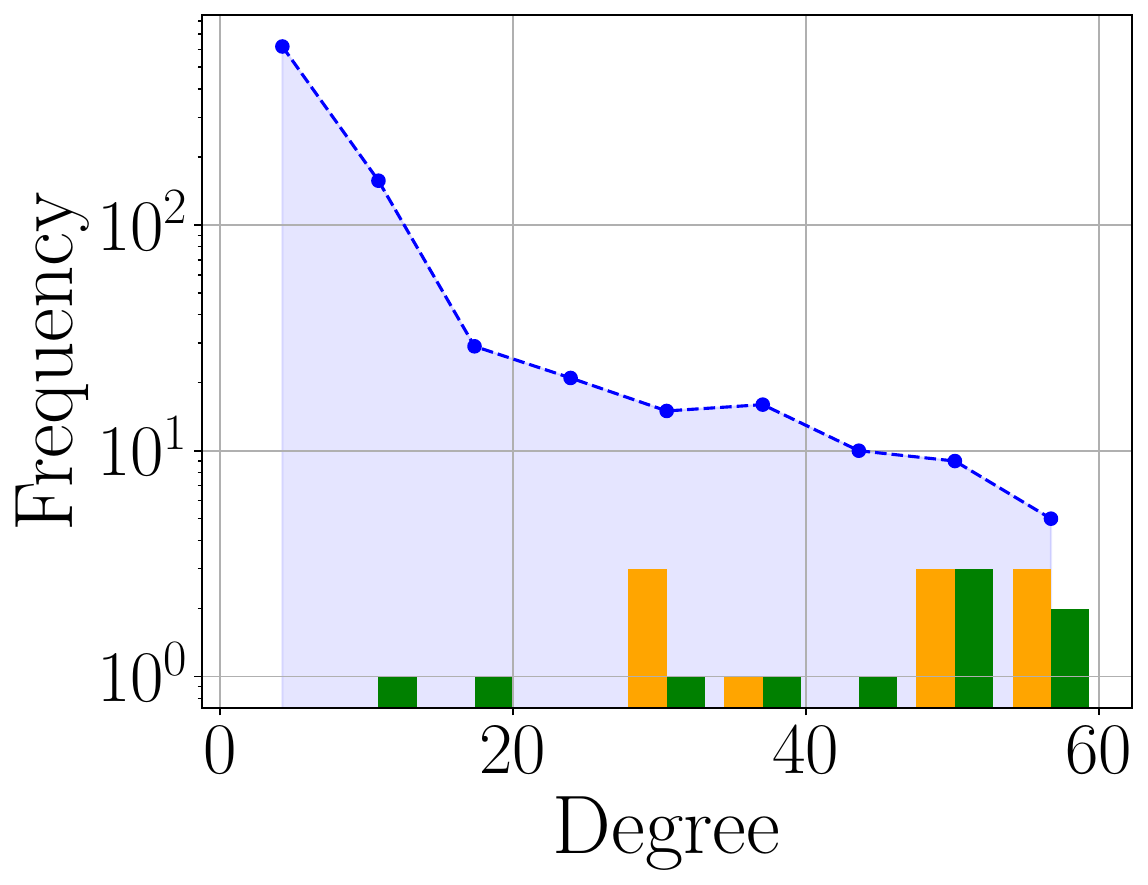}
    \caption{$IC(0.1,4)$}
    \end{subfigure}
    \includegraphics[width=0.5\linewidth]{degrees_legend.pdf}
    \caption{Comparison of the degree distribution of \textsc{GreedyMI} and \textsc{Vulnerable} under \rsource{} seeding on \icu{}}
    \label{fig:degree_dist_addn_random_icu}
\end{figure}

\begin{figure}
    \centering
    \begin{subfigure}{0.24\columnwidth}
    \includegraphics[width=0.99\linewidth]{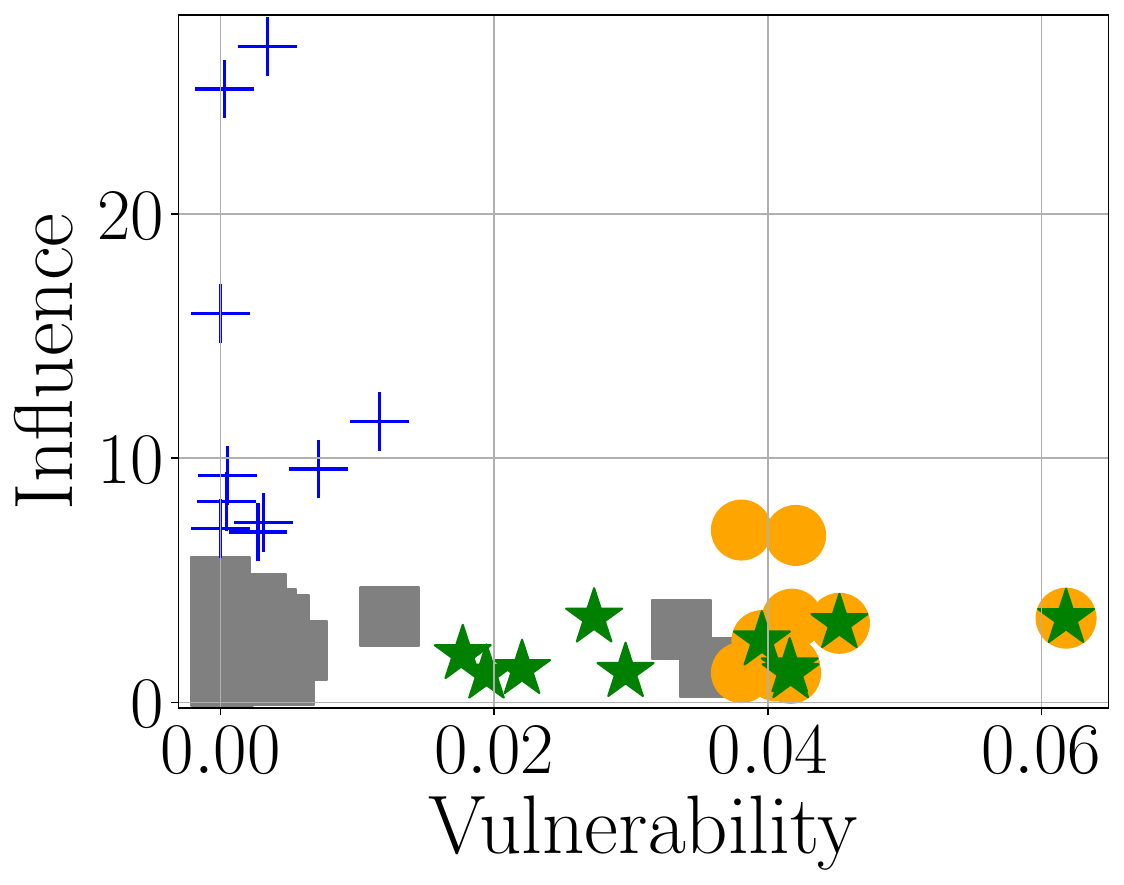}
         \caption{$IC(0.1,2)$}
    \end{subfigure}
   \begin{subfigure}{0.24\columnwidth}
    \includegraphics[width=0.99\linewidth]{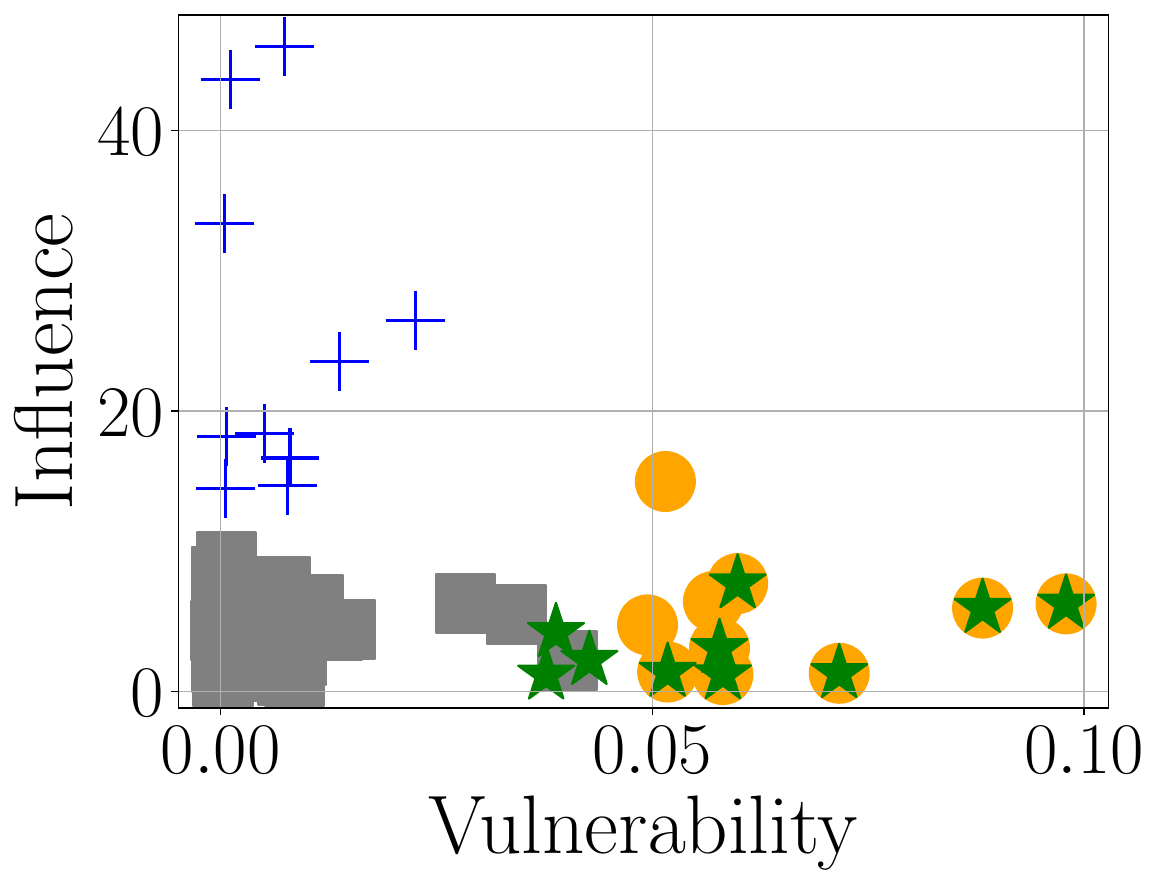}
    \caption{$IC(0.1,4)$}
    \end{subfigure}
    \begin{subfigure}{0.24\columnwidth}
    \includegraphics[width=0.99\linewidth]{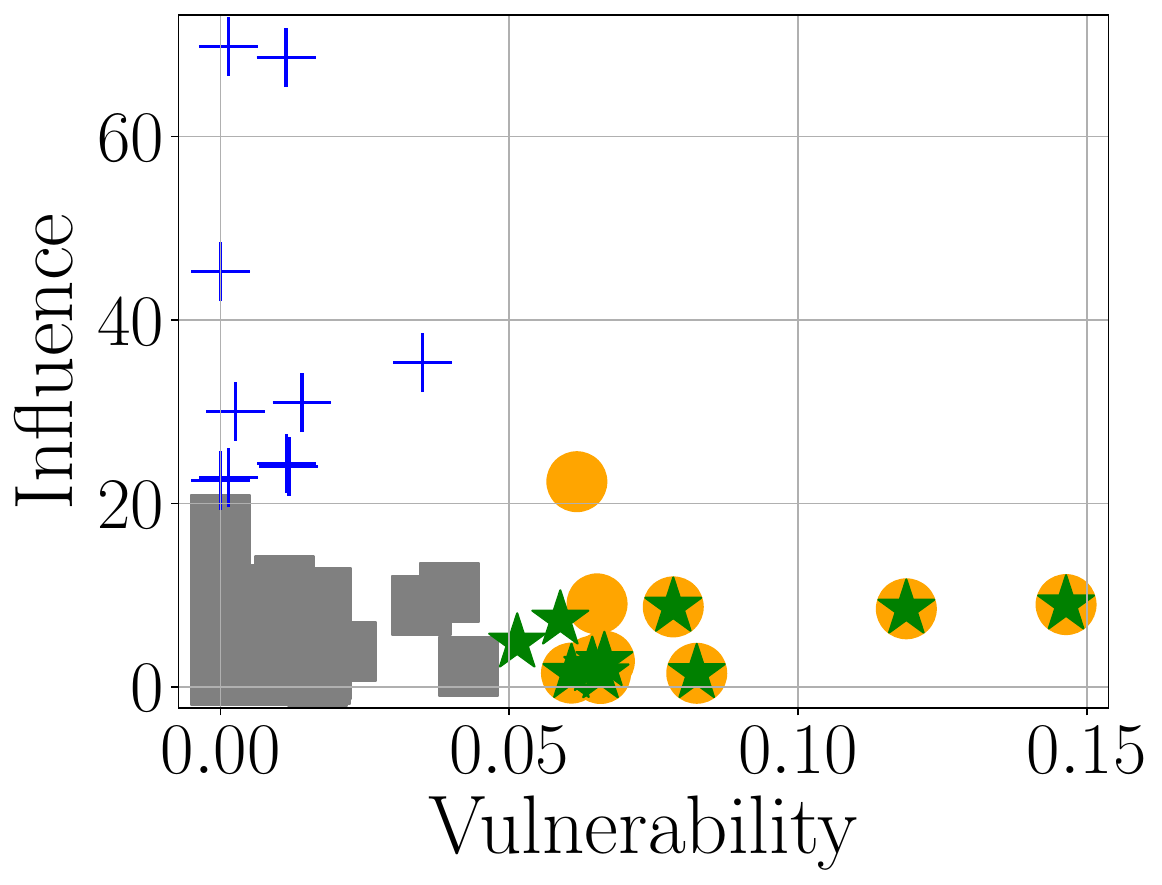}
    \caption{$IC(0.2,2)$}
    \end{subfigure}
    \begin{subfigure}{0.24\columnwidth}
    \includegraphics[width=0.99\linewidth]{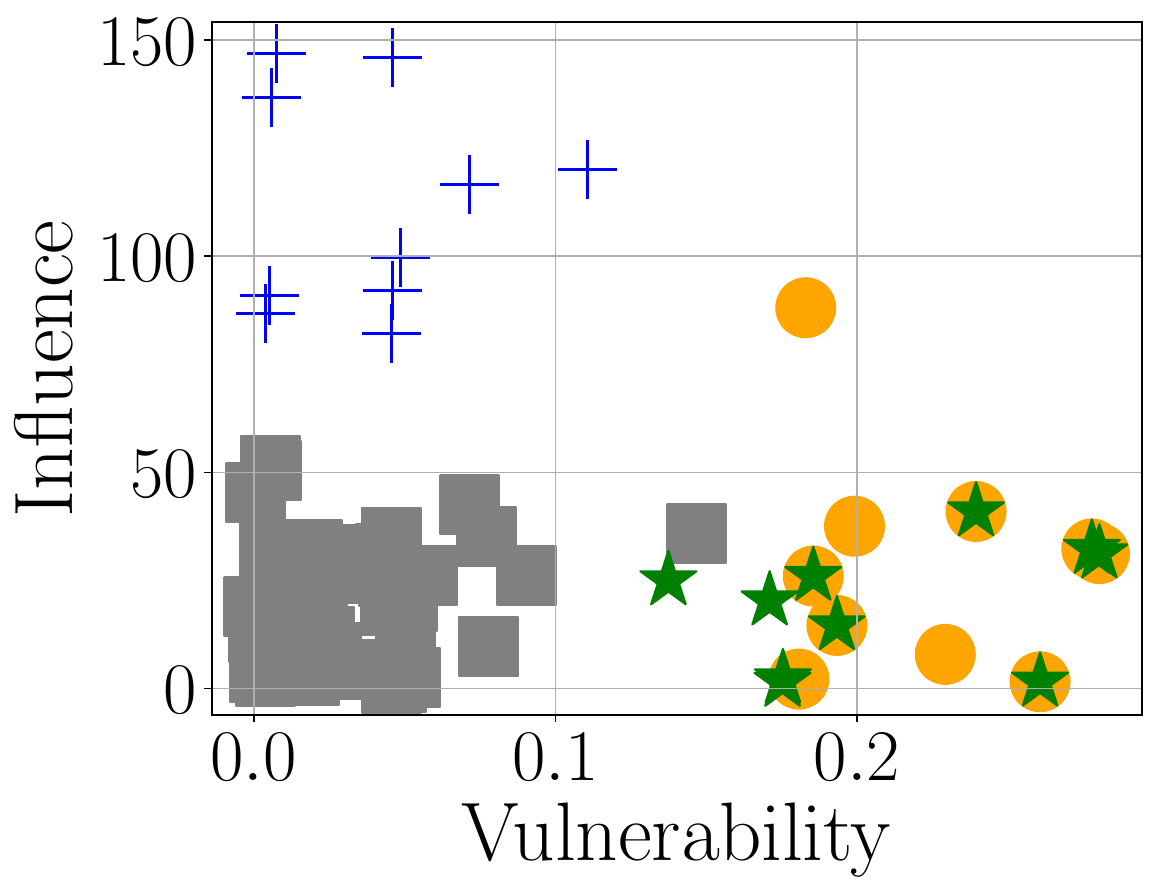}
     \caption{$IC(0.2,4)$}
     \end{subfigure}
     \includegraphics[width=0.5\linewidth]{influence_legend.pdf}
    \caption{Vulnerability vs Influence of the surveillance sets by \tool{} vs baselines  under \ksource{}~seeding on \pl{}}
    \label{fig:infl_add_fixed_pl}
\end{figure}

\begin{figure}
    \centering
    \begin{subfigure}{0.28\columnwidth}
    \includegraphics[width=0.99\linewidth]{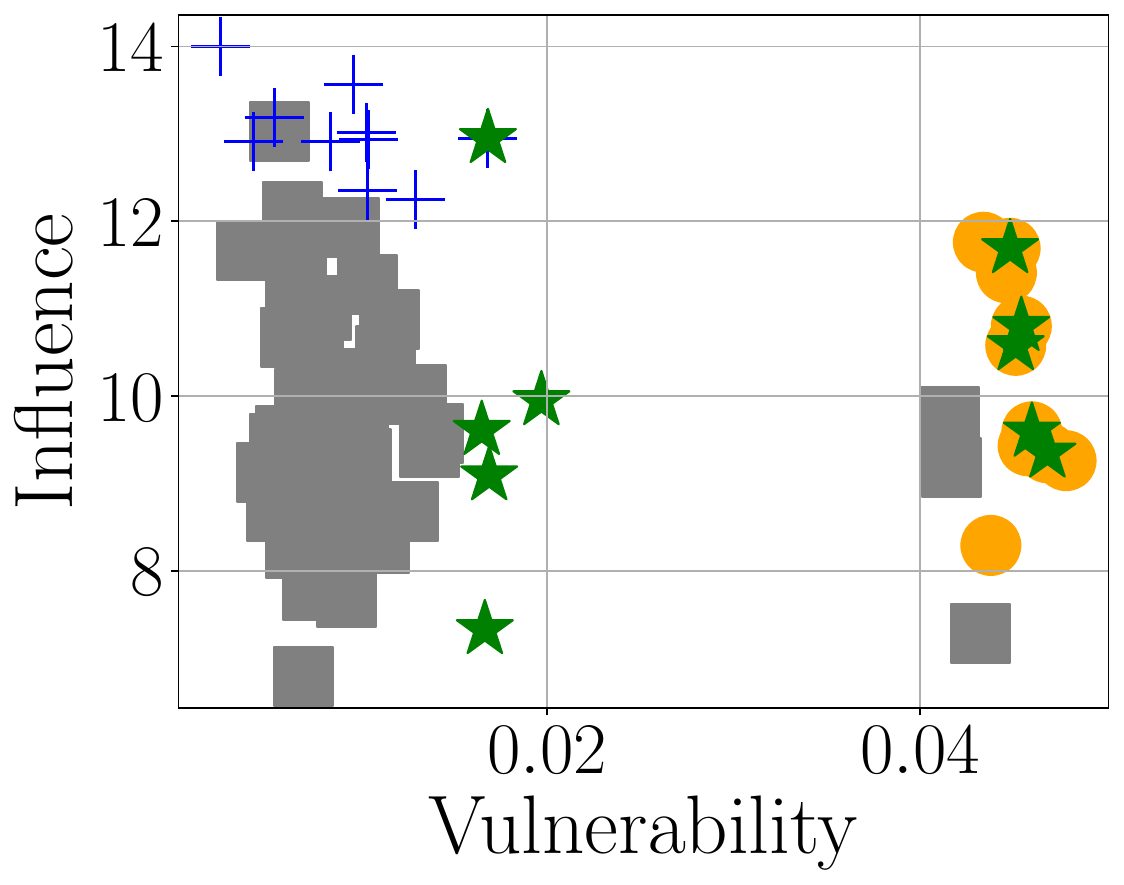}
    \caption{$IC(0.05,2)$}
    \end{subfigure}
    \begin{subfigure}{0.28\columnwidth}
     \includegraphics[width=0.99\linewidth]{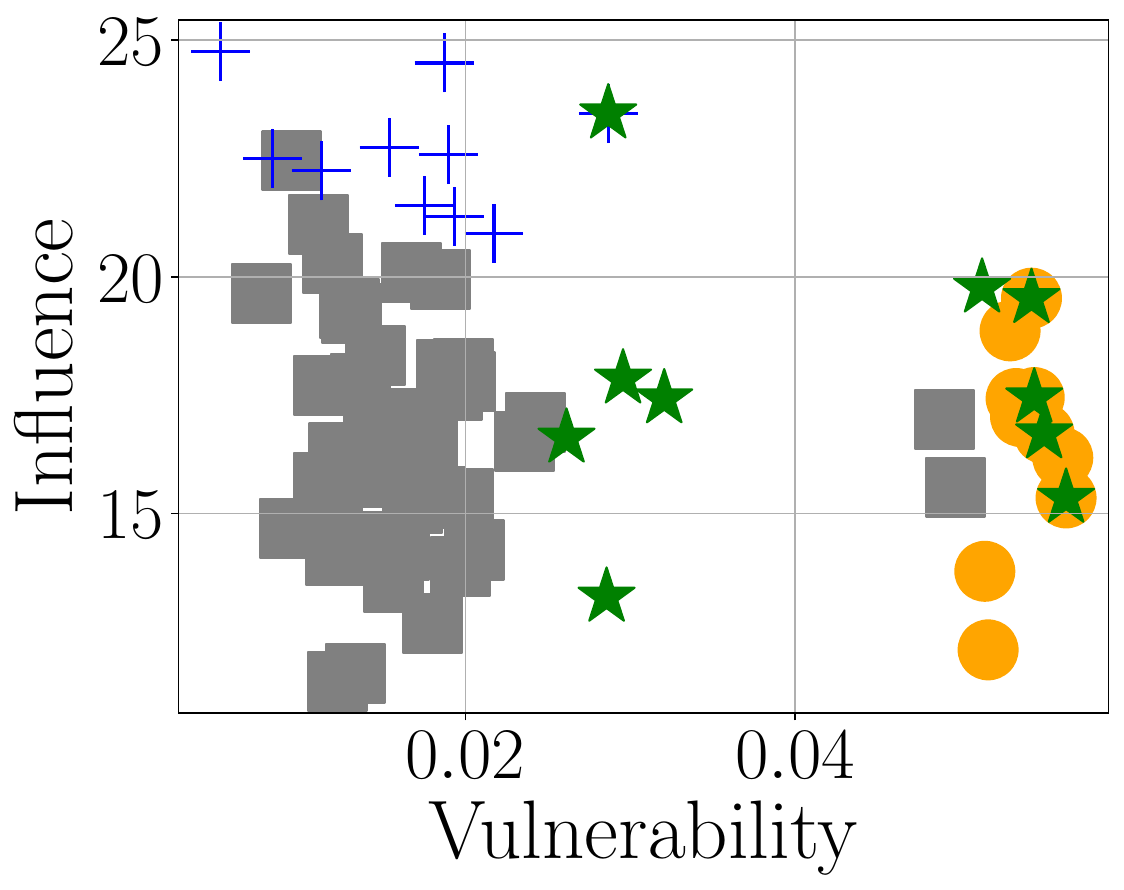}
     \caption{$IC(0.07,2)$}
     \end{subfigure}
     \includegraphics[width=0.5\linewidth]{influence_legend.pdf}
    \caption{Vulnerability vs Influence of the surveillance sets by \tool{} vs baselines  under \ksource{}~seeding on \er{}}
    \label{fig:infl_add_fixed_er}
\end{figure}

\begin{figure}
    \centering
    \begin{subfigure}{0.24\columnwidth}
    \includegraphics[width=0.99\linewidth]{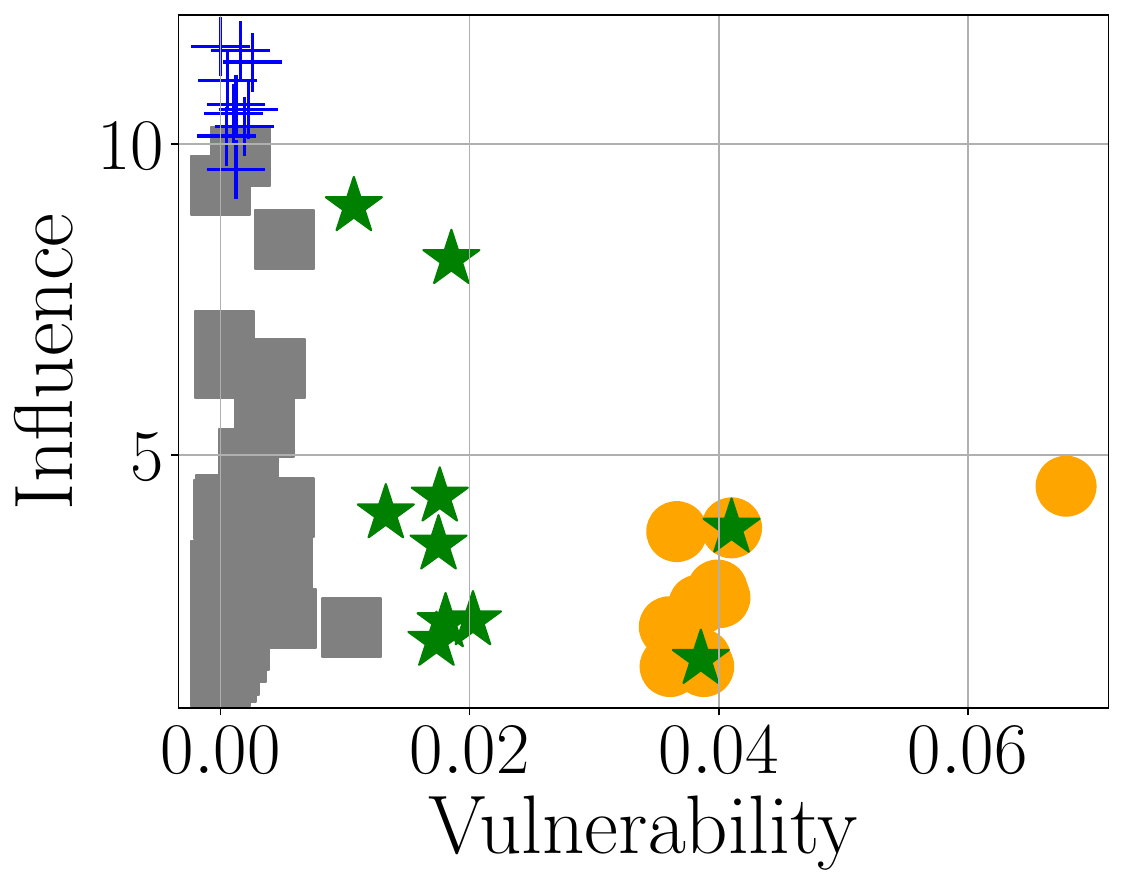}
         \caption{$IC(0.1,2)$}
    \end{subfigure}
   \begin{subfigure}{0.24\columnwidth}
    \includegraphics[width=0.99\linewidth]{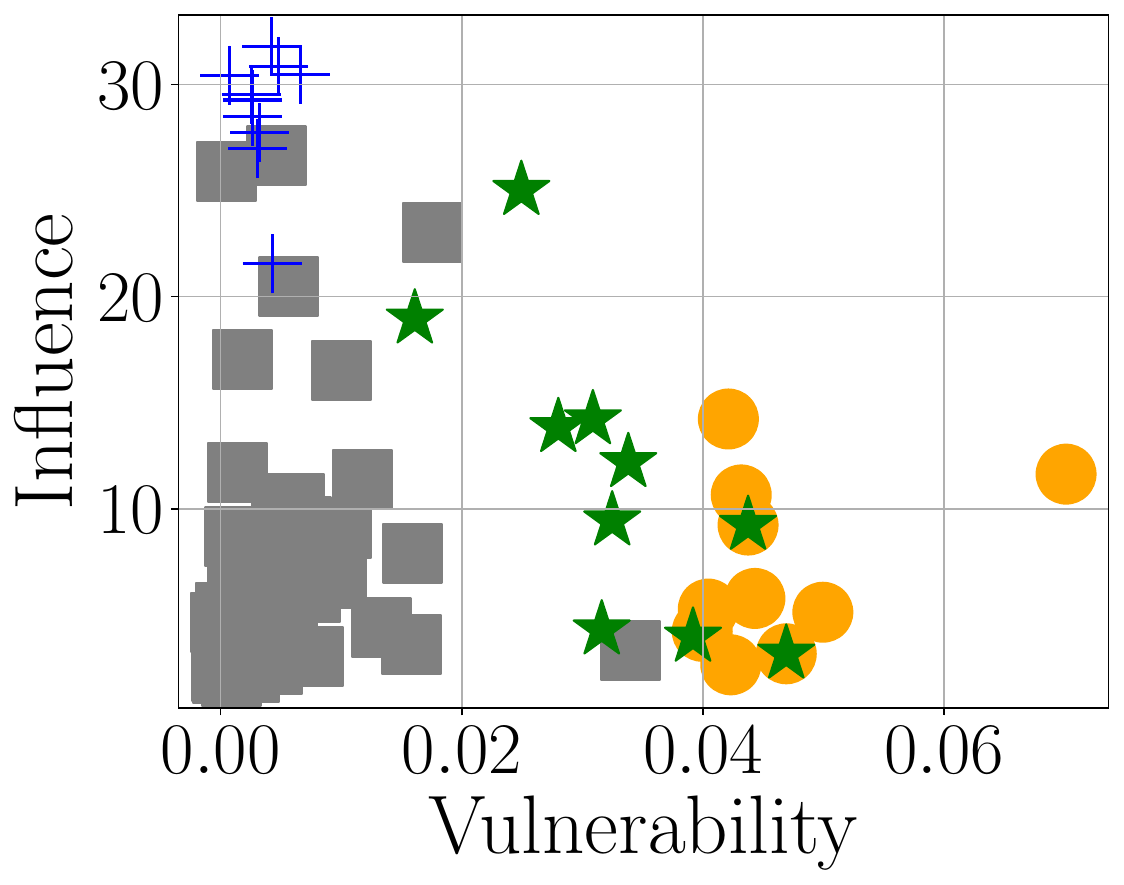}
    \caption{$IC(0.1,4)$}
    \end{subfigure}
    \begin{subfigure}{0.24\columnwidth}
    \includegraphics[width=0.99\linewidth]{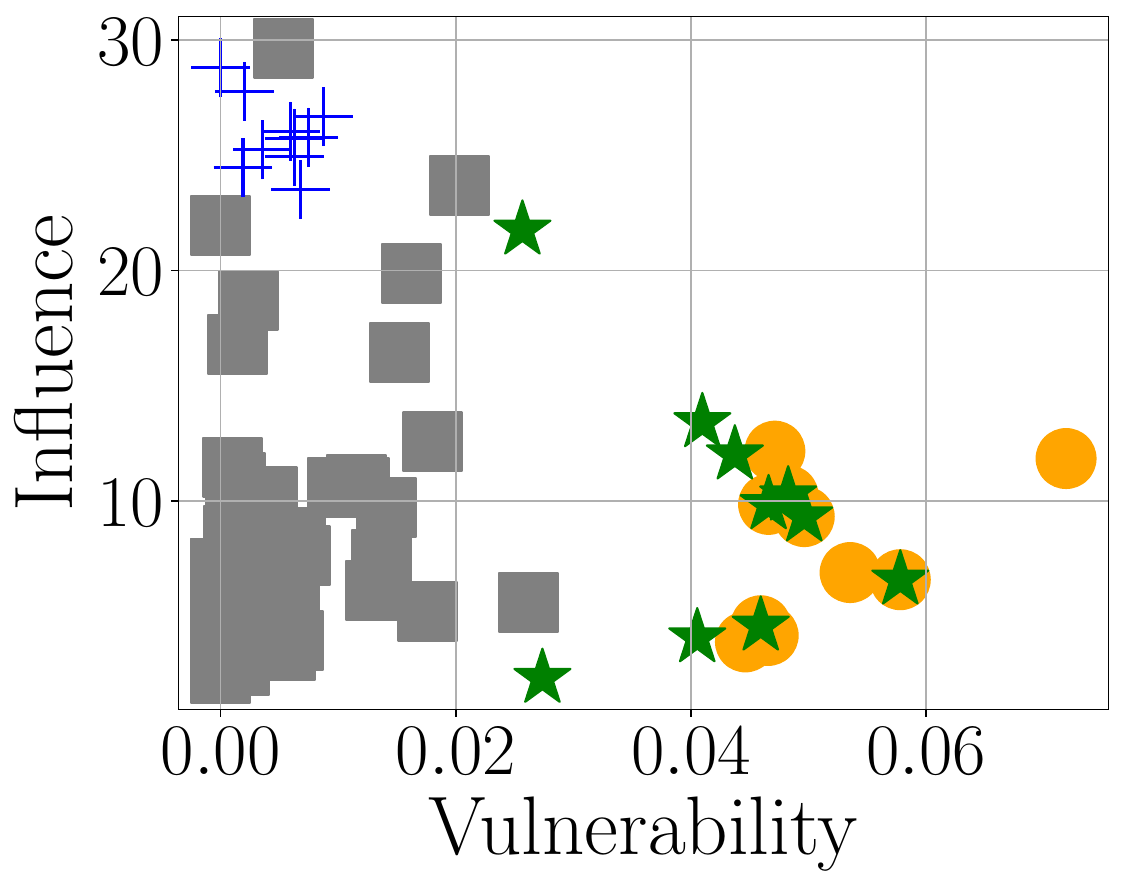}
    \caption{$IC(0.2,2)$}
    \end{subfigure}
    \begin{subfigure}{0.24\columnwidth}
    \includegraphics[width=0.99\linewidth]{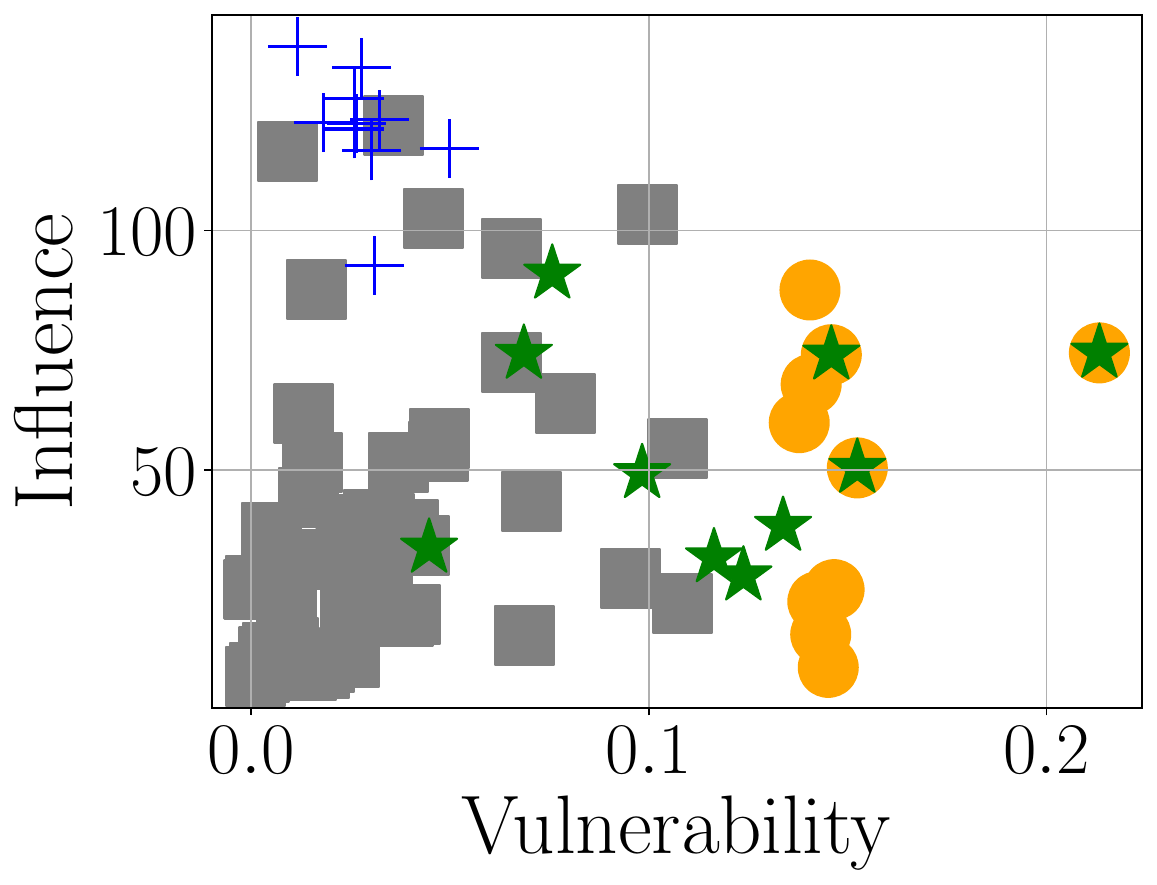}
     \caption{$IC(0.2,4)$}
     \end{subfigure}
     \includegraphics[width=0.5\linewidth]{influence_legend.pdf}
    \caption{Vulnerability vs Influence of the surveillance sets by \tool{} vs baselines  under \ksource{}~seeding on \icu{}}
    \label{fig:infl_add_fixed_icu}
\end{figure}

\begin{figure}
    \centering
    \begin{subfigure}{0.24\columnwidth}
    \includegraphics[width=0.99\linewidth]{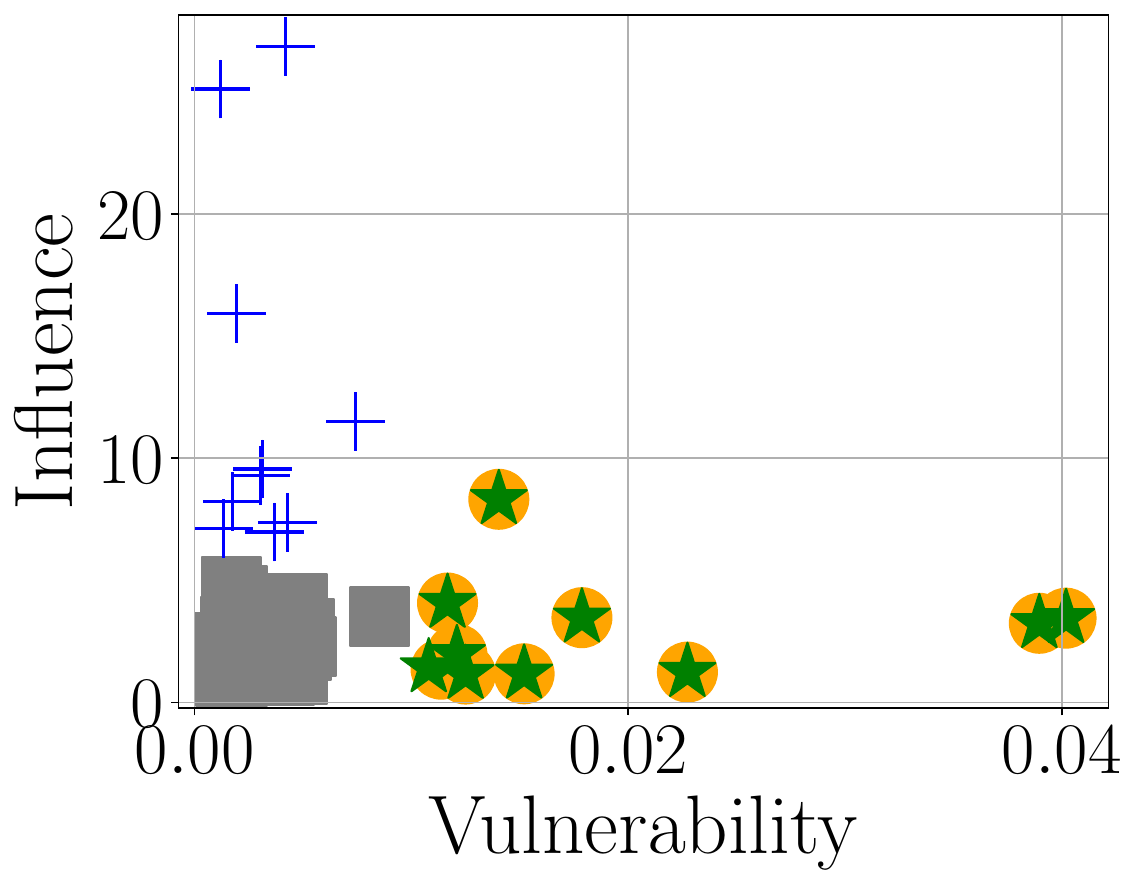}
    \caption{$IC(0.1,2)$}
    \end{subfigure}
    \begin{subfigure}{0.24\columnwidth}
    \includegraphics[width=0.99\linewidth]{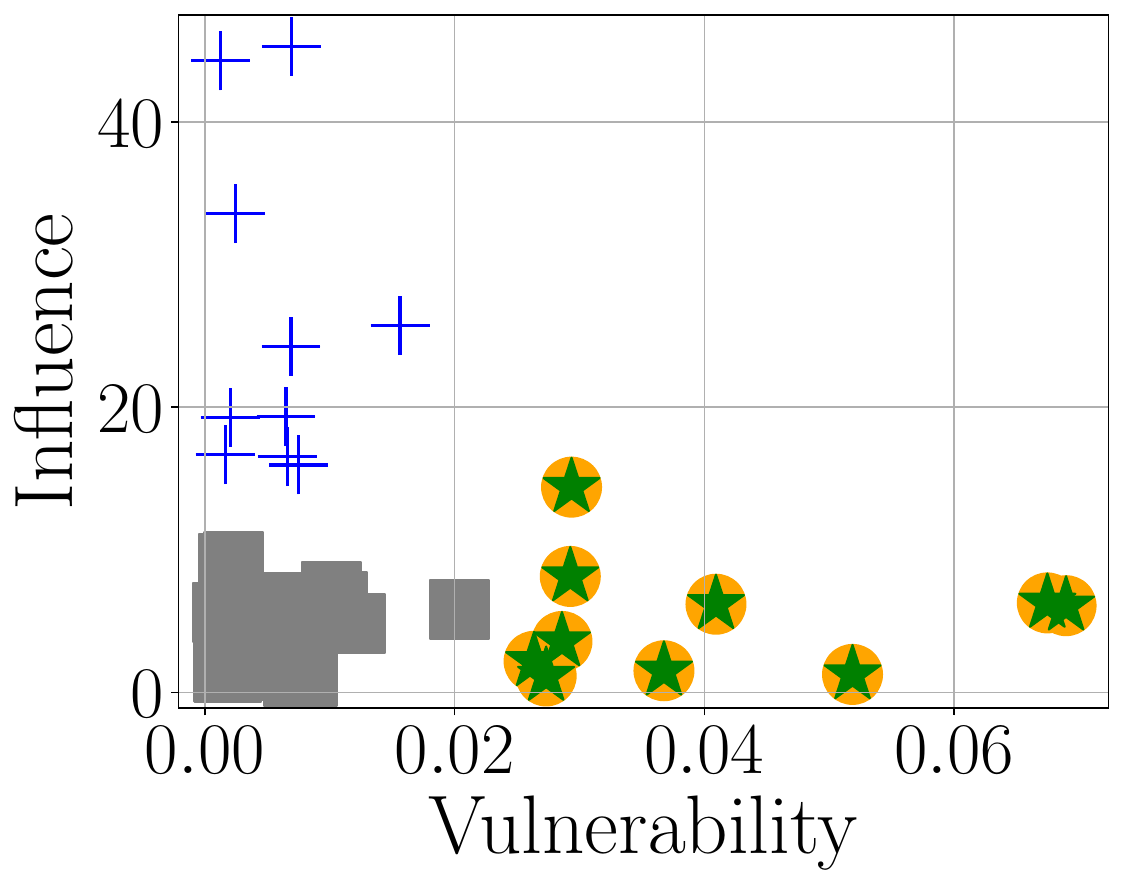}
        \caption{$IC(0.1,4)$}
    \end{subfigure}
    \begin{subfigure}{0.24\columnwidth}
    \includegraphics[width=0.99\linewidth]{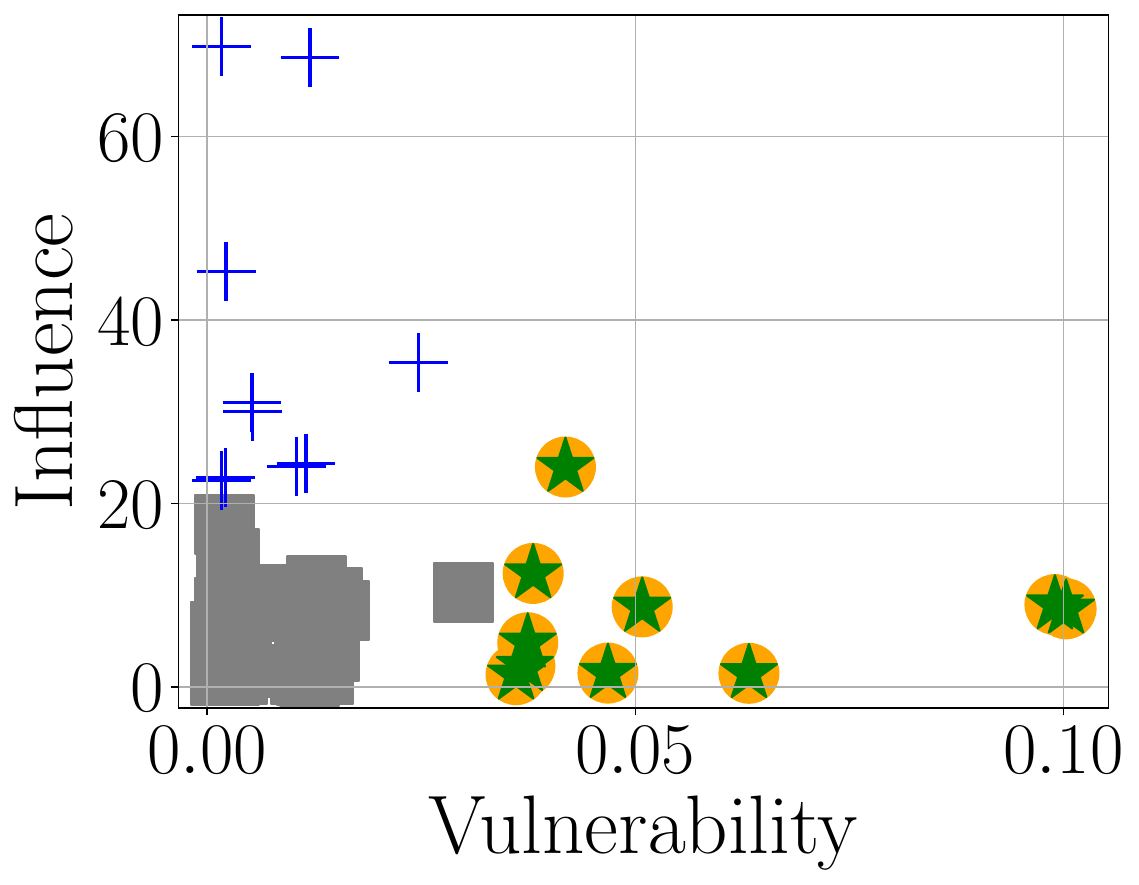}
    \caption{$IC(0.2,2)$}
    \end{subfigure}
    \begin{subfigure}{0.24\columnwidth}
    \includegraphics[width=0.99\linewidth]{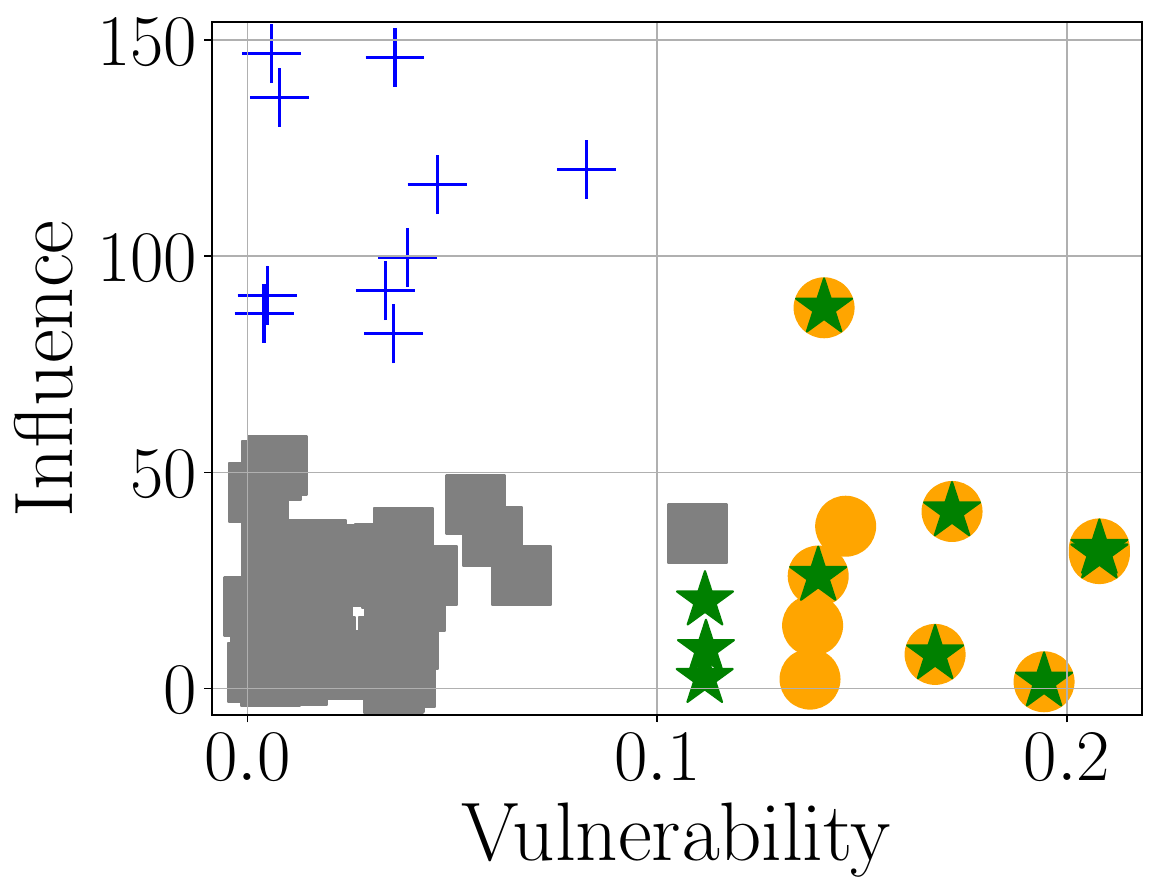}
    \caption{$IC(0.2,4)$}
    \end{subfigure}
     \includegraphics[width=0.5\linewidth]{influence_legend.pdf}
    \caption{Vulnerability vs Influence of the surveillance sets by \tool{} vs baselines  under \rsource{}~seeding on \pl{}}
    \label{fig:infl_addn_random_pl}
\end{figure}

\begin{figure}
    \centering
    \begin{subfigure}{0.28\columnwidth}
    \includegraphics[width=0.99\linewidth]{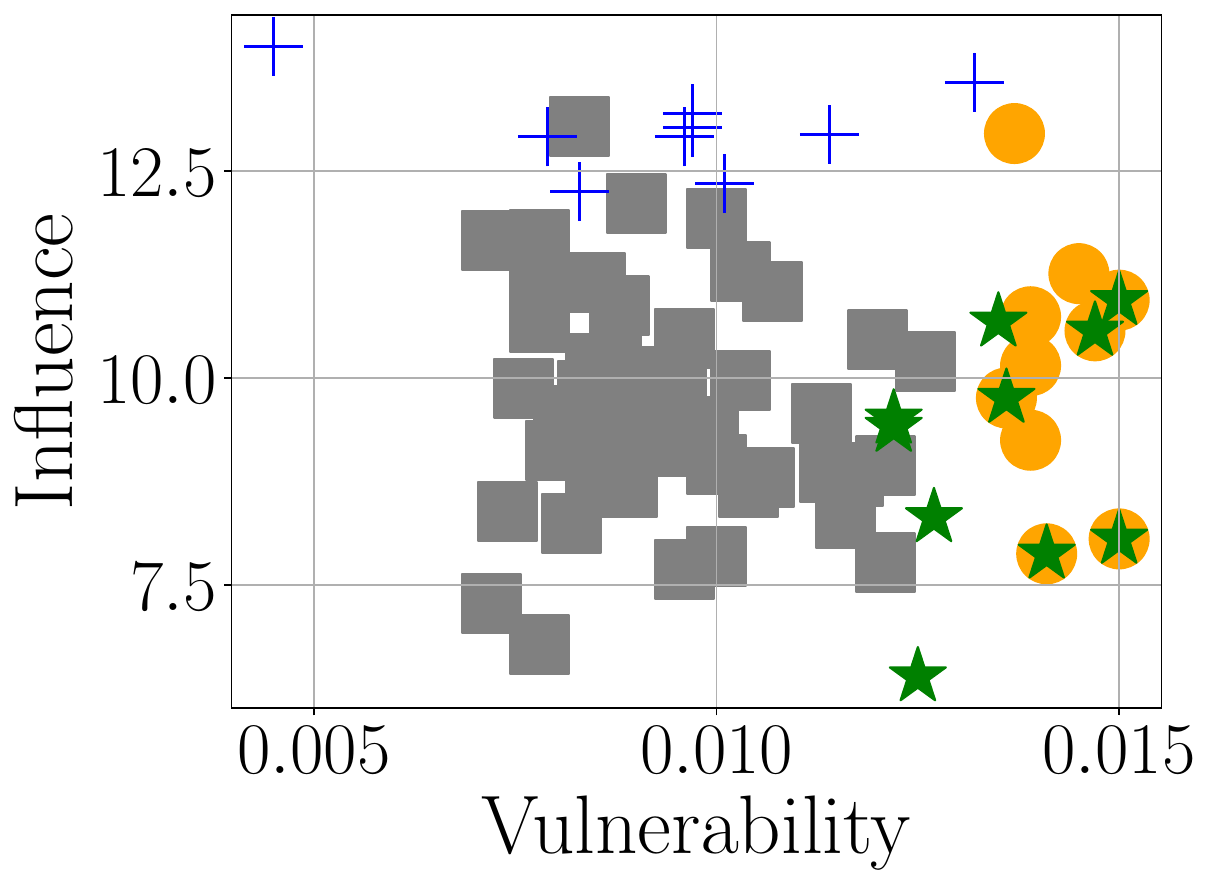}
    \caption{$IC(0.05,2)$}
    \end{subfigure}
    \begin{subfigure}{0.28\columnwidth}
     \includegraphics[width=0.99\linewidth]{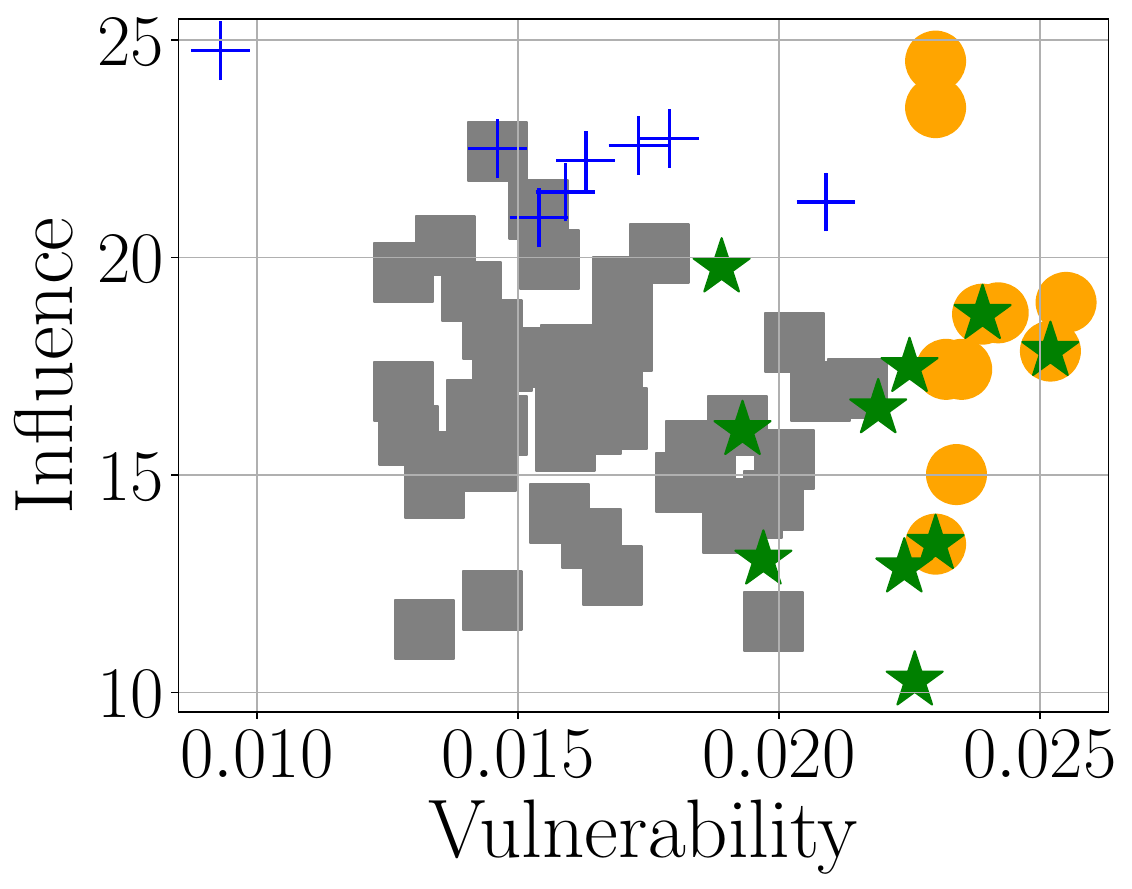}
     \caption{$IC(0.07,2)$}
     \end{subfigure}
     \includegraphics[width=0.5\linewidth]{influence_legend.pdf}
    \caption{Vulnerability vs Influence of the surveillance sets by \tool{} vs baselines  under \rsource{}~seeding on \er{}}
    \label{fig:infl_add_random_er}
\end{figure}

\begin{figure}
    \centering
    \begin{subfigure}{0.24\columnwidth}
    \includegraphics[width=0.99\linewidth]{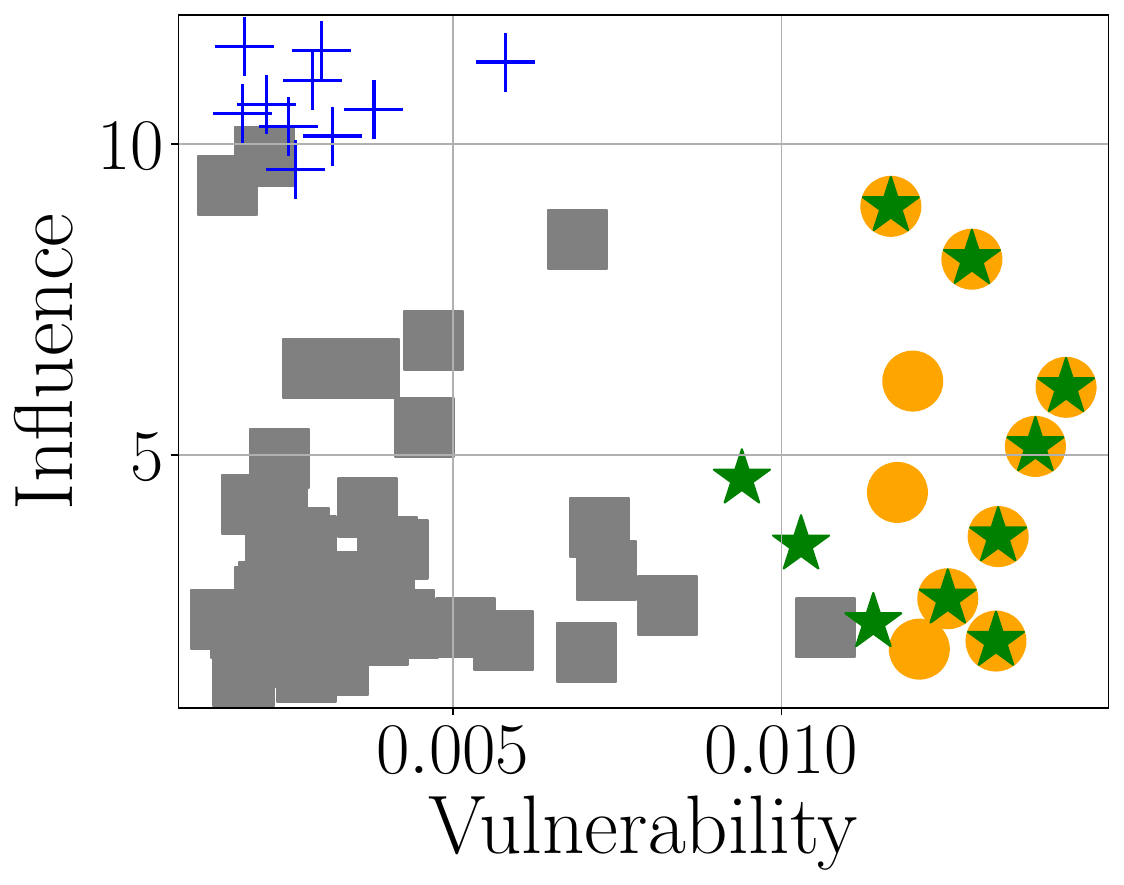}
    \caption{$IC(0.1,2)$}
    \end{subfigure}
    \begin{subfigure}{0.24\columnwidth}
    \includegraphics[width=0.99\linewidth]{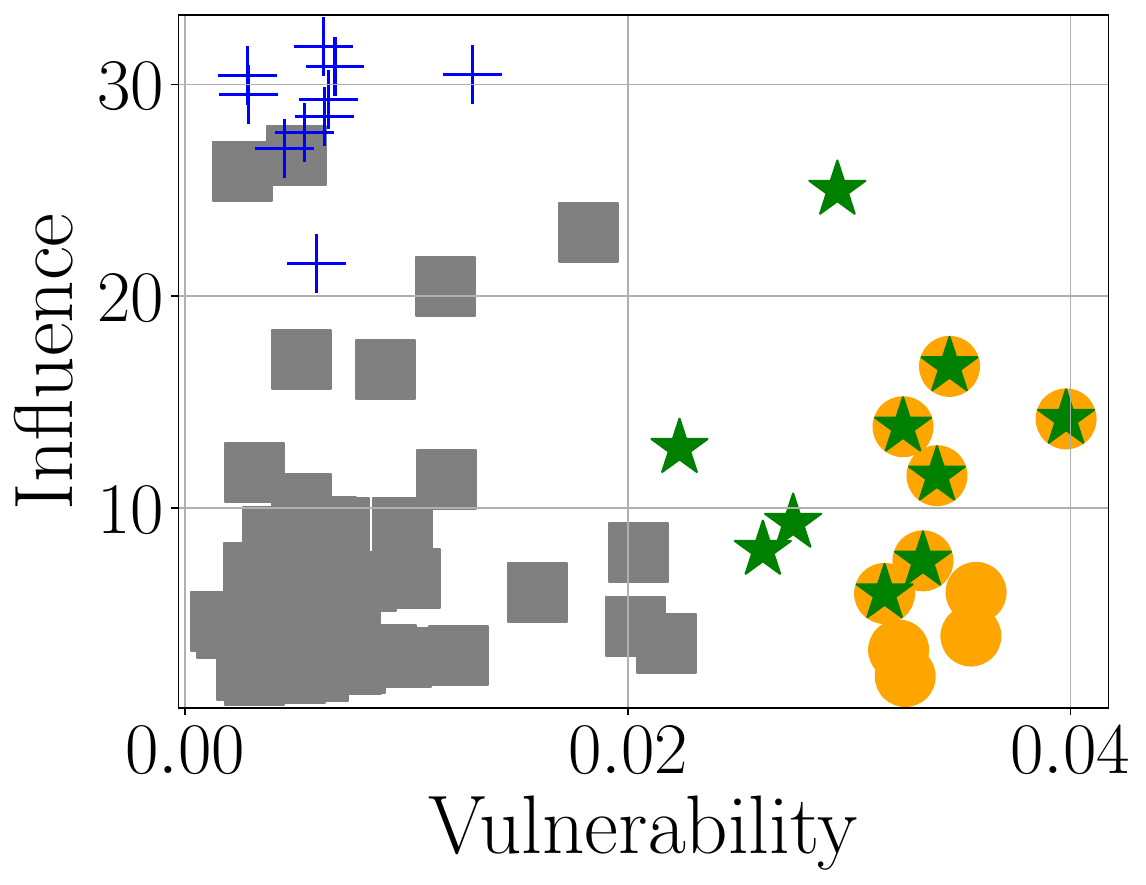}
    \caption{$IC(0.1,4)$}
    \end{subfigure}
    \begin{subfigure}{0.24\columnwidth}
    \includegraphics[width=0.99\linewidth]{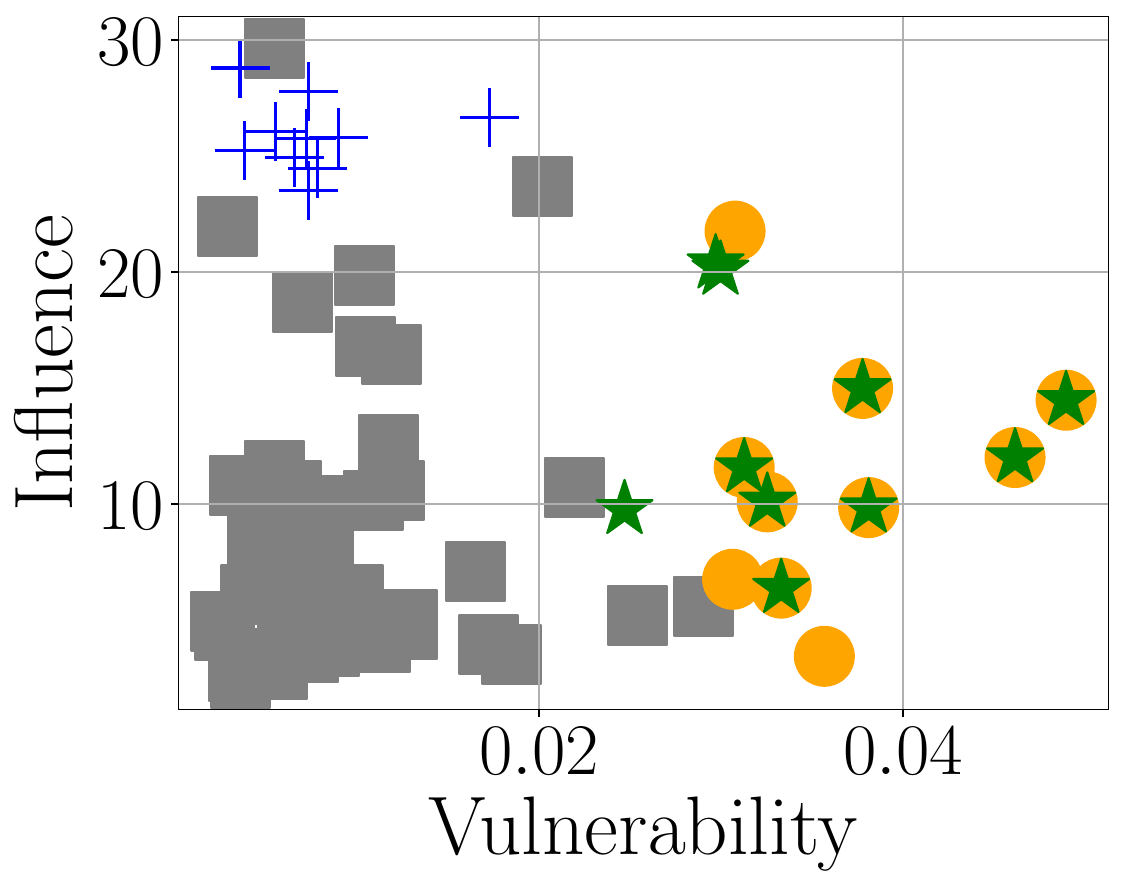}
    \caption{$IC(0.2,2)$}
    \end{subfigure}
    \begin{subfigure}{0.24\columnwidth}
    \includegraphics[width=0.99\linewidth]{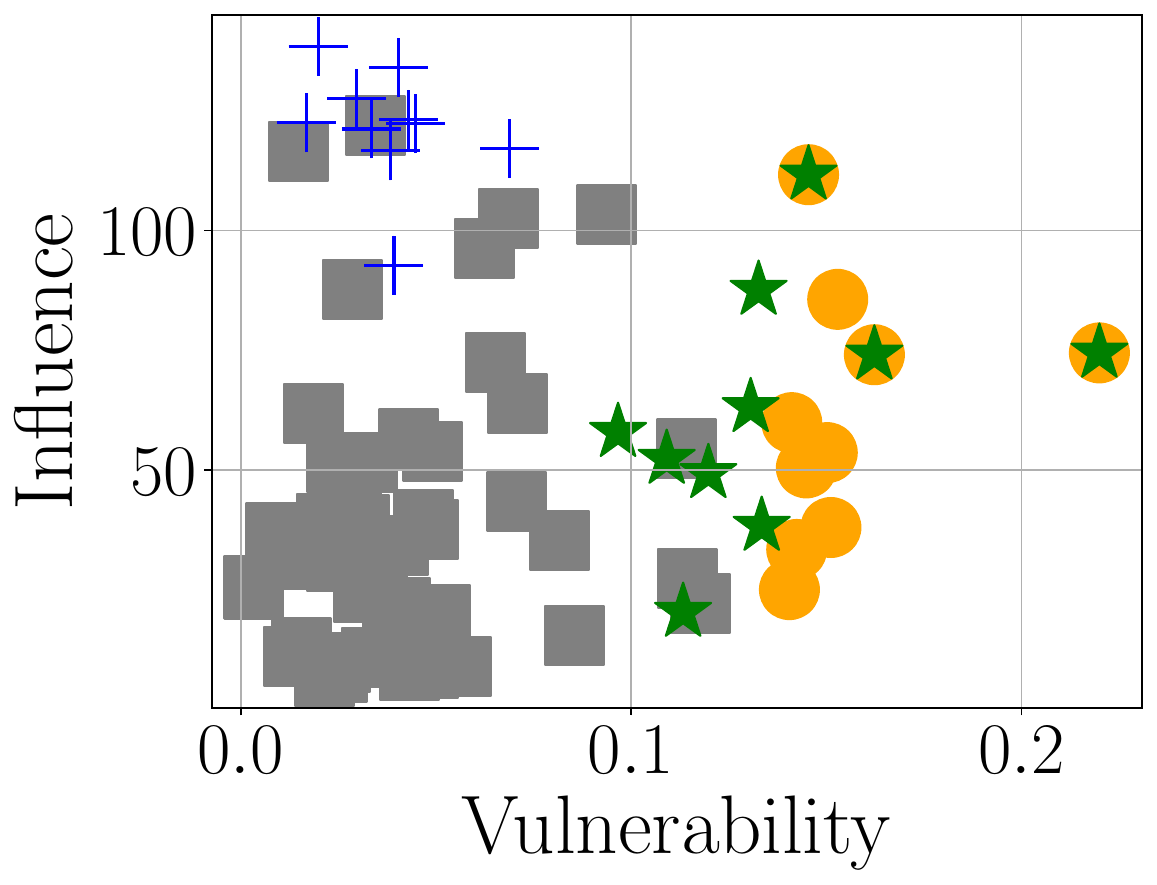}
    \caption{$IC(0.2,4)$}
    \end{subfigure}
    \includegraphics[width=0.5\linewidth]{influence_legend.pdf}
    \caption{Vulnerability vs Influence of the surveillance sets by \tool{} vs baselines  under \rsource{}~seeding on \icu{}}
    \label{fig:infl_addn_random_icu}
\end{figure}

\end{document}